\documentclass[acmlarge, 12pt]{acmart}
\usepackage{booktabs} 
\usepackage{subcaption} 
\usepackage{color}
\usepackage{amsmath}
\usepackage{amssymb}
\usepackage{graphicx}
\usepackage{dsfont}
\usepackage{amsfonts}
\usepackage{amsthm}
\usepackage{stmaryrd}
\usepackage[all]{xy}
\usepackage{multirow}

\def\>{\ensuremath{\rangle}}
\def\<{\ensuremath{\langle}}

\newcommand {\supp } {{\rm supp}}

\newcommand {\tr} {{\mathit{tr}}}

\newcommand{\cE}{\mathcal{E}}

\newcommand{\cI}{\mathcal{I}}

\newcommand{\hs}{\mathcal{H}}
\newcommand{\ci}{\mathbf{i}}
\newcommand{\sm}[1]{\ensuremath{\llbracket #1\rrbracket}}

\newtheorem{thm}{Theorem}[section]
\newtheorem{cor}{Corollary}[section]
\newtheorem{lem}{Lemma}[section]
\newtheorem{defn}{Definition}[section]
\newtheorem{prop}{Proposition}[section]
\newtheorem{exam}{Example}[section]

\newtheorem{rem}{Remark}[section]
\newtheorem{clm}{Claim}[section]
\newtheorem{fact}{Fact}[section]

\AtBeginDocument{
  \providecommand\BibTeX{{
    \normalfont B\kern-0.5em{\scshape i\kern-0.25em b}\kern-0.8em\TeX}}}

\setcopyright{acmcopyright}
\copyrightyear{2019}
\acmYear{2019}
\acmDOI{10.1145/1122445.1122456}

\begin{document}

             \title{Reasoning about Parallel Quantum Programs}

\author{Mingsheng Ying}
\affiliation{
  \institution{University of Technology Sydney}
  \streetaddress{Broad Way}
  \city{Sydney}
  \country{Australia}}
  \affiliation{
  \institution{Institute of Software, Chinese Academy of Sciences}
  \city{Beijing}
  \country{China}}
   \affiliation{
  \institution{Tsinghua University}
  \city{Beijing}
  \country{China}}
\email{Mingsheng.Ying@uts.edu.au}

\author{Li Zhou}
\affiliation{
  \institution{Max Planck Institute for Security and Privacy}
  \city{Bochum}
  \country{Germany}}
 \email{zhou31416@gmail.com} 
  
\author{Yangjia Li}

\affiliation{
  \institution{Institute of Software, Chinese Academy of Sciences}
  \city{Beijing}
  \country{China}}
\email{yangjia@ios.ac.cn} 

\renewcommand{\shortauthors}{Ying, Zhou and Li.}
 
\thanks{with paper note}

\begin{abstract}
We initiate the study of parallel quantum programming by defining the operational and denotational semantics of parallel quantum programs. The \textit{technical contributions} of this paper include: (1) find a series of useful proof rules for reasoning about correctness of parallel quantum programs; (2) prove a (relative) completeness of our proof rules for partial correctness of disjoint parallel quantum programs; and (3) prove a strong soundness theorem of the proof rules showing that partial correctness is well maintained at each step of transitions in the operational semantics of a general parallel quantum program (with shared variables). This is achieved by partially overcoming the following \textit{conceptual challenges} that are never present in classical parallel programming: (i) the intertwining of nondeterminism caused by quantum measurements and introduced by parallelism; (ii) entanglement between component quantum programs; and (iii) combining quantum predicates in the overlap of state Hilbert spaces of component quantum programs with shared variables. \textit{Applications} of the techniques developed in this paper are illustrated by a formal verification of Bravyi-Gosset-K\"{o}nig's parallel quantum algorithm solving a linear algebra problem, which gives for the first time an unconditional proof of a computational quantum advantage.
\end{abstract}

\ccsdesc[300]{Software and its engineering~Parallel programming languages}
\ccsdesc[300]{Theory of computation~Operational semantics}
\ccsdesc[300]{Theory of computation~Denotational semantics}
\ccsdesc[500]{Theory of computation~Program verification}
 
\keywords{quantum programming, parallel programs, operational semantics, denotational semantics, partial and total correctness, entanglement, interference}

\maketitle

\section{Introduction}\label{Intro}

Quantum programming research started from several high-level quantum programming languages proposed as early as in the later 1990's and early 2000's: QCL by \"{O}mer \cite{Om03}, qGCL by Sanders and Zuliani \cite{SZ00}, QPL by Selinger \cite{Selinger04} and QML by Altenkirch and Grattage \cite{AG05}. Now it has been extensively conducted for two decades; see \cite{Se04, Gay06, Ying16} for a survey. In particular, some more practical and scalable quantum programming languages have been defined and implemented in the last few years, including Quipper \cite{Green14}, Scaffold \cite{Sca12}, QWIRE \cite{Qwire}, and Microsoft's LIQUi$|\rangle$ \cite{WS14} and Q\# \cite{Svor18}. Various semantics and type theories of quantum programming languages have been extensively studied; for example, a denotational semantics of quantum lambda calculus with recursion was discovered by Hasuo and Hoshino \cite{Hasuo} and Pagani et al. \cite{Pagani}, an algebraic theory for equational reasoning about quantum programs was developed by Staton \cite{Staton}, and type systems have been established for quantum lambda-calculus \cite{SV09} and QWIRE \cite{Qwire}.

{\vskip 4pt}

\textbf{Quantum Hoare Logic}: Several verification techniques for classical programs have also been extended to quantum programs \cite{Baltag06, BJ04, CMS, FDJY07, Gay08, Kaku09, Rand17}. In particular, the notion of weakest precondition for a quantum program as a physical observable (or mathematically a Hermitian operator) was introduced by D'Hondt and Panangaden in \cite{DP06}, and then a Hoare-like logic for both partial and total correctness of quantum programs with (relative) completeness was built in \cite{Ying11}. In the last few year, some significant progress has been made in further developing quantum Hoare logic and related issues. An SDP (Semi-Definite Programming) algorithm for generating invariants and an SDP algorithm for termination analysis of quantum programs with ranking functions (or super-martingales) were presented in \cite{YYW17, LY18}. A theorem prover for quantum Hoare logic was implemented based on Isabelle/HOL in \cite{Liu19}. Ghost (i.e. auxiliary) variables in quantum Hoare logic were carefully examined in \cite{Unruh19b}. A simplification of quantum Hoare logic for more convenient applications was obtained in \cite{Zhou} by restricting to projective preconditions and postconditions. Quantum Hoare logic was also generalised in \cite{Wu19} for reasoning about robustness of quantum programs against quantum noise during execution. As a generalisation of relational Hoare logic \cite{Benton} and probabilistic relational Hoare logic \cite{Barthe}, a quantum relational Hoare logic with subspaces of (equivalently, projection operators on) the state Hilbert space as preconditions and postconditions was first proposed in \cite{Unruh19a}, targetting applications in security verification of quantum cryptographic protocols. It was further extended in \cite{Barthe19, Li19} to the general case where any Hermitian operators can be used as preconditions and postconditions. 

{\vskip 4pt}

\textbf{Why Parallel Quantum Programming?} The works mentioned above concentrate on sequential quantum programming. However, parallel programming problem for quantum computing has already arisen in the following four areas:
\begin{itemize}\item
Several models of parallel and distributed quantum computing were proposed more than fifteen years ago, mainly with the motivation of using the physical resources of two or more small-capacity quantum computers to realise large-capacity quantum computing, which is out of the reach of current technology; for example, a model of distributed quantum computing over noisy channels was considered in \cite{Cirac}. More recently, a quantum parallel RAM (Random Access Memory) model was defined in \cite{Harrow}, and a formal language for defining quantum circuits in distributed quantum computing was introduced in \cite{YF09}. 

\item Quantum algorithms for solving paradigmatic parallel and distributed computing problems that are faster than the known classical algorithms have been discovered; for example, a quantum algorithm for the leader election problem was given in \cite{leader} and a quantum protocol for the dinning philosopher problem was shown in \cite{dinning}. Also, several parallel implementations of the quantum Fourier transform and Shor's quantum factoring algorithm were presented in \cite{Cleve00, Moore01}. In particular, Bravyi, Gosset and K\"{o}nig recently discovered a parallel quantum algorithm solving a linear algebra problem called HLF (Hidden Linear Function), which gives for the first time an unconditional proof of a computational quantum advantage \cite{Bravyi} .

\item Parallelism has been carefully considered in the physical level design of quantum computer architecture; see for example \cite{Ion}. Furthermore, the issue of instruction parallelism has already been discussed in Rigetti's quantum instruction set architecture \cite{Rigetti} and IBM Q \cite{IBM}. Moreover, experiments of the physical implementation of parallel and distributed quantum computing have been frequently reported in the recent years.

\item Motivated by the tremendous progress toward practical quantum hardware in the lastest years, some authors \cite{Boneh17} started to consider how to design an operating system for quantum computers; in particular, what new abstractions could a quantum operating system expose to the programmer? It is well-known that parallelism is a major issue in operating systems for classical computers \cite{Ka15}. As one can imagine, it will also be a major issue in the design and implementation of future quantum operating systems.
\end{itemize}

{\vskip 4pt}

\textbf{Aims of the Paper}:  This paper initiates the study of parallel quantum programming by introducing a programming language that can be used to program parallel and distributed quantum algorithms like those mentioned above. This language is the quantum \textbf{while}-language \cite{Ying11, Ying16} expanded with the construct of parallel composition. We formally define the operational and denotational semantics of parallel composition of quantum programs. The emphasis of this paper is to establish a proof system for reasoning about correctness of parallel quantum programs. We expect that the results obtained in this paper can also be used to model and reason about parallelism in quantum operating systems.

{\vskip 4pt}

\textbf{Owicki-Gries and Lamport Method}: The proof system introduced by Owicki and Gries \cite{Owicki76} and Lamport \cite{Lamport77} is one of the most popular methods for reasoning about classical parallel programs. Roughly speaking, it consists of the Hoare logic for sequential programs, a rule for introducing auxiliary variables recording control flows and a key rule (R.PC) for parallel composition shown in Figure \ref{fig -1}.
\begin{figure}[h]\centering
\begin{align*}({\rm R.PC})\ \ \ \frac{{\rm Proofs\ of}\ \left\{A_i\right\}P_i\left\{B_i\right\}\ (i=1,...,n)\ {\rm are\ interference\ free}}{\left\{\bigwedge_{i=1}^n A_i\right\}P_1\|\cdots\|P_n\left\{\bigwedge_{i=1}^n B_i\right\}}
\end{align*}
\caption{Proof Rule for Parallel Composition.}\label{fig -1}
\end{figure}
The rule (R.PC) degenerates to Hoare's parallel rule introduced in \cite{Hoare72} when components $P_1,...,P_n$ are disjoint; that is, they do not share variables. 

Naturally, a starting point for our research on reasoning about parallel quantum programs is to generalise Hoare's parallel rule and the Owicki-Gries and Lamport method to the quantum setting. However, it is highly nontrivial to develop such a quantum generalisation, especially to find an \textit{appropriate} quantum version of inference rule (R.PC) for parallel composition of programs, and the unique features of quantum systems render us with several challenges in parallel quantum programming that would never be present in parallel programming for classical computers.

{\vskip 4pt}

\textbf{Major Challenges in Parallel Quantum Programming}:  \begin{itemize}\item \textit{Intertwined nondeterminism}: In a quantum \textbf{while}-program, nondeterminism is caused only by the involved quantum measurements, and in a classical parallel program, nondeterminism is introduced only by the parallelism. However, in a parallel quantum program, these two kinds of nondeterminism occur simultaneously, and their intertwining is hard to deal with in defining the denotational semantics of the program; in particular, when it  contains loops which can have infinite computations (see Definition \ref{dp-den-sem} and Example \ref{exam-maximal}).

\item \textit{Entanglement}: The denotational semantics achieved by solving the above challenge provides us with a basis for building an Owicki-Gries and Lamport-like proof system for parallel quantum programs. At the first glance, it seems that disjoint parallel quantum programs are easy to deal with because: (i) interference freedom is automatically there, as what happens in classical disjoint parallel programs; and (ii) conjunctives $\bigwedge_{i=1}^nA_i$ and $\bigwedge_{i=1}^nB_i$ in rule (R.PC) have proper quantum counterparts, namely tensor products $\bigotimes_{i=1}^nA_i$ and $\bigotimes_{i=1}^nB_i$, respectively, when $P_1,...,P_n$ are disjoint. But actually a difficulty that makes no sense in classical computing arises  in reasoning about parallel quantum programs even in this simple case. More explicitly, entanglement is indispensable for realising the advantage of quantum computing over classical computing, but
a quantum generalisation of (R.PC) (more precisely, Hoare's parallel rule) is not strong enough to cope with the situation where entanglement between component programs is present.

\item \textit{Combining predicates in the overlap of state Hilbert spaces}: When we further consider parallel quantum programs with shared variables, another difficulty appears which never happens in classical computation: the Hilbert spaces $\mathcal{H}_{P_i}$ $(i=1,...,n)$ of quantum predicates $A_i$, $B_i$ $(i=1,...,n)$ have overlaps. Then conjunctives $\bigwedge_{i=1}^n A_i$ and $\bigwedge_{i=1}^nB_i$ cannot be simply replaced by tensor products $\bigotimes_{i=1}^nA_i$ and $\bigotimes_{i=1}^nB_i$, respectively, because they are not well-defined in the state Hilbert space $\bigotimes_{i=1}^n\mathcal{H}_{P_i}$ of $P_1\|\cdots\|P_n$.
\end{itemize}

{\vskip 4pt}

\textbf{Technical Contributions of the Paper}: The main technical results are achieved by resolving the first two challenges and partially solving the third challenge discussed above.  
\begin{itemize}\item The challenge of \textit{intertwined nondeterminism} is settled in Section \ref{sec-disjoint} by establishing a subtle confluence property for different execution paths of the parallel quantum program (see Lemmas \ref{lem-det} and \ref{lem-seq} and their proofs in Appendices C and D).

\item We propose two techniques to tame the difficulty of \textit{entanglement}: (a) introducing an additional inference rule obtained by invoking a deep theorem about the relation between noise and entanglement from quantum physics \cite{Gur03} (see rule (R.S2E) in Figure \ref{fig 6}); and (b) introducing auxiliary variables (see Subsection \ref{sss-aux}) based on the observation in physics that entanglement may emerge when reducing a state of a composite system to its subsystems \cite{NC00}. It turns out that technique (a) can only deal with some special cases of entanglement, but (b) is generic. Using technique (b), we are able to develop a proof system for disjoint parallel quantum programs and establish its (relative) completeness theorem in presence of entanglement (see Theorems \ref{thm-complete-dis} and \ref{complete-dis-p}). 

\item We only have a partial solution to the difficulty of \textit{overlaping state Hilbert spaces}. The idea is that probabilistic (convex) combinations of $A_i$ $(i=1,...,n)$ and $B_i$ $(i=1,...,n)$ 
are well-defined in $\bigotimes_{i=1}^n\mathcal{H}_{P_i}$, even when $P_1,...,P_n$ share variables, 
and can serve as a kind of approximations to the quantum counterparts of conjunctives $\bigwedge_{i=1}^nA_i,\bigwedge_{i=1}^nB_i$, respectively. Although a probabilistic combination is not a perfect quantum version of conjunctive, as a tensor product did in the case of disjoint parallel quantum programs, its reasonableness and usefulness can be clearly seen through its connection to local Hamiltonians in many-body quantum systems (see a detailed discussion in Remark \ref{remark-loc-ham}). Furthermore, we can define a notion of parametrised interference freedom between the proof outlines of component quantum programs. Then a quantum variant of inference rule (R.PC) can be introduced to reason about parallel quantum programs with shared variables.
A strong soundness theorem is proved for the rules showing that partial correctness is well maintained at each step of the transitions in the operational semantics of a parallel quantum program with shared variables (see Theorem \ref{convex-sound}), which can be seen as a quantum generalisation of Lemma 8.8 in \cite{Apt09} or the strong soundness theorem in Section 7.4 of \cite{Francez}.
\end{itemize}

{\vskip 4pt}

\textbf{Organisation of the Paper}: For convenience of the reader, we briefly review quantum Hoare logic in Section \ref{qwhile}. Our study of parallel quantum programming starts in Section \ref{sec-disjoint} where we define the operational and denotational semantics  of disjoint parallel quantum programs. In Section \ref{dis-correctness}, we develop a proof system for reasoning about disjoint parallel quantum programs, including a quantum generalisation of rule (R.PC). In particular, we prove its (relative) completeness for both partial and total correctness in Subsection \ref{subsec-entangle}. The syntax and semantics of parallel quantum programs with shared variables are defined in Section \ref{sec-shared-varaiables}. Section \ref{correct-shared} is devoted to develop proof techniques for parallel quantum programs with shared variables. The notion of proof outline is required to present inference rule (R.PC) for classical parallel programs with shared variables. A corresponding notion is needed to present the quantum generalisation(s) of rule (R.PC). As a preparation, such a notion is introduced for quantum \textbf{while}-programs in Subsection \ref{sec-outline}. Then we use it to introduce the notion of parameterised noninterference and present an inference rule for a parallel quantum program with its precondition (resp. postcondition) as a probabilistic combination of the preconditions (resp. postconditions) of its component programs. Several simple examples are given along the way to illustrate the notions and proof rules introduced in these sections and especially to show the subtle difference between the classical and quantum cases. A detailed case study is presented in Section \ref{sec-case} where a formal verification of Bravyi-Gosset-K\"{o}nig's parallel quantum algorithm solving a linear algebra problem, which gives for the  first time an unconditional proof of a computational quantum advantage. Section \ref{sec-con} is the concluding section where several unsolved problems are pointed out and their difficulties are briefly discussed. For readability, all lengthy proofs are postponed into the Appendices. 

\section{Hoare Logic for Quantum \textbf{While}-Programs}\label{qwhile}

The parallel quantum programs considered in this paper are parallel compositions of quantum \textbf{while}-programs studied in \cite{Ying11, Ying16}. In this section, we briefly review the syntax and semantics of quantum \textbf{while}-language and quantum Hoare logic from \cite{Ying11, Ying16}. They will serve as a basis of the subsequent sections.

\subsection{Syntax and Semantics of Quantum \textbf{while}-Programs}

We assume a countably infinite set $\mathit{Var}$ of quantum variables. For each $q\in\mathit{Var}$, we write $\mathcal{H}_q$ for its state Hilbert space. In this paper, it is always assumed to be finite-dimensional or separable. For any $X\subseteq\mathit{Var}$, we put: $$\mathcal{H}_X=\bigotimes_{q\in X}\mathcal{H}_q.$$

\begin{defn}[Syntax \cite{Ying11, Ying16}]\label{q-syntax}
The quantum \textbf{while}-programs are defined by
the grammar:
\begin{align}\label{syntax}P::=\ \mathbf{skip}\ & |\ P_1;P_2\ |\ q:=|0\rangle\ |\ \overline{q}:=U[\overline{q}]\\ \label{syntax+}&|\ \mathbf{if}\ \left(\square m\cdot M[\overline{q}] =m\rightarrow P_m\right)\ \mathbf{fi}\\ \label{syntax++}&|\ \mathbf{while}\ M[\overline{q}]=1\ \mathbf{do}\ P\ \mathbf{od}\end{align}\end{defn}

Here, $q:=|0\rangle$ means that quantum variable $q$ is initialised in a basis state $|0\rangle$. $\overline{q}:=U[\overline{q}]$ denotes that unitary transformation $U$ is applied to quantum register $\overline{q}$, which is a sequence of quantum variables. In the case statement $\mathbf{if}\cdots\mathbf{fi}$, quantum measurement $M$ is performed on the register $\overline{q}$ and then a subprogram $P_m$ is selected for next execution according to the measurement outcome $m$. In the loop $\mathbf{while}\cdots\mathbf{od}$, measurement $M$ in the loop guard has only two possible outcomes $0,1$; if the outcome is $0$ the loop terminates, and if the outcome is $1$ the program executes the loop body $P$ and enters the loop again.

For each quantum program $P$, we write $\mathit{var}(P)$ for the set of quantum variables occurring in $P$. Let $\mathcal{H}_P=\mathcal{H}_{\mathit{var}(P)}$ be the state Hilbert space of $P$. We write $\mathcal{D}(\mathcal{H}_P)$ for the set of partial density operators, i.e. positive operators with traces $\leq 1$, in $\mathcal{H}_P$. A configuration is a pair $C=\langle P,\rho\rangle,$
where $P$ is a program or the termination symbol $\downarrow$, and $\rho\in\mathcal{D}(\mathcal{H}_P)$ denotes the state of quantum variables.

\begin{defn}[Operational Semantics \cite{Ying11, Ying16}]\label{def-op-sem} The operational semantics of quantum \textbf{while}-programs is defined as a transition relation $\rightarrow$ by the transition rules in Figure \ref{fig 3.1}. \begin{figure}[h]\centering
\begin{equation*}\begin{split}&({\rm Sk})\ \ \langle\mathbf{skip},\rho\rangle\rightarrow\langle \downarrow,\rho\rangle\ \ \ \ \ \ \ \ \ \ \ \ \ \ \ \ \ \ \ \ \ \ \ \ \ \ \ \ \ \ \ \ \ \ ({\rm In})\ \ \ \langle
q:=|0\rangle,\rho\rangle\rightarrow\langle \downarrow,\rho^{q}_0\rangle\\
&({\rm UT})\ \ \langle\overline{q}:=U[\overline{q}],\rho\rangle\rightarrow\langle
\downarrow,U\rho U^{\dag}\rangle\ \ \ \ \ \ \ \ \ \ \ \ \ \ \ \ \ ({\rm SC})\ \ \ \frac{\langle P_1,\rho\rangle\rightarrow\langle
P_1^{\prime},\rho^{\prime}\rangle} {\langle
P_1;P_2,\rho\rangle\rightarrow\langle
P_1^{\prime};P_2,\rho^\prime\rangle}\\
&({\rm IF})\ \ \ \langle\mathbf{if}\ (\square m\cdot
M[\overline{q}]=m\rightarrow P_m)\ \mathbf{fi},\rho\rangle\rightarrow\langle
P_m,M_m\rho M_m^{\dag}\rangle\\
&({\rm L}0)\ \ \ \langle\mathbf{while}\
M[\overline{q}]=1\ \mathbf{do}\
P\ \mathbf{od},\rho\rangle\rightarrow\langle \downarrow, M_0\rho M_0^{\dag}\rangle\\
&({\rm L}1)\ \ \ \langle\mathbf{while}\
M[\overline{q}]=1\ \mathbf{do}\ P\ \mathbf{od},\rho\rangle\rightarrow \langle
P;\mathbf{while}\ M[\overline{q}]=1\ \mathbf{do}\ P\ \mathbf{od}, M_1\rho
M_1^{\dag}\rangle\end{split}\end{equation*}
\caption{Transition Rules for Quantum \textbf{while}-Programs.\ \ \ \ In rule (In), $\rho^{q}_0=\sum_i |0\rangle_q\langle i|\rho |i\rangle_q\langle 0|$ for an orthonormal basis $\{|i\rangle\}$ of $\mathcal{H}_q$; e.g. $\rho^{q}_0=|0\rangle_q\langle 0|\rho|0\rangle_q\langle
0|+|0\rangle_q\langle 1|\rho|1\rangle_q\langle 0|$ if
$\mathit{type}(q)=\mathbf{Bool}$ and $\rho^{q}_0=\sum_{n=-\infty}^{\infty}|0\rangle_q\langle n|\rho|n\rangle_q\langle
0|$ if $\mathit{type}(q)=\mathbf{Int}$ (see \cite{Ying16}, page 63 for the definitions of data types $\mathbf{Bool}$ and $\mathbf{Int}$). 
In (SC), we make the convention $\downarrow;P_2=P_2.$
In (IF), $m$ ranges over every possible outcome of measurement $M=\{M_m\}.$}\label{fig 3.1}
\end{figure}\end{defn}

Note that the transitions in rules (IF), (L0) and (L1) are essentially probabilistic; for example, for each $m$, the transition in (IF) happens with probability $p_m=\mathit{tr}(M^\dag M_m\rho)$, and the program state $\rho$ is changed to $\rho_m=M_m\rho M_m^\dag /p_m$. But following Selinger \cite{Selinger04}, we choose to combine probability $p_m$ and density operator $\rho_m$ into a partial density operator $M_m\rho M_m^\dag=p_m\rho_m$. This convention allows us to present the operational semantics as a non-probabilistic transition system, and it further works for the composition of a sequence of transitions because all transformations in quantum mechanics are linear. Thus, it significantly simplifies the presentation.

\begin{defn}[Denotational Semantics \cite{Ying11, Ying16}]\label{den-sem-def} For any quantum \textbf{while}-program $P$, its semantic function is the mapping $\llbracket P\rrbracket:\mathcal{D}(\mathcal{H}_P)\rightarrow \mathcal{D}(\mathcal{H}_P)$ defined by \begin{equation}\llbracket P\rrbracket(\rho)=\sum\left\{|\rho^\prime: \langle P,\rho\rangle\rightarrow^\ast\langle \downarrow,\rho^\prime\rangle|\right\}\end{equation} for every $\rho\in\mathcal{D}(\mathcal{H}_P)$, where $\rightarrow^\ast$ is the reflexive and transitive closure of transition relation $\rightarrow$ given in Definition \ref{def-op-sem}, and $\left\{|\cdot|\right\}$ denotes a multi-set.
\end{defn}

Intuitively, for an input $\rho$, if for each $k\geq 0$, program $P$ terminates at step $k$ with probability $q_k$ and outputs density operator $\sigma_k$, then with the explanation given in the paragraph before the above definition in mind it is easy to see that $\llbracket P\rrbracket (\rho)=\sum_{k=0}^\infty q_k\sigma_k$.

\subsection{Correctness}

First-order logical formulas are used as the assertions about the properties of classical program states.
The properties of quantum program states are described by quantum predicates introduced by D'Hondt and Panangaden in \cite{DP06}. The L\"{o}wner order between operators is defined as follows: $A\sqsubseteq B$ if and only if $B-A$ is positive. Then a quantum predicate in a Hilbert space $\mathcal{H}$ is an observable (a Hermitian operator) $A$ in $\mathcal{H}$ with $0\sqsubseteq A\sqsubseteq I$, where $0$ and $I$ are the zero operator and the identity operator in $\mathcal{H}$, respectively.
Whenever $\mathcal{H}$ is infinite-dimensional, a quantum predicate in it is always required to be a bounded operator.  

\begin{defn}[Correctness Formula, Hoare Triple \cite{DP06, Ying11, Ying16}] A correctness formula (or a Hoare triple) is a statement of the form $\{A\}P\{B\}$, where $P$ is a quantum \textbf{while}-program, and both $A, B$
are quantum predicates in $\mathcal{H}_P$, called the precondition and postcondition, respectively.\end{defn}

\begin{defn}[Partial and Total Correctness \cite{Ying11, Ying16}]\label{correctness-interpretation}
\begin{enumerate}\item The correctness formula $\{A\}P\{B\}$ is true in
the sense of total correctness, written $$\models_{\mathit{tot}}\{A\}P\{B\},$$ if for all
$\rho\in\mathcal{D}(\mathcal{H}_P)$ we have: $$\tr(A\rho)\leq \tr(B\llbracket P \rrbracket (\rho)).$$

\item The correctness formula $\{A\}P\{B\}$ is true in
the sense of partial correctness, written $$\models_{\mathit{par}}\{A\}P\{B\},$$ if for all
$\rho\in\mathcal{D}(\mathcal{H}_P)$ we have: $$\tr(A\rho)\leq \tr(B\llbracket P \rrbracket (\rho))+
[\tr(\rho)-\tr(\llbracket P \rrbracket (\rho))].$$\end{enumerate}
\end{defn}

The defining inequalities of total and partial correctness can be easily understood by noting that the interpretation of $\mathit{tr}(A\rho)$ in physics is the expectation (i.e. average value) of observable $A$ in state $\rho$, and $\mathit{tr}(\rho)-\mathit{tr}(\llbracket P\rrbracket(\rho))$ is indeed the probability that with input $\rho$ program $P$ does not terminate.

\subsection{Proof System}
A Hoare-like logic for quantum \textbf{while}-programs was established in \cite{Ying11, Ying16}. It includes a proof system qPD for partial correctness and a system qTD for total correctness. The axioms and inference rules of qPD are presented in Figure \ref{fig 3.2}.
\begin{figure}[h]\centering
\begin{equation*}\begin{split}
&({\rm Ax.Sk})\ \ \ \{A\}\mathbf{Skip}\{A\}\ \ \ \ \ \ \ \ \ \ \ \ \ \ \ \ \ \ \ \ \ \ \ \ \ \ ({\rm Ax.In}) \ \ \ \left\{\sum_{i}|i\rangle_q\langle 0|A|0\rangle_q\langle
i|\right\}q:=|0\rangle\{A\}\\
&({\rm Ax.UT})\ \ \
\{U^{\dag}AU\}\overline{q}:=U\left[\overline{q}\right]\{A\}\ \ \ \ \ \ \ \ ({\rm R.SC})\ \ \
\frac{\{A\}P_1\{B\}\ \ \ \ \ \ \{B\}P_2\{C\}}{\{A\}P_1;P_2\{C\}}\\
&({\rm R.IF})\ \ \
\frac{\{A_m\}P_m\{B\}\ {\rm for\ all}\ m}{\left\{\sum_m
M_m^{\dag}A_mM_m\right\}\mathbf{if}\ (\square m\cdot
M[\overline{q}]=m\rightarrow P_m)\ \mathbf{fi}\{B\}}\\
&({\rm R.LP})\ \ \
\frac{\{B\}P\left\{M_0^{\dag}AM_0+M_1^{\dag}BM_1\right\}}{\{M_0^{\dag}AM_0+M_1^{\dag}BM_1\}\mathbf{while}\
M[\overline{q}]=1\ \mathbf{do}\ P\ \mathbf{od}\{A\}}\\
&({\rm R.Or})\ \ \ \frac{A\sqsubseteq
A^{\prime}\ \ \ \ \{A^{\prime}\}P\{B^{\prime}\}\ \ \ \
B^{\prime}\sqsubseteq B}{\{A\}P\{B\}}
\end{split}\end{equation*}
\caption{Proof System qPD for Quantum \textbf{while}-Programs.\ \ \ \ In axiom (Ax.In), $\{|i\rangle\}$ is an orthonormal basis of $\mathcal{H}_q$. In rule (R.Or), $\sqsubseteq$ stands for the L\"{o}wner order.}\label{fig 3.2}
\end{figure}
Similar to the classical case, qTD is obtained from qPD by adding a ranking function into rule (R.LP) to guarantee termination (with probability $1$).

The soundness and (relative) completeness of both qPD and qTD were proved in \cite{Ying11, Ying16}.
\begin{thm}[Soundness and Completeness \cite{Ying11, Ying16}]\label{sound-complete} For any quantum \textbf{while}-program $P$, and for any quantum predicates $A,B$,
\begin{align*}&\models_\mathit{par}\{A\}P\{B\}\Leftrightarrow\ \vdash_\mathit{qPD}\{A\}P\{B\},\ \ \ \ \ \ \ \ \ \ \ \ \models_\mathit{tot}\{A\}P\{B\}\Leftrightarrow\ \vdash_\mathit{qTD}\{A\}P\{B\}.
\end{align*}
\end{thm}

\subsection{Auxiliary Axioms and Rules}\label{sub-auxiliary-rules}

Several auxiliary axioms and rules introduced in \cite{Gor75, Harel79} (see also \cite{Apt09}, Section 3.8) are very useful for simplifying the presentation of correctness proofs of classical programs. 
They are generalised in \cite{Ying18} for quantum \textbf{while}-programs. Here, we recall some of them needed in subsequent sections for our purpose of reasoning about parallel quantum programs. 

Let us first introduce several notations. For any $X\subseteq Y\subseteq\mathit{Var}$ and operator $A$ in $\mathcal{H}_X$, $\mathit{cl}_Y(A)=A\otimes I_{\mathcal{H}_{Y\setminus X}}$ is called the cylindric extension of $A$ in $\mathcal{H}_Y$. If $X,Y\subseteq\mathit{Var}$ and $X\cap Y=\emptyset$. Then the partial trace $\mathit{tr}_Y$ is a mapping from operators in $\mathcal{H}_{X\cup Y}$ to operators in $\mathcal{H}_X$ defined by $\mathit{tr}_Y(|\varphi\rangle\langle\psi|\otimes |\varphi^\prime\rangle\langle\psi^\prime|)=\langle\psi^\prime|\varphi^\prime\rangle\cdot |\varphi\rangle\langle\psi|$ for every $|\varphi\rangle,|\psi\rangle$ in $\mathcal{H}_X$ and $|\varphi^\prime\rangle,|\psi^\prime\rangle$ in $\mathcal{H}_Y$, together with linearity. Let $\{A_n\}$ be a sequence of operators on a Hilbert space $\hs$. We say that $\{A_n\}$ weakly converges to an operator $A$, written 
$A_n\overset{w.o.t.}{\longrightarrow}A,$ if $\lim_{n\rightarrow\infty}\<\psi|A_n|\phi\>\> = \<\psi|A|\phi\>$ for all $|\psi\>,|\phi\>\in\hs$. 
Then we can present the auxiliary axioms and rules  in Figure \ref{fig 3.41}.
\begin{figure}[h]\centering
\begin{align*}
&({\rm Ax.Inv})\ \ \ \left\{A\right\}P\left\{A\right\}\ \ \ \ \ \ \ \ \ \ \ \ \ \ \ \ \ \ \ \ \ \ \ \ \ \ \ \ \ \ \ \ ({\rm R.TI})\ \ \ \ \frac{\{A\otimes I_W\}P\{B\otimes I_W\}}{\left\{A\right\}P\left\{B\right\}}\\
&({\rm R.CC})\ \ \ \frac{\left\{A_i\right\}P\left\{B_i\right\}\ (i=1,...,m)}{\left\{\sum_{i=1}^mp_iA_i\right\}P\left\{\sum_{i=1}^mp_iB_i\right\}}\qquad\ \ ({\rm R.Lin})\ \ \ \frac{\{A\}P\{B\}}{\left\{\lambda A\right\}P\left\{\lambda B\right\}}\\
&({\rm R.Inv})\ \ \ \frac{\{A\}P\{B\}}{\{pA+qC\}P\{pB+qC\}}\ \ \ \ \ \ \ \ \ \ \ \ \ \ \ \ ({\rm R.SO})\ \ \ \ \frac{\{A\}P\{B\}}{\left\{\mathcal{E}^\ast(A)\right\}P\left\{\mathcal{E}^\ast(B)\right\}}\\
&({\rm R.Lim})\ \ \ \frac{A_n\overset{w.o.t.}{\longrightarrow}A\ \ \ \left\{A_n\right\}P\left\{B_n\right\}\ \ \ B_n\overset{w.o.t.}{\longrightarrow}B}{\{A\}P\{B\}}
\end{align*}
\caption{Auxiliary Axioms and Rules for Quantum \textbf{while}-Programs.\ \ \ \ In axiom (Ax.Inv), $\mathit{var}(P)\cap V=\emptyset$ and $A=\mathit{cl}_{V\cup\mathit{var}(P)}(B)$ for some $V\subseteq\mathit{Var}$ and for some quantum predicate $B$ in $\mathcal{H}_V.$ In rule (R.TI), $V,W\subseteq \mathit{Var},$ $V\cap W=\emptyset,$ $A, B$ are quantum predicates in $\mathcal{H}_V,$ $I_W$ is the identity operator on $\mathcal{H}_W$ and $\mathit{var}(P)\subseteq V.$
In  (R.CC), $p_i\geq 0$ $(i=1,...,m)$ and $\sum_{i=1}^m p_j\leq 1.$ In (R.Lin), $0\leq\lambda$ and $\lambda A,\lambda B\sqsubseteq I$. 
In (R.Inv), $p,q\geq 0$, $p+q\leq 1$, and $C$ is a quantum predicate in $\mathcal{H}_V$ for some $V\subseteq\mathit{Var}$ with $V\cap\mathit{var}(P)=\emptyset$.
In (R.SO), $\mathcal{E}$ is a super-operator in $\mathcal{H}_V$ for some $V\subseteq\mathit{Var}$ with $V\cap\mathit{var}(P)=\emptyset$. In (R.Lim), $\{A_n\}$ and $\{B_n\}$ are sequences of quantum predicates.}\label{fig 3.41}
\end{figure}

The following lemma establishes soundness of the auxiliary axioms and rules in Figure \ref{fig 3.41}. 

\begin{lem}[Soundness of Auxiliary Axioms and Rules \cite{Ying18}]\label{Aux-Sound}\begin{enumerate}\item The axiom (Ax.Inv) is sound for partial correctness.
\item The rules (R.TI), (R.CC), (R.Inv) and (R.Lim) are sound both for partial and total correctness. 
\item The rule (R.SO) is sound for total correctness, and it is sound for partial correctness whenever $\mathcal{E}$ is trace-preserving. 
\item The rule (R.Lin) is sound for total correctness, and it is sound for partial correctness whenever $\lambda\leq 1$.  
\end{enumerate}\end{lem}

The auxiliary rules in Figure \ref{fig 3.41} will be combined with a rule for parallel composition in Subsection \ref{dis-completeness} to obtain a (relatively) complete axiomatisation of partial and total correctness of disjoint parallel quantum programs. However, rule (R.CC) is not strong enough in the case of partial correctness. To present a strengthened version of (R.CC), we first introduce:

\begin{defn}\label{def-abort} Let $A$ be a quantum predicate and $P$ a quantum program.
\begin{enumerate}\item We say that $A$ characterises nontermination of quantum program $P$, written $$\models P:{\rm Term}(A),$$ if $\models_\mathit{tot}\{I-A\}P\{I\}$, where $I$ is the identity operator on $\mathcal{H}_P$; that is, for all density operators $\rho$: 
\begin{equation}\label{eq-term}1-\tr(\llbracket P\rrbracket(\rho))\leq\mathit{tr}(A\rho).\end{equation}
\item We say that $A$ characterises abortion of $P$, written $$\models P:{\rm Abort}(A),$$ if $\models_{par}\{A\}P\{0\},$ where $0$ is the zero operator on $\mathcal{H}_P$; that is, hat is, for all density operators $\rho$:
\begin{equation}\label{eq-abort}\tr(A\rho)\leq 1-\mathit{tr}(\llbracket P\rrbracket (\rho)).\end{equation}
\end{enumerate}
\end{defn}

\begin{rem}\label{rem-au-rules}\begin{enumerate}\item Note that $\mathit{tr}(\llbracket P\rrbracket(\rho))$ is the probability that program $P$ with input $\rho$ terminates. Thus, inequality (\ref{eq-term}) shows that its nontermination probability is upper-bounded by predicate $A$.  
On the other hand, the intuition behind inequality (\ref{eq-abort}) is that predicate $A$ implies nontermination.  
\item It is obvious that $\models P:{\rm Term(A)}$ and $\models P:{\rm Abort}(A)$ can be verified in qTD and qPD, respectively.
\end{enumerate}\end{rem}

With the notations introduced in Definition \ref{def-abort}, for partial correctness, rule (R.CC) can be refined into two rules (R.CC1) and (R.CC2) in Figure \ref{fig RCC}.  

\begin{figure}[h]
\begin{align*}
&{\rm(R.CC1)}\quad \frac{\{A_i\}P\{B_i\}\ (i=1,\cdots,m)\qquad \models P:{\rm Abort}(A)
}{\left\{\sum_{i=1}^mp_iA_i + (1-\sum_{i=1}^mp_i)A\right\}P\left\{\sum_{i=1}^mp_iB_i\right\}} \\
&{\rm(R.CC2)}\quad \frac{\{A_i\}P\{B_i\}\ (i=1,\cdots,m)\qquad \models P:{\rm Term}(A)
}{\left\{\sum_{i=1}^m\lambda_iA_i - (\sum_{i=1}^m\lambda_i-1)A\right\}P\left\{\sum_{i=1}^m\lambda_iB_i\right\}} \end{align*}
\caption{Convex Combination Rules for Partial Correctness.\ \ \ \ In rule (R.CC1), $p_i\ge0,\ \sum_{i=1}^mp_i\le1$. In (R.CC2), 
$\lambda_i\ge0,\ \sum_{i=1}^m\lambda_i\ge1$ so that 
the precondition and post condition are quantum predicates.}
\label{fig RCC}
\end{figure}

\begin{lem}\label{SCC} The rules (R.CC1) and (R.CC2) are sound for partial correctness. 
\end{lem}

\begin{proof} See Appendix \ref{proof-SCC}. 
\end{proof}

\section{Syntax and Semantics of Disjoint Parallel Quantum Programs}\label{sec-disjoint}

Now we start to deal with parallel quantum programs. As the first step, let us consider the simplest case, namely disjoint parallel quantum programs, in this and next section. In this section, we define their syntax and operational and denotational semantics. As we saw in Definitions \ref{def-op-sem} and \ref{def-tran-ensemble}, the statistical nature of quantum measurements introduces nondeterminism even in the operational semantics of quantum \textbf{while}-programs. Such nondeterminism is much more complicated in parallel quantum programs; in particular when they contain loops and thus can have infinite computations, because it is intertwined with another kind of nondeterminism, namely nondeterminism introduced in parallelism (see Example \ref{exam-maximal}). But surprisingly, the determinism is still true for the denotational semantics of disjoint parallel quantum programs, and it further entails that disjoint parallel compositions of quantum programs can always be sequentialised.

\subsection{Syntax}\label{disjoint-syntax}

Let us first define the syntax of disjoint parallel quantum programs. 
\begin{defn}[Syntax]\label{syntax-dp} Disjoint parallel quantum programs are generated by the grammar given in equations (\ref{syntax}) and (\ref{syntax+}) together with the following clause: \begin{equation}\label{DP-syntax}P::=P_1\|\cdots\| P_n\equiv \|_{i=1}^n P_i\end{equation} where $n>1$, $P_1,...,P_n$ are quantum \textbf{while}-programs, and $\mathit{var}(P_i)\cap\mathit{var}(P_j)=\emptyset$ for $i\neq j$.
\end{defn}

Program $P$ in equation (\ref{DP-syntax}) is called the disjoint parallel composition of $P_1,...,P_n$. We write: $$\mathit{var}(P)=\bigcup_{i=1}^n\mathit{var}(P_i)$$ for the set of quantum variables in $P$. Thus, the state Hilbert space of $P$ is $$\mathcal{H}_P=\mathcal{H}_{\mathit{var}(P)}=\bigotimes_{i=1}^n\mathcal{H}_{P_i}.$$

\subsection{Operational Semantics}

To accommodate the intertwined nondeterminism introduced by quantum measurements and parallelism together, we have to first recast the operational semantics of quantum \textbf{while}-programs in a slightly different way.  
We define a configuration ensemble as a multi-set $\mathcal{A}=\{|\langle P_i,\rho_i\rangle|\}$ of configurations with $\sum_i\mathit{tr}(\rho_i)\leq 1$.
For simplicity, we identify a singleton $\{|\langle P,\rho\rangle|\}$ with the configuration $\langle P,\rho\rangle$.
Moreover, we need to extend the transition relation between configurations given in Definition \ref{def-op-sem} to a transition relation between configuration ensembles.

\begin{defn}\label{def-tran-ensemble}
The transition relation between configuration ensembles is of the form:
\begin{equation*}\{|\langle P_i,\rho_i\rangle|\}\rightarrow\{|\langle Q_j,\sigma_j\rangle|\}\end{equation*}
and defined by rules (Sk), (In), (UT), (SC) in Figure \ref{fig 3.1} together with the rules presented in Figure \ref{fig ext-3.1}.
\begin{figure}[h]\centering
\begin{equation*}\begin{split}
&({\rm IF'})\ \ \ \langle\mathbf{if}\ (\square m\cdot
M[\overline{q}]=m\rightarrow P_m)\ \mathbf{fi},\rho\rangle\rightarrow\{|\langle
P_m,M_m\rho M_m^{\dag}\rangle|\}\\
&({\rm L}')\ \ \ \langle\mathbf{while}\
M[\overline{q}]=1\ \mathbf{do}\
P\ \mathbf{od},\rho\rangle\rightarrow \{|\langle \downarrow, M_0\rho M_0^{\dag}\rangle, \langle
P;\mathbf{while}\ M[\overline{q}]=1\ \mathbf{do}\ P\ \mathbf{od}, M_1\rho
M_1^{\dag}\rangle|\}\\
&({\rm MS1})\ \ \ \frac{C\rightarrow\mathcal{A}}{\{C\}\rightarrow \mathcal{A}}\ \ \ \ \ \ \ \ \ \ \ \ \ \ \ \ \ \ \ \ \ \ \ \ ({\rm MS2})\ \ \ \frac{\begin{array}{cc}\{\mathcal{A}_i\}_{i\in I}\ {\rm is\ a\ partition\ of}\ \mathcal{A}& I=I_0\cup I_1\\ \mathcal{A}_i\not\rightarrow\ {\rm for\ every}\ i\in I_0& \mathcal{A}_i\rightarrow\mathcal{B}_i\ {\rm for\ every}\ i\in I_1\end{array}}{\mathcal{A}\rightarrow \left(\bigcup_{i\in I_0}\mathcal{A}_i\right)\cup\left(\bigcup_{i\in I_1}\mathcal{B}_i\right)}
\end{split}\end{equation*}
\caption{Extended Transition Rules for Quantum \textbf{while}-Programs.\ \ \ \ In rule (MS1), $C$ is a configuration and $\mathcal{A}$ is a configuration ensembles. In rule (MS2), $\mathcal{A}_i$ and $\mathcal{B}_i$ are all configuration ensembles. Note that in (MS2), $\bigcup$ stands for union of multi-sets.}\label{fig ext-3.1}
\end{figure}
\end{defn}

We observe that for each possible measurement outcome $m$, transition rule (IF) in Figure \ref{fig 3.1} gives a transition from configuration $\langle\mathbf{if}\cdots\mathbf{fi},\rho\rangle$. Transition rule (IF') in Figure \ref{fig ext-3.1} is essentially a merge of these transitions by collecting all the target configurations into a configuration ensemble. Similarly, transition rule (L') is a merge of (L0) and (L1) in Figure \ref{fig 3.1}. Transition rule (MS1) is introduced simply for lifting transitions of configurations to transitions of configuration ensembles. Rule (MS2) allows us to combine several transitions from some small ensembles into a single transition from a large ensemble.

With the above preparation, we can define the operational semantics of disjoint parallel quantum programs in a simple way. 

\begin{defn}[Operational Semantics]\label{disjoint-op-semantics} The operational semantics of disjoint parallel quantum program is the transition relation between configuration ensembles defined by the rules used in Definition \ref{def-tran-ensemble} together with rule (PC) in Figure \ref{fig 2}.
\begin{figure}[h]\centering
\begin{equation*}({\rm PC})\ \ \ \ \frac{\langle P_i,\rho\rangle\rightarrow\{|\langle
P_{ij}^{\prime},\rho_j^{\prime}\rangle|\}} {\begin{array}{cc}\langle
P_1\|\cdots\|P_{i-1}\| P_i\|P_{i+1}\|\cdots\| P_n,\rho\rangle\rightarrow \{|\langle
P_1\|\cdots\|P_{i-1}\| P_{ij}^{\prime}\|P_{i+1}\|\cdots\| P_n,\rho_j^\prime\rangle|\}\end{array}}\end{equation*}
\caption{Transition Rule for (Disjoint) Parallel Quantum Programs.\ \ \ \ Here, $1\leq i\leq n$.}\label{fig 2}
\end{figure}\end{defn}

Intuitively, transition rule (PC) models interleaving concurrency; more precisely, it means that for a fixed $1\leq i\leq n$, the $i$th component $P_i$ of parallel quantum programs $P\equiv P_1\|\cdots\|P_n$ performs a transition, then $P$ can perform the same transition. We will use the convention that $P_1\|\cdots\| P_n=\ \downarrow$ when $P_i=\ \downarrow$ for all $i$. 

To further illustrate the transition rule (PC), we consider the following simple example . In this paper, to simplify the presentation, for a pure state $|\varphi\rangle$ and a complex number $\alpha$ with $|\alpha|\leq 1$, we often use the vector $\alpha |\varphi\rangle$ to denote the corresponding partial density operator $|\alpha|^2|\varphi\rangle\langle\varphi|$.
\begin{exam}\label{comp-value-easy}
Let $p, q, r$ be three qubit variables, \begin{align*}P_1\equiv\ p:=X[p];q:=Z[q],\ \ \ \ \ \ \ \ \ \ P_2\equiv\ &\mathbf{if}\ M[r]=0\rightarrow\mathbf{skip}\\ &\square \ \ \ \ \ \ \ \ \ \ \ \ \ \ 1\rightarrow r:=H[r]\\ &\mathbf{fi} \end{align*}
where $X, Z$ are the Pauli gates, $H$ the Hadamard gate and $M=\{M_0=|0\rangle\langle 0|,M_1=|1\rangle\langle 1|\}$ is the measurement in the computational basis, and let $|\psi\rangle=\frac{1}{\sqrt{2}}(|000\rangle+|111\rangle)$ be the GHZ (Greenberger-Horne-Zeilinger) state.
Then \begin{align*}
\langle P_1\|P_2,|\psi\rangle\rangle &\rightarrow_1\langle q:=Z[q]\|P_2,\frac{1}{\sqrt{2}}(|100\rangle+|011\rangle)\rangle \rightarrow_2\begin{cases}\langle q:=Z[q]\|\mathbf{skip},\frac{1}{\sqrt{2}}|100\rangle\rangle\\ \langle q:=Z[q]\|r:=H[r],\frac{1}{\sqrt{2}}|011\rangle\rangle\end{cases}
\\
&\rightarrow_1\begin{cases}\langle \downarrow\|\mathbf{skip},\frac{1}{\sqrt{2}}|100\rangle\rangle\\ \langle q:=Z[q]\|r:=H[r],\frac{1}{\sqrt{2}}|011\rangle\rangle\end{cases}\rightarrow_2\begin{cases}\langle \downarrow\|\mathbf{skip},\frac{1}{\sqrt{2}}|100\rangle\rangle\\ \langle q:=Z[q]\|\downarrow,\frac{1}{\sqrt{2}}|01-\rangle\rangle\end{cases}\\
&\rightarrow_1\begin{cases}\langle \downarrow\|\mathbf{skip},\frac{1}{\sqrt{2}}|100\rangle\rangle\\ \langle \downarrow,-\frac{1}{\sqrt{2}}|01-\rangle\rangle\end{cases}\rightarrow_2\begin{cases}\langle \downarrow,\frac{1}{\sqrt{2}}|100\rangle\rangle\\ \langle\downarrow,-\frac{1}{\sqrt{2}}|01-\rangle\rangle\end{cases}
\end{align*} is a computation of parallel program $P_1\|P_2$ starting in state $|\psi\rangle$.
Here, we use $\rightarrow_i$ to indicate that the transition is made by $P_i$ according to rule (PC), and $|-\rangle=\frac{1}{\sqrt{2}}(|0\rangle-|1\rangle)$.
\end{exam}

It is interesting to see that at the second step of the computation in the above example, measurement $M$ is performed by component $P_2$ and thus certain nondeterminism occurs; that is, two different configurations are produced according to the two different outcomes $0,1$ of $M$. Then in steps 3, 4 and 5, the following kind of interleaving appears: an action of component $P_2$ happens between two actions of component $P_1$ executed on the two different configurations that come from the same measurement $M$. Here, in a sense, nondeterminism caused by quantum measurements is intertwined with nondeterminism introduced by parallelism.
It is worth noting that for a classical parallel program $P\equiv P_1\|\cdots\|P_n$ with $P_i$ $(1\leq i\leq n)$ being \textbf{while}-programs, such an interleaving never happens because nondeterminism does not occur in the execution of any component $P_i$.

\subsection{Denotational Semantics}

In the last section, operational semantics of quantum \textbf{while}-programs was redefined in terms of the transition between configuration ensembles. Accordingly, denotational semantics (i.e. semantic function) of a quantum \textbf{while}-program can be represented using configuration ensembles. For any configuration ensemble $\mathcal{A}$, we define:
$$\mathit{val}(\mathcal{A})=\sum\{|\rho^\prime:\langle \downarrow,\rho^\prime\rangle\in\mathcal{A}|\}.$$ It is evident that if $\mathcal{A}\rightarrow\mathcal{B}$ then $\mathit{val}(\mathcal{A})\sqsubseteq\mathit{val}(\mathcal{B})$ because $\langle\downarrow,\rho\rangle$ has no transition; that is, $\langle\downarrow,\rho\rangle\in\mathcal{A}$ implies $\langle\downarrow,\rho\rangle\in\mathcal{B}$.

\begin{defn}\label{comp-value}\begin{enumerate}\item A computation of a quantum \textbf{while}-program $P$ starting in a state $\rho\in\mathcal{D}\left(\mathcal{H}_P\right)$ is a maximal finite sequence $$\pi=\langle P,\rho\rangle\rightarrow\mathcal{A}_1\rightarrow\cdots\rightarrow\mathcal{A}_n\not\rightarrow$$ or an infinite sequence: $$\pi=\langle P,\rho\rangle\rightarrow\mathcal{A}_1\rightarrow\cdots\rightarrow\mathcal{A}_n\rightarrow\cdots.$$
\item The value of computation $\pi$ is defined as follows:\begin{equation*}\mathit{val}(\pi)=\begin{cases}&\mathit{val}\left(\mathcal{A}_n\right)\ {\rm if}\ \pi\ {\rm is\ finite\ and}\ \mathcal{A}_n\ {\rm is\ the\ last} \ {\rm configuration\ ensemble},\\ &\lim_{n\rightarrow\infty}\mathit{val}\left(\mathcal{A}_n\right)\ {\rm if}\ \pi\ {\rm is\ infinite}.
\end{cases}
\end{equation*}
\end{enumerate}\end{defn}

Note that in the case of infinite $\pi$, sequence $\left\{\mathit{val}\left(\mathcal{A}_n\right)\right\}$ is increasing according to the L\"{o}wner order $\sqsubseteq$. On the other hand, we know that $\mathcal{D}\left(\mathcal{H}_P\right)$ with $\sqsubseteq$ is a CPO (see \cite{Ying16}, Lemma 3.3.2). So, $\lim_{n\rightarrow\infty}\mathit{val}\left(\mathcal{A}_n\right)$ exists.

The following lemma shows determinism of quantum \textbf{while}-programs.  

\begin{lem}\label{while-det}For any quantum \textbf{while}-program $P$ and $\rho\in\mathcal{D}\left(\mathcal{H}_P\right)$, there is exactly one computation $\pi$ of $P$ starting in $\rho$ and $\llbracket P\rrbracket (\rho)= \mathit{val}(\pi).$\end{lem}

\begin{proof}The uniqueness of the computation $\pi=\langle P,\rho\rangle\rightarrow\mathcal{A}_1\rightarrow\cdots\rightarrow\mathcal{A}_n\rightarrow\cdots$ of $P$ starting in $\rho$ follows immediately from Definition \ref{def-tran-ensemble}.
Furthermore, with Definition \ref{den-sem-def} we have: \begin{align*}\llbracket P\rrbracket (\rho)&=\sum\{|\rho^\prime:\langle P,\rho\rangle\rightarrow^\ast\langle\downarrow,\rho^\prime\rangle|\}=\lim_{n\rightarrow\infty}\sum\{|\rho^\prime:\langle P,\rho\rangle\rightarrow^n\langle\downarrow,\rho^\prime\rangle|\}\\ &=\lim_{n\rightarrow\infty}\mathit{val}\left(\mathcal{A}_n\right)=\mathit{val}(\pi).\end{align*} \end{proof}

Now we are ready to introduce the denotational semantics of disjoint parallel quantum programs. But it cannot be defined by simply mimicking Definitions \ref{den-sem-def} and \ref{comp-value}. For each parallel quantum program $P$, we set: 
\begin{equation}\label{value-set}\mathcal{V}(P,\rho)=\{\mathit{val}(\pi): \pi\ {\rm is\ a\ computation\ of}\ P\ {\rm starting\ in}\ \rho\}\end{equation} for any $\rho\in\mathcal{D}\left(\mathcal{H}_P\right)$, where $\mathit{val}(\pi)$ is given as in Definition \ref{comp-value}. Then we have: 

\begin{defn}[Denotational Semantics]\label{dp-den-sem} The semantic function of a disjoint parallel program $P$ is the mapping $\llbracket P\rrbracket:\mathcal{D}(\mathcal{H}_P)\rightarrow2^{\mathcal{D}(\mathcal{H}_P)}$ defined by $$\llbracket P\rrbracket (\rho)=\left\{{\rm maximal\ elements\ of}\ \left(\mathcal{V}(P,\rho),\sqsubseteq\right)\right\}$$ for any $\rho\in\mathcal{D}(\mathcal{H}_P)$. 
\end{defn}

The above definition deserves a careful explanation. First, the reader may be wondering why we need to take maximal elements in the definition of $\llbracket P\rrbracket(\rho)$. For a parallel quantum programs without loop, it is unnecessary to consider maximal elements; for instance, we simply have: $$\llbracket P_1\|P_2\rrbracket(|\psi\rangle)=\left\{\frac{1}{2}(|100\rangle\langle100|+|01-\rangle\langle01-|)\right\}$$ in Example \ref{comp-value-easy}.
However, the following example clearly shows that only maximal elements are appropriate whenever infinite computations occur.

\begin{exam}\label{exam-maximal} Let $q_0,q_1$ be two qubit variables, and for $k=0,1$, \begin{align*}P_k\equiv\ &\mathbf{if}\ M[q_k]=0\rightarrow\mathbf{skip}\\
&\square\ \ \ \ \ \ \ \ \ \ \ \ \ \ \ \ 1\rightarrow\mathbf{skip}\\ &\mathbf{fi};\ \mathbf{while}_k
\end{align*} where $$\mathbf{while}_k\equiv\ \mathbf{while}\ M[q_k]=k\ \mathbf{do}\ \mathbf{skip}\ \mathbf{od},$$ and $M$ is the measurement in the computational basis. Then the following are three computations of parallel program $P_0\|P_1$ starting in state $|++\rangle$ with $|+\rangle=\frac{1}{\sqrt{2}}(|0\rangle+|1\rangle)$:
\begin{enumerate}\item All transitions are performed by $P_0$:
\begin{align*}\pi_0=\langle P_0\|P_1,|++\rangle\rangle\rightarrow_0\mathcal{A}_1\rightarrow_0\mathcal{A}_2\rightarrow_0\mathcal{A}_3\rightarrow_0\cdots \rightarrow_0\mathcal{A}_{2n}\rightarrow_0\mathcal{A}_{2n+1}\rightarrow_0\cdots\end{align*} where \begin{align*}&\mathcal{A}_1=\left\{\langle\mathbf{while}_0\|P_1,\frac{1}{\sqrt{2}}|0+\rangle\rangle,\langle\mathbf{while}_0\|P_1,\frac{1}{\sqrt{2}}|1+\rangle\rangle\right\},\\
&\mathcal{A}_{2n}=\left\{\langle\mathbf{skip};\mathbf{while}_0\|P_1,\frac{1}{\sqrt{2}}|0+\rangle\rangle,\langle\downarrow\|P_1,\frac{1}{\sqrt{2}}|1+\rangle\rangle\right\},\\
&\mathcal{A}_{2n+1}=\left\{\langle\mathbf{while}_0\|P_1,\frac{1}{\sqrt{2}}|0+\rangle\rangle,\langle\downarrow\|P_1,\frac{1}{\sqrt{2}}|1+\rangle\rangle\right\}
\end{align*} for every $n\geq 1$.
\item All transitions are performed by $P_1$:
\begin{align*}\pi_1=\langle P_0\|P_1,|++\rangle\rangle\rightarrow_1\mathcal{B}_1\rightarrow_1\mathcal{B}_2\rightarrow_1\mathcal{B}_3\rightarrow_1\cdots \rightarrow_1\mathcal{B}_{2n}\rightarrow_1\mathcal{B}_{2n+1}\rightarrow_1\cdots\end{align*} where \begin{align*}&\mathcal{B}_1=\left\{\langle P_0\|\mathbf{while}_1,\frac{1}{\sqrt{2}}|+0\rangle\rangle,\langle P_0\|\mathbf{while}_1,\frac{1}{\sqrt{2}}|+1\rangle\rangle\right\},\\
&\mathcal{B}_{2n}=\left\{\langle P_0\|\downarrow,\frac{1}{\sqrt{2}}|+0\rangle\rangle, \langle P_0\|\mathbf{skip};\mathbf{while}_1,\frac{1}{\sqrt{2}}|+1\rangle\rangle\right\},\\
&\mathcal{B}_{2n+1}=\left\{P_0\|\downarrow,\frac{1}{\sqrt{2}}|+0\rangle\rangle, \langle P_0\|\mathbf{while}_1,\frac{1}{\sqrt{2}}|+1\rangle\rangle\right\}
\end{align*} for every $n\geq 1$.
\item The transitions are fairly performed by $P_0$ and $P_1$:
\begin{align*}\pi=\langle P_0\|P_1,|++\rangle\rangle\rightarrow_0\mathcal{A}_1\rightarrow_0\mathcal{A}_2\rightarrow_0\mathcal{A}_3\rightarrow_1\mathcal{C}_{4}\rightarrow_1\mathcal{C}_{5}\rightarrow\cdots\end{align*} where \begin{align*}\mathcal{C}_{4}&=\left\{\langle \mathbf{while}_0\|\mathbf{while}_1,\frac{1}{2}|00\rangle\rangle,\langle \mathbf{while}_0\|\mathbf{while}_1,\frac{1}{2}|01\rangle\rangle, \langle\downarrow\|\mathbf{while}_1,\frac{1}{2}|10\rangle\rangle,\langle\downarrow\|\mathbf{while}_1,\frac{1}{2}|11\rangle\rangle\right\},\\
\mathcal{C}_{5}&=\left\{\langle \mathbf{while}_0\|\downarrow,\frac{1}{2}|00\rangle\rangle,\langle \mathbf{while}_0\|\mathbf{skip};\mathbf{while}_1,\frac{1}{2}|01\rangle\rangle, \langle\downarrow\|\downarrow,\frac{1}{2}|10\rangle\rangle,\langle\downarrow\|\mathbf{skip};\mathbf{while}_1,\frac{1}{2}|11\rangle\rangle\right\}.\end{align*}
\end{enumerate}

Obviously, $\mathit{val}\left(\pi_0\right)=\mathit{val}\left(\pi_1\right)=0<\frac{1}{4}|10\rangle\langle 10|=\mathit{val}(\pi)$, and $\mathit{val}(\pi)$ is a maximal element of $\mathcal{V}\left(P_0\|P_1,|++\rangle\right)$. Furthermore, we have: $\llbracket P_0\|P_1\rrbracket(|++\rangle)=\left\{\frac{1}{4}|10\rangle\langle 10|\right\}.$
\end{exam}

Second, the output $\llbracket P\rrbracket(\rho)$ of a parallel program $P$ with input $\rho$ is defined as the set of maximal elements of a partially ordered set. In general, there may be no or more than one maximal element. But in the case of disjoint parallelism, the structure of $\llbracket P\rrbracket(\rho)$ is simple. As stated at the beginning of this subsection, the denotational semantics of a disjoint parallel quantum program is deterministic although its operational semantics may demonstrate a very complicated nondeterminism; that is, as a generalisation of Lemma \ref{while-det}, we have: 

\begin{lem}[Determinism]\label{lem-det} For any disjoint parallel quantum program $P$ and $\rho\in\mathcal{D}(\mathcal{H}_P)$, $\llbracket P\rrbracket(\rho)$ is a singleton.\end{lem}
\begin{proof} See Appendix \ref{proof-lem-det}.
\end{proof}

For a disjoint parallel quantum program $P$ and for any $\rho\in\mathcal{D}\left(\mathcal{H}_P\right)$, if singleton $\llbracket P\rrbracket(\rho)=\{\rho^\prime\}$, then we will always identify $\llbracket P\rrbracket(\rho)$ with the partial density operator $\rho^\prime$. Indeed, $\rho^\prime$ must be the greatest element of $(\mathcal{V}(P,\rho),\sqsubseteq).$

It is well-known that every disjoint parallel composition of classical \textbf{while}-programs can be sequentialised (see \cite{Apt09}, Lemma 7.7). This result can also be generalised to the quantum case.

\begin{lem}[Sequentialisation]\label{lem-seq} Suppose that quantum \textbf{while}-programs $P_1,\cdots,P_n$ are disjoint. Then:
\begin{enumerate}\item For any permutation $i_1,\cdots,i_n$ of $1,\cdots,n$, $\llbracket P_1\|\cdots\|P_n\rrbracket =\llbracket P_{i_1}\|\cdots\| P_{i_n}\rrbracket.$
\item $\llbracket P_1\|\cdots\|P_n\rrbracket = \llbracket P_1;\cdots;P_n\rrbracket.$
\end{enumerate}
\end{lem}

\begin{proof} See Appendix \ref{proof-lem-seq}.
\end{proof}

\section{Proof Rules for Disjoint Parallel Programs}\label{dis-correctness}

In this section, we derive a series of rules for proving correctness of disjoint parallel quantum programs. In classical computing, the behaviour of a disjoint parallel program is relatively simple due to noninterference between its components; in particular, only a simplified version of rule (R.PC) in Figure \ref{fig -1} (without noninterference condition) is needed for reasoning about them (see \cite{Apt09}, Lemmas 7.6 and 7.7 and Rule 24 on page 255). As we will see shortly, however, one of the three major challenges pointed out in the Introduction - entanglement - already appear in verification of disjoint parallel quantum programs.

Due to its determinism (Lemma \ref{lem-det}), (partial and total) correctness of a disjoint parallel quantum program $P$ can be defined simply using Definition \ref{correctness-interpretation} provided that for each input $\rho$, we identify the singleton $\llbracket P\rrbracket(\rho)=\{\rho^\prime\}$ with the partial density operator $\rho^\prime$.

Naturally, we first try to find appropriate quantum generalisations of the inference rules for classical disjoint parallel programs. But at the end of this subsection, we will see that some novel rules that have no classical counterpart are needed to cope with entanglement.

\subsection{Sequentialisation Rule}

To warm up, let us first consider a simple inference rule. As mentioned in the previous section, all disjoint parallel programs in classical computing can be sequentialised with the same denotational semantics. Accordingly, they can be verified through sequentialisation (\cite{Apt09}, Section 7.3). For quantum computing, the following sequentialisation rule is valid too:

\begin{figure}[h]\centering
\begin{align*}&({\rm R.Seq})\ \ \ \ \ \ \ \ \ \ \ \ \ \ \ \ \ \ \ \ \ \ \frac{\left\{A\right\}P_1;\cdots; P_n\left\{B\right\}}{\left\{ A\right\}P_1\|\cdots\|P_n\left\{B\right\}}\ \ \ \ \ \ \ \ \ \ \ \ \ \ \ \ \ \ \ \ \ \ \ \ \ \ \ \ \ \ \ \ \
\end{align*}
\caption{Sequentialisation Rule for Disjoint Parallel Programs.}\label{fig 3-1}
\end{figure}

\begin{lem} The rule (R.Seq) is sound for both partial and total correctness.
\end{lem}

\begin{proof}Immediate from Lemma \ref{lem-seq}(2).\end{proof}

Let us give a simple example to show how rule (R.Seq) can be applied to verify disjoint parallel quantum programs. Our example is a quantum analog of the following simple example given in \cite{Apt09} to show the necessity of introducing auxiliary variables: $$\{x=y\}x:=x+1\|y:=y+1\{x=y\}.$$ This correctness formula for a disjoint parallel program cannot be proved by merely using the parallel composition rule (R.PC) in Fig. \ref{fig -1}. However, it can be simply derived by rule (R.Seq). Similarly, we have:

\begin{exam}\label{equal-pred} Let $p,q$ be two quantum variables with the same state Hilbert space $\mathcal{H}$. For each orthonormal basis $\Phi=\{|\varphi_i\rangle\}$ of $\mathcal{H}$, we define a quantum predicate:
\begin{equation}\label{entangle-0}A_\Phi=\sum_i \mu_i|\varphi_i\varphi_i\rangle\langle\varphi_i\varphi_i|\end{equation} in $\mathcal{H}\otimes\mathcal{H}$, where $\mu_i>0$ for every $i$. It can be viewed as a quantum counterpart of equality $x=y$. It is interesting to note that the quantum counterpart of $x=y$ is not unique because for different bases $\Phi=\{|\varphi_i\rangle\}$, $A_\Phi$ are different.
For any unitary operator $U$ in $\mathcal{H}$, we have: \begin{equation}\label{qdis-seq}\models_\mathit{tot}\{A_\Phi\}p:=U[p]\|q:=U[q]\{A_{U(\Phi)}\}\end{equation} where $A_{U(\Phi)}$ is the quantum counterpart of equality defined by orthonormal basis $U(\Phi)=\left\{U|\varphi_i\rangle\right\}$. Clearly, (\ref{qdis-seq}) can be proved using rule (R.Seq) together with (Ax.UT) in Figure \ref{fig 3.2}.
\end{exam}

It is worth pointing out that the quantum generalisation of a concept in a classical system usually has the flexibility arising from different choices of the basis of its state Hilbert space.

\subsection{Tensor product of quantum predicates}

Although rule (R.Seq) in Figure \ref{fig 3-1} can be used to verify a disjoint parallel program $P_1\|\cdots\| P_n$, it does not reflect the essence of (disjoint) parallelism where $P_1,...,P_n$ are independent processes. Moreover, it does not allows us to combine local reasoning about each process $P_i$ to form a global judgement about the parallel program $P_1\|\cdots\| P_n$. 
So, we will not use it in the sequel. Instead, we now start to consider how the crucial rule for reasoning about parallel programs, rule (R.PC) in Figure \ref{fig -1}, can be generalised to the quantum case. To this end, we first need to identify a quantum counterpart of conjunction $\bigwedge_{i=1}^nA_i$ (and $\bigwedge_{i=1}^nB_i$) in rule (R.PC). For disjoint parallel quantum programs, a natural choice is tensor product $\bigotimes_{i=1}^nA_i$ because
it enjoys a nice physical interpretation:
$$\mathit{tr}\left(\left(\bigotimes_{i=1}^n A_i\right)\left(\bigotimes_{i=1}^n\rho_i\right)\right)=\prod_{i=1}^n \mathit{tr}\left(A_i\rho_i\right).$$
The above equation shows that the probability that a product state $\bigotimes_{i=1}^n\rho_i$ satisfies quantum predicate $\bigotimes_{i=1}^nA_i$ is the product of the probabilities that each component state $\rho_i$ satisfies the corresponding predicate $A_i$. This observation motivates an inference rule for tensor product of quantum predicates presented in Figure \ref{fig 3}. It can be seen as the simplest quantum generalisation of rule (R.PC) in Figure \ref{fig -1}.

\begin{figure}[h]\centering
\begin{align*}&({\rm R.PC.P}) \ \ \ \ \ \ \ \ \ \ \ \ \ \ \ \frac{\left\{A_i\right\}P_i\left\{B_i\right\}\ (i=1,...,n)}{\left\{\bigotimes_{i=1}^n A_i\right\}P_1\|\cdots\|P_n\left\{\bigotimes_{i=1}^n B_i\right\}}\ \ \ \ \ \ \ \ \ \ \ \ \ \ \ \ \ \ \ \ \ \
\end{align*}
\caption{Rule for Tensor Product of Quantum Predicates}\label{fig 3}
\end{figure}

\begin{lem}\label{TP-sound} The rule (R.PC.P) is sound with respect to both partial and total correctness.
\end{lem}
\begin{proof} See Appendix \ref{proof-TP-sound}.
\end{proof}

The rule (R.PC.P) can only be used to infer correctness of disjoint parallel quantum programs with respect to (tensor) product predicates. For instance, we can use (R.PC.P) to prove a very special case of (\ref{qdis-seq}) in Example \ref{equal-pred} with $\{p_i\}$ being a degenerate distribution at some $i_0$:
$$\vdash_\mathit{tot}\left\{|\varphi\varphi\rangle\langle\varphi\varphi|\right\}p:=U[p]\| q:=U[q]\left\{|\psi\psi\rangle\langle\psi\psi|)|\right\}$$ where $|\varphi\rangle=|\varphi_{i_0}\rangle$ and $|\psi\rangle=U|\varphi\rangle$,
but it is not strong enough to derive the entire (\ref{qdis-seq}).

\subsection{Separable Quantum Predicates}

A larger family of predicates in $\bigotimes_{i=1}^n\mathcal{H}_{P_i}$ than product predicates is separable predicates defined in the following:
\begin{defn} Let $A$ be a quantum predicate in $\bigotimes_{i=1}^n\mathcal{H}_{P_i}$. Then: \begin{enumerate}\item $A$ is said to be separable if there exist $p_j\geq 0$ and quantum predicates $A_{ji}$ in $\mathcal{H}_{P_i}$ $(i=1,...,n;\ j=1,...,m)$ such that $\sum_{j=1}^mp_j\leq 1$ and $$A=\sum_{j=1}^mp_j\left(\bigotimes_{i=1}^nA_{ji}\right)$$ where $m$ is a positive integer or $\infty$.
\item $A$ is entangled if it is not separable.\end{enumerate}
\end{defn}

A combination of rule (R.PC.P) with the auxiliary axioms and rules (R.CC), (Ax.Inv), (R.Inv) and (R.Lim) in Figure \ref{fig 3.41} yields rule (R.PC.S) in Figure \ref{fig 5}.
\begin{figure}[h]\centering
\begin{equation*}({\rm R.PC.S})\ \ \ \frac{\left\{A_{ji}\right\}P_i\left\{B_{ji}\right\}\ (i=1,...,n;\ j=1,...,m)}{\begin{array}{cc}&\left\{\sum_{j=1}^mp_j\left(\bigotimes_{i=1}^nA_{ji}\right)\right\}P_1\|\cdots\|P_n \left\{\sum_{j=1}^mp_j\left(\bigotimes_{i=1}^nB_{ji}\right)\right\}\end{array}}\end{equation*}
\caption{Rule for Separable Quantum Predicates.\ \ \ \ Coefficients $p_j\geq 0$ and $\sum_{j=1}^mp_j\leq 1$; $m$ is a positive integer or $\infty$.}\label{fig 5}
\end{figure}

Obviously, rule (R.PC.S) can reason about disjoint parallel quantum programs with separable quantum predicates; for example,
 correctness (\ref{qdis-seq}) in Example \ref{equal-pred} can be proved using rule (R.PC.S).

\subsection{Entangled Quantum Predicates}\label{subsec-entangle} It is well-understood that entangled states are indispensable physical resources that make quantum computers outperform classical computers.
Entangled quantum predicates represent quantum non-locality in a dual setup where more information can be revealed by joint (i.e. globally entangled) measurements than can be gained by local operations and classical communications (LOCC) \cite{Peres, Bennett}.  

Obviously, inference rule (R.PC.S) is unable to prove any correctness of the form $\{A\}P_1\|\cdots\|P_n\{B\}$ for a parallel quantum program $P_1\|\cdots\|P_n$ where $A$ or $B$ is an entangled predicate, as shown in the following:

\begin{exam}\label{correct-max-ent-1}We consider a variant of Example \ref{equal-pred}. For each orthonormal basis $\Phi=\{|\varphi_i\rangle\}$ of $\mathcal{H}$, we write:
$$\beta_\Phi=\frac{1}{\sqrt{d}}\sum_i|\varphi_i\varphi_i\rangle$$ for the maximally entangled state in $\mathcal{H}\otimes\mathcal{H}$, where $d=\dim\mathcal{H}$. Then $E_\Phi=|\beta_\Phi\rangle\langle\beta_\Phi|$ can be seen as another quantum counterpart of equality $x=y$ (different from $A_\Phi$ defined by equation (\ref{entangle-0})). Obviously, \begin{equation}\label{correct-max-ent}
\models_\mathit{tot}\left\{|\beta_\Phi\rangle\langle\beta_\Phi|\right\}p:=U[p]\|q:=U[q]\left\{|\beta_{U(\Phi)}\rangle\langle\beta_{U(\Phi)}|\right\};
\end{equation} that is, if the input is maximally entangled, so is the output after the same unitary operator is performed separately on two subsystems. Indeed, we can prove correctness (\ref{correct-max-ent}) by using rules (R.Seq) and (Ax.UT), but (\ref{correct-max-ent}) cannot be derived by directly using rule (R.PC.S).
\end{exam}

\subsection{Transferring Separable Predicates to Entangled Predicates}

Interestingly, a deep result in the theoretical analysis of NMR (Nuclear Magnetic Resonance) quantum computing provides us with a partial solution. It was discovered in \cite{Zy98, Braun99} that all mixed states of $n$ qubits in a sufficiently small neighbourhood of the maximally mixed state are separable. The interpretation of this result in physics is that entanglement cannot exist in the presence of too much noise. The result was generalised in \cite{Gur03} to the case of any quantum systems with finite-dimensional state Hilbert spaces.
Recall that the Hilbert-Schmidt norm (or $2$-norm) of operator $A$ is defined as follows:
$\|A\|_2=\sqrt{\mathit{tr}(A^\dag A)}.$ In particular, if $A=\left(A_{ij}\right)$ is a matrix, then $\|A\|_2=\sqrt{\sum_{i,j}|A_{ij}|^2}.$

\begin{thm}[Gurvits and Barnum \cite{Gur03}]\label{gur-thm} Let $\mathcal{H}_1,\cdots,\mathcal{H}_n$ be finite-dimensional Hilbert spaces, and let $A$ be a positive operator in $\bigotimes_{i=1}^n\mathcal{H}_n$. If $$\|A-I\|_2\leq \frac{1}{2^{n/2-1}}$$ where $I$ is the identity operator in $\bigotimes_{i=1}^n\mathcal{H}_n$, then $A$ is separable.\end{thm}

The following corollary can be easily derived from the above theorem.

\begin{cor}\label{cor-Gur}For any two positive operators $A, B$ in $\bigotimes_{i=1}^n\mathcal{H}_i$, there exists $0<\epsilon\leq 1$ such that both $(1-\epsilon)I+\epsilon A$ and $(1-\epsilon)I+\epsilon B$ are separable.
\end{cor}
\begin{proof}Let $C=(1-\epsilon)I+\epsilon A$ and $D=(1-\epsilon)I+\epsilon B$. Then $\|C-I\|_2=\epsilon\|A-I\|_2$ and $\|D-I\|_2=\epsilon\|B-I\|_2$. So, by Theorem \ref{gur-thm} it suffices to take $$\epsilon\leq\frac{1}{2^{n/2-1}\max [\|A-I\|_2, \|B-I\|_2]}.$$
\end{proof}

Motivated by Corollary \ref{cor-Gur}, we introduce a new inference rule (R.S2E) in Figure \ref{fig 6}.

\begin{figure}[h]\centering
\begin{equation*}({\rm R.S2E})\ \ \ \ \ \ \ \ \ \ \ \ \ \ \ \ \ \ \ \ \ \ \ \ \ \ \ \frac{\left\{(1-\epsilon)I+\epsilon A\right\}P\left\{(1-\epsilon)I+\epsilon B\right\}}{\left\{A\right\}P\left\{B\right\}}\ \ \ \ \ \ \ \ \ \ \ \ \ \ \ \ \ \ \ \ \ \ \ \ \ \ \ \ \ \ \ \ \ \ \ \ \ \ \ \ \ \ \end{equation*}
\caption{Rule for Transforming Separable Predicates to Entangled Predicates.\ \ \ \ Here, $0<\epsilon\leq 1$.}\label{fig 6}
\end{figure}

The idea behind rule (R.S2E) is that in order to prove correctness $\{A\}P_1\|\cdots\|P_n\{B\}$ for entangled predicates $A$ and $B$, we find a parameter $\epsilon>0$ such that $(1-\epsilon)I+\epsilon A$ and $(1-\epsilon)I+\epsilon B$ are separable, and then sometimes we can prove: \begin{equation}\label{mid-entangle}\{(1-\epsilon)I+\epsilon A\}P_1\|\cdots\|P_n\{(1-\epsilon)I+\epsilon B\}\end{equation} by using rule (R.PC.S). It is worth pointing out that Corollary \ref{cor-Gur} warrants that we can choose the same parameter $\epsilon$ in the precondition and postcondition.

\begin{exam}\label{exam:S2E} For $k=0,1$, consider the quantum program $\mathbf{while}_k$ given in Example 4.2. We write: $|\Phi\rangle=\frac{1}{\sqrt{2}}(|01\rangle+|10\rangle)$ for a maximally entangled state of a 2-qubit system. Then it holds that \begin{equation}\models_{par}\left\{I_4-\frac{1}{2}|10\rangle\langle10|\right\}\mathbf{while}_0\parallel\mathbf{while}_1\{|\Phi\rangle\langle\Phi|\},\label{equ:S2E}\end{equation} where $I_4$ is the $4\times 4$ unit matrix. The correctness formula (\ref{equ:S2E}) has entangled precondition and postcondition, and thus
cannot be proved by only using rule (R.PC.S). Here, we show that it can be proved by combining rule (R.S2E) with (R.PC.S). In fact, one can first verify that
\begin{equation}\models_{par}\left\{I_2-|\alpha|^2|1-k\rangle\langle1-k|\right\}\mathbf{while}_k\{|\psi\rangle\langle\psi|\}\label{equ:compo_form}\end{equation}
for $k=0,1$ and any state $|\psi\rangle=\alpha|k\rangle+\beta|1-k\rangle$, where $I_2$ is the $2\times 2$ unit matrix. Moreover, we write:
$$|\curvearrowright\rangle=\frac{1}{\sqrt{2}}(|0\rangle+\mathrm{i}|1\rangle),\ |\curvearrowleft\rangle=\frac{1}{\sqrt{2}}(|0\rangle-\mathrm{i}|1\rangle).$$ Then we have the following decomposition of separable operator:
\begin{align*}\left(1-\frac{2}{3}\right)I_4+\frac{2}{3}|\Phi\rangle\langle\Phi|=\frac{1}{3}(|01\rangle\langle01|&+|10\rangle\langle10|+|++\rangle\langle++|+|--\rangle\langle--|\\ &+|\curvearrowright\curvearrowright\rangle\langle\curvearrowright\curvearrowright|+|\curvearrowleft\curvearrowleft\rangle\langle\curvearrowleft\curvearrowleft|),\end{align*} and it is derived that \begin{equation}\label{equ:S2E1}\left\{I_4-\frac{1}{3}|10\rangle\langle10|\right\}\mathbf{while}_0\parallel\mathbf{while}_1\left\{(1-\frac{2}{3})I_4+\frac{2}{3}|\Phi\rangle\langle\Phi|\right\}\end{equation} by applying (\ref{equ:compo_form}) for $|\psi\rangle=|0\rangle,|1\rangle,|+\rangle,|-\rangle,|\curvearrowright\rangle,|\curvearrowleft\rangle$ and $i=0,1$, respectively, and applying rule (R.PC.S). Finally, correctness (\ref{equ:S2E}) is obtained by applying rule (R.S2E) to (\ref{equ:S2E1}) with $\epsilon=\frac{2}{3}$.
\end{exam}

We conclude this subsection by presenting the soundness of inference rule R.S2E).

\begin{lem}\label{S2E-sound}The rule (R.S2E) is sound for both partial and total correctness.
\end{lem}

\begin{proof} See Appendix \ref{proof-S2E-sound}.\end{proof}

\subsection{Auxiliary Variables}\label{sss-aux}

It was shown in the last subsection that rule (R.S2E) can be used to derive correctness of some parallel programs with entangled preconditions or postconditions. 
But it is obviously not strong enough to deal with all entangled preconditions and postconditions because it is not always possible to find the same probability (sub-)distribution $\{p_j\}$ such that the precondition  and postcondition in (\ref{mid-entangle}) can be written as $\sum_jp_j\left(\bigotimes_{i=1}^n A_{ji}\right)$ and $\sum_jp_j\left(\bigotimes_{i=1}^nB_{ji}\right)$, respectively, but such a match of probabilities in the precondition and postcondition is required in applying rule (R.PC.S). In this subsection, we present another solution to the verification problem for entangled preconditions and postconditions; namely a combination of (R.PC.S) and several rules for introducing auxiliary variables.

It is interesting to note that rule (R.TI) in Figure \ref{fig 3.41} is a quantum generalisation of two rules (DISJUNCTION) and ($\exists$-INTRODUCTION) in Section 3.8 of \cite{Apt09}, where partial trace is considered as a quantum counterpart of logical disjunction and existence quantifier; and (R.SO) in Figure \ref{fig 3.41} is a quantum generalisation of rule (SUBSTITUTION) there, with the substitution $\overline{z}:=\overline{t}$ being replaced by a super-operator $\mathcal{E}$.

Let us start to introduce our method of using auxiliary variables by considering an example. 

\begin{exam}\label{exam-auxi} We use rule (R.PC.S) together with (R.TI) and (R.SO) to prove correctness (\ref{correct-max-ent}) in Example \ref{correct-max-ent-1}.
The key idea is to introduce two auxiliary variables $p^\prime,q^\prime$ with the same state space $\mathcal{H}$. First, by (Ax.UT) we have:
\begin{equation}\label{ent-aux-1} \vdash_\mathit{tot}\left\{\left(E_\Phi\right)_{pp^\prime}\right\}p:=U[p]\left\{|\alpha\rangle_{pp^\prime}\langle\alpha|\right\},\ \ \ \ \ \ \ \ \vdash_\mathit{tot}\left\{\left(E_\Phi\right)_{qq^\prime}\right\}q:=U[q]\left\{|\alpha\rangle_{qq^\prime}\langle\alpha|\right\}\end{equation} where  we use subscripts $p,q,p^\prime,q^\prime$ to indicate the corresponding subsystems, and $|\alpha\rangle=\sum_i\left(U|i\rangle\right)|i\rangle.$ Now applying rule (R.PC.S) to (\ref{ent-aux-1}) yields:
\begin{equation}\label{ent-aux-3}\vdash_\mathit{tot} \left\{\left(E_\Phi\right)_{pp^\prime}\otimes\left(E_\Phi\right)_{qq^\prime}\right\}p:=U[p]\| q:=U[q]\left\{|\alpha\rangle_{pp^\prime}\langle\alpha|\otimes|\alpha\rangle_{qq^\prime}\langle\alpha|\right\}\end{equation}
Finally, we define superoperator: 
$$\mathcal{E}(\rho)=\sum_i\left(|\beta\rangle_{p^\prime q^\prime}\langle i|\right)\rho\left(|i\rangle_{p^\prime q^\prime}\langle\beta|\right)$$ for all mixed states $\rho$ of $p'$ and $q'$, and obtain (\ref{correct-max-ent}) by applying rule (R.SO) to (\ref{ent-aux-3}) because \begin{align*}E_\Phi\otimes I_{p^\prime q^\prime}&=\mathcal{E}^\ast\left(\left(E_\Phi\right)_{pp^\prime}\otimes \left(E_\Phi\right)_{qq^\prime}\right)=\sum_i\left(|i\rangle_{p^\prime q^\prime}\langle\beta|\right)\left(\left(E_\Phi\right)_{pp^\prime}\otimes \left(E_\Phi\right)_{qq^\prime}\right)\left(|\beta\rangle_{p^\prime q^\prime}\langle i|\right),\\
E_{U(\Phi)}\otimes I_{p^\prime q^\prime}&=\mathcal{E}^\ast\left(|\alpha\rangle_{pp^\prime}\langle\alpha|\otimes|\alpha\rangle_{qq^\prime}\langle\alpha|\right)=\sum_i\left(|i\rangle_{p^\prime q^\prime}\langle\beta|\right)\left(|\alpha\rangle_{pp^\prime}\langle\alpha|\otimes|\alpha\rangle_{qq^\prime}\langle\alpha|\right)\left(|\beta\rangle_{p^\prime q^\prime}\langle i|\right).
\end{align*}\end{exam}

\subsection{Completeness Theorems}\label{dis-completeness}

Fortunately the strategy of introducing auxiliary variables used in Example \ref{exam-auxi} can be generalised to deal with all entangled preconditions and postconditions for disjoint parallel quantum programs. More precisely, it provides with us a (relatively) complete proof system for reasoning about disjoint parallel quantum programs. For total correctness, we have the following:

\begin{thm}[Completeness for Total Correctness of Disjoint Parallel Quantum Programs]\label{thm-complete-dis} Let proof systems $qPD$ be extended with the parallel composition rule (R.PC.P) for tensor products of quantum predicates and appropriate auxiliary rules: \begin{equation*}
{\rm qTP} = {\rm qTD} \cup \{{\rm (R.PC.P), (R.CC), (R.Lin), (R.SO), (R.TI), (R.Lim)}\}.
\end{equation*}
Then qPP is complete for total correctness of disjoint parallel quantum programs; that is, for any disjoint quantum programs $P_1,...,P_n$ and quantum predicates $A, B$: 
$$\models_\mathit{tot} \{A\}P_1\|\cdots\|P_n\{B\} \Leftrightarrow\ \vdash_\mathit{qTP}\{A\}P_1\|\cdots\|P_n\{B\}.$$
\end{thm}

\begin{proof} The basic idea is essentially the same as Example \ref{exam-auxi}; namely: (1) introducing a fresh copy of each quantum variable as an auxiliary variable; (2) establishing the maximal entanglement between each original variable and its corresponding auxiliary variable; and (3) pushing certain entanglement between the auxiliary variables through the entanglement between the original and auxiliary variables to generate indirectly the entanglement between the original variables in precondition and postcondition. But the calculation is very involved, and we defer it to Appendix \ref{proof-d-complete}. 
\end{proof}

\begin{figure}[h]
\begin{align*}
&{\rm(R.PC.SP)}\quad\ \ \frac{P_i:{\rm Abort}(C_i)\qquad P_i:{\rm Term}(D_i)\qquad \left\{D_i+A_i\right\}P_i\left\{B_i\right\}\quad(i=1,\cdots,n)
}{\left\{I - \bigotimes_{i=1}^n (I_i-C_i) + \bigotimes_{i=1}^nA_i\right\} P_1\|\cdots\|P_n \left\{\bigotimes_{i=1}^nB_i\right\}}\\
&{\rm(R.A.P)}\qquad\ \ \frac{P_i:{\rm Abort}(A_i)\ (i=1,\cdots,n)}{P_1\|\cdots\|P_m: {\rm Abort}\left(I - \bigotimes_{i=1}^n (I_i-A_i)\right)}\\
&{\rm(R.T.P)}\qquad\ \ \frac{P_i:{\rm Term}(A_i)\ (i=1,\cdots,n)
}{P_1\|\cdots\|P_m: {\rm Term}\left(I - \bigotimes_{i=1}^m (I_i-A_i)\right)}
\end{align*}
\caption{Rules for Partial Correctness of Disjoint Parallel Programs.\ \ \ \ In this rules, $I_i$ is the identity operator on $\mathcal{H}_{P_i}$ for each $i$, and $I=\bigotimes_{i=1}^nI_i$ the identity operator on $\bigotimes_{i=1}^n\mathcal{H}_i$.}
\label{qPP-newrule}
\end{figure}

For partial correctness, however, we have to strengthen rule (R.PC.P) to (P.PC.SP) and introduce two rules for reasoning about abortion and termination of disjoint parallel programs. They are presented in Figure \ref{qPP-newrule}. With these new rules, we can prove the following: 

\begin{thm}[Completeness for Partial Correctness of Disjoint Parallel Quantum Programs]\label{complete-dis-p} Let proof systems $qPD$ be extended with the parallel composition rule (R.PC.SP) for tensor products of quantum predicates and appropriate auxiliary rules: \begin{equation*}
{\rm qPP} = {\rm qPD} \cup \{{\rm (R.PC.SP), (R.A.P), (R.T.P), (R.CC1), (R.CC2), (R.SO), (R.TI), (R.Lim)}\}.
\end{equation*}
Then qPP is complete for total correctness of disjoint parallel quantum programs; that is, for any disjoint quantum programs $P_1,...,P_n$ and quantum predicates $A, B$: 
$$\models_\mathit{par} \{A\}P_1\|\cdots\|P_n\{B\}\ {\rm implies}\ \vdash_\mathit{qPP}\{A\}P_1\|\cdots\|P_n\{B\}.$$
\end{thm}

\begin{proof} The idea is similar to the proof of Theorem \ref{thm-complete-dis}, but the calculation is much more involved. We defer it to Appendix \ref{proof-dp-complete}.
\end{proof}

\begin{rem}\label{aux-dis-p}\begin{enumerate}\item Sequentialisation rule (R.Seq) and rule (R.S2E) for transforming separable predicates to entangled ones are not included in the proof systems qPP and qTP. 

\item The rule (R.A.P) in the proof system qPP is actually a special case of (R.PC.SP) with $A_i=\llbracket P_i\rrbracket^\ast(B_i)$ and $D_i=I_i-\llbracket P_i\rrbracket^\ast(I_i)$.  

\item Note that assertions $P_i:{\rm Abort}(C_i)$ and $P_i:{\rm Term}(D_i)$ appear in the premise of rule (R.PC.SP). As pointed out in Remark \ref{rem-au-rules}, the first assertion can be verified in qPD, and the second can be verified in qTD but not in qPD. On the other hand, So, qPP is only complete relative to a theory about termination assertions $P:{\rm Term}(D)$, which is a sub-theory of qTD.\end{enumerate}
\end{rem}

\section{Syntax and Semantics of Parallel Quantum Programs with Shared Variables}\label{sec-shared-varaiables}

Disjoint parallel quantum programs were considered in the last two sections. This and next sections are devoted to deal with a class of more general parallel quantum programs, namely parallel quantum programs with shared variables. In this section, we first introduce their  syntax and operational and denotational semantics. 

\subsection{Syntax}

In this subsection, we define the syntax of parallel quantum programs with shared variables by removing the constraint of disjoint variables in Definition \ref{syntax-dp}.

\begin{defn}\begin{enumerate}\item Component quantum programs are generated by the grammar given in Eqs. (\ref{syntax}), (\ref{syntax+}) and (\ref{syntax++}) together with the following clause for atomic regions: $$P::=\langle P_0\rangle$$ where $P_0$ is loop-free and contains no further atomic regions; that is, it is generated only by Eqs. (\ref{syntax}) and (\ref{syntax+}).
\item Parallel quantum programs (with shared variables) are generated by the grammar given in Eqs. (\ref{syntax}), (\ref{syntax+}) and (\ref{syntax++}) together with the following clause for parallel composition:$$P::=P_1\|\cdots \| P_n\equiv \|_{i=1}^n P_i$$ where $n>1$, and $P_1,...,P_n$ are component quantum programs.
\end{enumerate}\end{defn}

The syntax of parallel quantum programs defined above is similar to that of classical parallel programs. In particular, as in the classical case, atomic regions are introduced to prevent interference from other components in their computation. A \textit{normal}  subprogram of program $P$ is defined to be a subprogram of $P$ that does not occur within any atomic region of $P$.

The set of quantum variables in a parallel quantum program is defined as follows: $\mathit{var}(\langle P\rangle)=\mathit{var}(P),$ and if $P\equiv P_1\|\cdots\|P_n$ then $$\mathit{var}(P)=\bigcup_{i=1}^n\mathit{var}(P_i).$$ Furthermore, the state Hilbert space of a parallel quantum program $P$ is $\mathcal{H}_P=\mathcal{H}_{\mathit{var}(P)}.$ It is worth pointing out that in general for a parallel quantum program $P\equiv P_1\|\cdots\|P_n$ with shared variables, $$\mathcal{H}_P\neq\bigotimes_{i=1}^n\mathcal{H}_{P_i}$$ because it is not required that $\mathit{var}(P_1),\cdots,\mathit{var}(P_n)$ are disjoint.

\subsection{Semantics}

In this subsection, we further define the operational and denotational semantics of parallel quantum programs with shared variables. Superficially, they are straightforward generalisations of the corresponding notions in classical programming. But as we already saw in Subsection \ref{disjoint-syntax}, even for disjoint parallel quantum programs, nondeterminism induced by quantum measurements and its intertwining with parallelism; in particular when some infinite computations of loops are involved, make the semantics much harder to deal with than in the classical case. We will see shortly that shared quantum variables brings a new dimension of complexity. 

\begin{defn} The operational semantics of parallel quantum programs is defined by the transitions rules in Figures \ref{fig 3.1} and \ref{fig 2} and rule (AR) in Figure \ref{fig shared-1} for atomic regions:
\begin{figure}[h]\centering
\begin{equation*}({\rm AR})\ \ \ \ \ \ \ \ \ \ \ \ \ \ \ \ \ \ \ \ \ \ \ \ \frac{\langle P,\rho\rangle\rightarrow^\ast\langle
\downarrow, \rho^{\prime}\rangle} {\langle
\langle P\rangle,\rho\rangle\rightarrow\langle
\downarrow,\rho^\prime\rangle}\ \ \ \ \ \ \ \ \ \ \ \ \ \ \ \ \ \ \ \ \ \ \ \ \ \ \ \ \ \ \ \ \ \ \ \ \ \ \ \ \ \ \ \ \end{equation*}
\caption{Transition Rule for Atomic Regions}\label{fig shared-1}
\end{figure}
\end{defn}

The rule (AR) means that any terminating computation of $P$ is reduced to a single-step computation of atomic region $\langle P\rangle$. Such a reduction guarantees that a computation of $\langle P\rangle$ may not be interfered by other components in a parallel composition. The rule (PC) in Figure \ref{fig 2} applies to both disjoint and shared-variable parallelism. 

Based on the operational semantics defined above, the denotational semantics of parallel quantum programs with shared variables can be defined in a way similar to but more involved than Definition \ref{dp-den-sem}. 
First, for a program $P$ and an input $\rho$, we recall from equation (\ref{value-set}) that $\mathcal{V}(P,\rho)$ is the set of values $\mathit{val}(\pi)$, where $\pi$ ranges over all computations of $P$ starting in $\rho$. We further define the upper closure of $\mathcal{V}(P,\rho)$: 
$$\overline{\mathcal{V}(P,\rho)}=\left\{\bigsqcup_k \rho_k: \left\{\rho_k\right\}\ {\rm is\ an\ increasing\ chain\ in}\ (\mathcal{V}(P,\rho),\sqsubseteq)\right\},$$  
where $\sqsubseteq$ is the L\"{o}wner order, and $\bigsqcup_k\rho_k$ stands for the least upper bound of $\left\{\rho_k\right\}$ in CPO $\left(\mathcal{D}\left(\mathcal{H}_P\right),\sqsubseteq\right)$, which always exists (\cite{Ying16}, Lemma 3.3.2). Then we have: 

\begin{defn}[Denotational Semantics]\label{dp-den-sem-1} The semantic function of a parallel program $P$ (with shared variables) is the mapping $\llbracket P\rrbracket:\mathcal{D}(\mathcal{H}_P)\rightarrow2^{\mathcal{D}(\mathcal{H}_P)}$ defined by $$\llbracket P\rrbracket (\rho)=\left\{{\rm maximal\ elements\ of}\ \left(\overline{\mathcal{V}(P,\rho)},\sqsubseteq\right)\right\}$$ for any $\rho\in\mathcal{D}(\mathcal{H}_P)$. 
\end{defn}

Let us carefully explain the design decision behind the above definition. First, it follows from rule (AR) that the semantics of an atomic region $\langle P\rangle$ is the same as that of $P$ as a \textbf{while}-program; that is, for any input $\rho$: $$\llbracket\langle P\rangle\rrbracket(\rho)=\llbracket P\rrbracket(\rho).$$ Second, we notice a difference between Definition \ref{dp-den-sem} for disjoint parallelism and Definition \ref{dp-den-sem-1} for shared-variable parallelism:   
in the latter, $\llbracket P\rrbracket(\rho)$ consists of the maximal elements of $\overline{\mathcal{V}(P,\rho)}$, rather than simply $\mathcal{V}(P,\rho)$ as in the former. Indeed, it is easy to show that $\left(\overline{\mathcal{V}(P,\rho)},\sqsubseteq\right)$ is inductive; that is, it contains an upper bound of every increasing chain in it. Then we see that $\llbracket P\rrbracket (\rho)$ is nonempty by Zorn's lemma. In particular, if $(\mathcal{V}(P,\rho),\sqsubseteq)$ has a maximal element, then it must be in $\llbracket P\rrbracket (\rho)$. In general, however, for a parallel program $P$ with shared variables, $\mathcal{V}(P,\rho)$ may have no maximal element, as shown in the following:

\begin{exam}Consider parallel program: $$P\equiv \mathbf{while}\ M[q]=1\ \mathbf{do}\ q:=U[q]\ \|\ q:=V[q]$$ where: \begin{itemize}\item two processes share a variable $q$, which is a qutrit with state Hilbert space $\mathcal{H}_q=\mathit{span}\ \{|0\rangle,|1\rangle,|2\rangle\};$
\item measurement $M=\left\{M_0,M_1\right\}$ with $M_0=|2\rangle\langle 2|$ and $M_1=|0\rangle\langle 0|+|1\rangle\langle 1|$;
\item unitary operators:
\begin{align*} U=|+\rangle\langle +|+e^{i\pi c}|-\rangle\langle -|+|2\rangle\langle 2|,\qquad V=|1\rangle\langle 0|+|2\rangle\langle 1|+|0\rangle\langle 2|.
\end{align*} Here, $|\pm\rangle=\frac{1}{\sqrt{2}}(|0\rangle\pm |1\rangle).$
\end{itemize} For input pure state $|0\rangle$, we can calculate $\mathit{val}(\pi)$ for a computation $\pi$ of $P$ in the following cases:

{\vskip 3pt}

\textbf{Case 1}. The second component $q:=V[q]$ is executed first, and then the \textbf{while}-loop (i.e. the first component) is executed. Then the state is first changed from $|0\rangle$ to $|1\rangle$, and it is $U^n|1\rangle\in\mathit{span}\ \{|0\rangle,|1\rangle\}$ immediately after the $n$th iteration of $U$ in the loop body. So, the program never terminates, and $\mathit{val}(\pi)=0$.

{\vskip 3pt}

\textbf{Case 2}. The \textbf{while}-loop is executed first and the second component is never executed. Then the program does not terminate and $\mathit{val}(\pi)=0$.

{\vskip 3pt}

\textbf{Case 3}. The \textbf{while}-loop is executed first, and then the second component is executed during the $n$th iteration. Then either $V$ occurs before $U$, and it holds that 
\begin{align*}\mathit{val}(\pi)=M_0UVU^{n-1}|0\rangle\langle 0|U^{\dag n-1}V^\dag U^\dag M_0^\dag = |\langle 2|UVU^{n-1}|0\rangle|^2\cdot |2\rangle\langle 2|=|\langle 1|U^{n-1}|0\rangle|^2\cdot |2\rangle\langle 2|, 
\end{align*} or $U$ occurs before $V$, and  
\begin{align*}\mathit{val}(\pi)=M_0VU^{n}|0\rangle\langle 0|U^{\dag n} M_0^\dag = |\langle 2|VU^{n}|0\rangle|^2\cdot |2\rangle\langle 2|=|\langle 1|U^{n}|0\rangle|^2\cdot |2\rangle\langle 2|.
\end{align*} Note that $$|\langle 1|U^n|0\rangle|^2=|\frac{1-e^{i\pi nc}}{2}|^2=\frac{1-\cos\pi nc}{2}.$$ Then we obtain:$$\mathcal{V}(P,|0\rangle\langle 0|)=\{0\}\cup\left\{\frac{1-\cos\pi nc}{2}\cdot |2\rangle\langle 2|: n=0,1,2,...\right\}.$$ If we choose parameter $c$ being an irrational number, then by Kronecker's theorem we assert that the set $\left\{\frac{1-\cos\pi nc}{2}: n=0,1,2,...\right\}$ of coefficients is dense in the unit interval $[0,1]$, but the supremum $1$ is not attainable. Therefore, $\mathcal{V}(P,|0\rangle\langle 0|)$ has no maximal element with respect to the L\"{o}wner order $\sqsubseteq$. Furthermore, it holds that $\overline{\mathcal{V}(P,|0\rangle\langle 0|)}=\{a\cdot |2\rangle\langle 2|: a\in [0,1]\}$, and thus $\llbracket P\rrbracket(|0\rangle\langle 0|)=|2\rangle\langle 2|$. 
\end{exam}

To conclude this subsection, we present an example showing the difference between the behaviours of a quantum program and its atomic version in parallel with another quantum program involving a quantum measurement on a shared variable.
\begin{exam}\label{exam-one-more} Let $p,q$ be qubit variables and \begin{align*}P_1\equiv p:=H[p];p:=H[p],\ \ \ \ \ \ P_1^\prime\equiv \langle P_1\rangle,\ \ \ \ \ \
&P_2\equiv\mathbf{if}\ M[p]=0\rightarrow \mathbf{skip}\\
&\ \ \ \ \ \ \ \ \ \square\ \ \ \ \ \ \ \ \ \ \ \ \ 1\rightarrow q:=X[q]\\
&\ \ \ \ \ \ \ \ \ \mathbf{fi}
\end{align*} where $H, X$ are the Hadamard and Pauli gates, respectively and $M$ the measurement in the computational basis. Consider the EPR (Einstein-Podolsky-Rosen) pair $|\psi\rangle=\frac{1}{\sqrt{2}}(|00\rangle+|11\rangle)$ as an input, where the first qubit is $p$ and the second is $q$.
\begin{enumerate}\item One of the computations of parallel composition $P_1^\prime\|P_2$ is \begin{align*}
\pi_1=\langle P_1^\prime\|P_2,|\psi\rangle\rangle &\rightarrow_1\langle\downarrow\|P_2,|\psi\rangle\rangle\rightarrow_2\left\{\langle\downarrow\|\mathbf{skip},\frac{1}{\sqrt{2}}|00\rangle\rangle,\langle\downarrow\|q:=X[q],\frac{1}{\sqrt{2}}|11\rangle\rangle\right\}\\
&\rightarrow_2\left\{\langle\downarrow,\frac{1}{\sqrt{2}}|00\rangle\rangle,\langle\downarrow,\frac{1}{\sqrt{2}}|10\rangle\rangle\right\}
\end{align*}  Indeed, for all other computations $\pi$ of $P_1^\prime\|P_2$ starting in $|\psi\rangle$, we have: $$\mathit{val}(\pi)=\mathit{val}(\pi_1)=\frac{1}{2}(|00\rangle\langle 00|+|10\rangle\langle 10|)\stackrel{\triangle}{=}\rho_1,$$ and thus $\llbracket P_1^\prime\|P_2\rrbracket(|\psi\rangle)=\{\rho_1\}.$
\item $P_1\|P_2$ has a computation starting in $|\psi\rangle$ that is quite different from $\pi_1$:
\begin{align*}
\pi_2=\ &\langle P_1\|P_2,|\psi\rangle\rangle \rightarrow_1\left\{\langle p:=H[p]\|P_2,\frac{1}{\sqrt{2}}(|+0\rangle+|-1\rangle)\rangle\right\}\\ &\rightarrow_2\begin{cases}\langle p:=H[p]\|\mathbf{skip},\frac{1}{2}(|00\rangle+|01\rangle)\rangle,\\ \langle p:=H[p]\|q:=X[q],\frac{1}{2}(|10\rangle-|11\rangle)\rangle\end{cases} \rightarrow_2\begin{cases}\langle p:=H[p]\|\downarrow,\frac{1}{2}(|00\rangle+|01\rangle)\rangle,\\ \langle p:=H[p]\|\downarrow,\frac{1}{2}(|11\rangle-|10\rangle)\rangle\end{cases}\\
&\rightarrow_1\left\{\langle \downarrow,\frac{1}{2}(|+0\rangle+|+1\rangle)\rangle, \langle \downarrow,\frac{1}{2}(|-1\rangle-|-0\rangle)\rangle\right\}\end{align*}
\end{enumerate} We have: \begin{align*}\mathit{val}(\pi_1)\neq\mathit{val}(\pi_2)=\ &\frac{1}{4}(|00\rangle\langle 00|+|00\rangle\langle 11|+|01\rangle\langle 01|+|01\rangle\langle 10|\\ &+|10\rangle\langle 01|+|10\rangle\langle 10|+|11\rangle\langle 00|+|11\rangle\langle 11|)\stackrel{\triangle}{=}\rho_2
\end{align*} and $\llbracket P_1\|P_2\rrbracket(|\psi\rangle)=\{\rho_1,\rho_2\}$.
\end{exam}

The above example indicates that the determinism of the denotational semantics of disjoint parallel quantum programs (Lemma \ref{lem-det}) is no longer true for parallel quantum programs with shared variables.

\subsection{Correctness of Parallel Quantum Programs}

Now we can define the notion of correctness for parallel quantum programs with shared variables based on their denotational semantics introduced in the previous subsection. As pointed out at the beginning of last section, the definition of correctness of quantum \textbf{while}-programs (Definition \ref{correctness-interpretation}) can be directly adopted for disjoint parallel quantum programs. However, Example \ref{exam-one-more} shows that for a parallel quantum program $P$ with shared variables and an input $\rho$, $\llbracket P\rrbracket(\rho)$ may have more than one element. Therefore, the notion of correctness of quantum \textbf{while}-programs is not directly applicable to parallel quantum programs with shared variables. But a simple modification of it works.

\begin{defn}[Partial and Total Correctness]\label{correctness-parallel} Let $P$ be a parallel quantum program (with shared variables) and A, B quantum predicates in $\mathcal{H}_P$.
Then the correctness formula $\{A\}P\{B\}$ is true in
the sense of total correctness (resp. partial correctness), written $$\models_{\mathit{tot}}\{A\}P\{B\}\ \ \ \ (\mathit{resp.}\ \models_{\mathit{par}}\{A\}P\{B\}),$$ if for each
input $\rho\in\mathcal{D}(\mathcal{H}_P)$, it holds that \begin{align*}\tr(A\rho)\leq \tr(B\rho^\prime)\ \ \ \ 
(\mathit{resp.}\ \tr(A\rho)\leq \tr(B\rho^\prime)+
[\tr(\rho)-\tr(\rho^\prime)])\end{align*} for all $\rho^\prime\in\llbracket P\rrbracket(\rho)$.
\end{defn}

\section{Proof Rules for Parallel Quantum Programs with Shared Variables}\label{correct-shared}

Our aim of this section is to introduce some useful rules for reasoning about correctness of parallel quantum programs with shared variables. 
In Section \ref{dis-correctness}, we were able to develop a (relatively) complete logical system for disjoint parallel quantum programs by finding an appropriate quantum generalisation of a special case of rule (R.PC) in  Figure \ref{fig -1} (i.e. Hoare's parallel rule) together with several auxiliary rules. Unfortunately, the idea used in Section \ref{dis-correctness} does not work here because the third major challenge pointed out in the Introduction - combining quantum predicates in the overlap of state Hilbert spaces - will emerge in the case of shared variables. Let us gradually introduce a new idea to partially avoid this hurdle.  

\subsection{A Rule for Component Quantum Programs}
As a basis for dealing with parallel quantum programs, we first consider component quantum programs. The proof techniques for classical component programs can be generalised to the quantum case without any difficulty. More precisely,
partial and total correctness of component quantum programs can be verified with the proof system qPD and qTD for quantum \textbf{while}-programs plus the rule (R.AT) in Figure \ref{fig atomic} for atomic regions.
\begin{figure}[h]\centering
\begin{equation*}({\rm R.At})\ \ \ \ \ \ \ \ \ \ \ \ \ \ \ \ \ \ \ \ \ \ \ \ \ \ \ \frac{\{A\}P\{B\}}{\{A\}\langle P\rangle\{B\}}\ \ \ \ \ \ \ \ \ \ \ \ \ \ \ \ \ \ \ \ \ \ \ \ \ \ \ \ \ \ \ \ \ \ \ \ \ \ \ \ \ \ \end{equation*}
\caption{Rule for Atomic Regions.}\label{fig atomic}
\end{figure}

\subsection{Proof Outlines}\label{sec-outline}

The most difficult issue in reasoning about parallel programs with shared variables is interference between their different components. The notion of proof outline was introduced in classical programming theory so that the proofs of programs can be organised in a structured way. More importantly, it provides an appropriate way to describe interference freedom between the component programs --- a crucial premise in inference rule (R.PC) for a parallel program with shared variables. So in this subsection, we generalise the notion of proof outline to quantum \textbf{while}-programs so that it can be used in next subsection to present our inference rules for parallel quantum programs with shared variables.

\begin{defn} Let $P$ be a quantum \textbf{while}-program. A proof outline for partial correctness of $P$ is a formula $$\{A\}P^\ast\{B\}$$ formed by the formation axioms and rules in Figure \ref{fig 3.2+}, where $P^\ast$ results from $P$ by interspersing quantum predicates.
\begin{figure}[h]\centering
\begin{equation*}\begin{split}
&({\rm Ax.Sk'})\ \ \ \{A\}\mathbf{Skip}\{A\}\ \ \ \ \ \ \ \ \ \ \ \ \ \ \ \ \ \ \ \ \ \ \ \ \ \ \ \ \ \ \ \ \ \ \ \ \ \ \ \ ({\rm Ax.In'}) \ \ \ \left\{\sum_{i}|i\rangle_q\langle 0|A|0\rangle_q\langle
i|\right\}q:=|0\rangle\{A\}\\
&({\rm Ax.UT'})\ \ \
\{U^{\dag}AU\}\overline{q}:=U\left[\overline{q}\right]\{A\}\ \ \ \ \ \ \ \ \ \ \ \ \ \ \ \ \ \ \ \ \ \ \ \ \ \ \ ({\rm R.SC'})\ \ \
\frac{\{A\}P^\ast_1\{B\}\ \ \ \ \ \ \{B\}P^\ast_2\{C\}}{\{A\}P^\ast_1;\{B\}P^\ast_2\{C\}}\\
&({\rm R.IF'})\ \ \
\frac{\left\{A_{m_i}\right\}P^\ast_{m_i}\left\{B\right\}\ (i=1,...,k)}{\begin{array}{ccc}\left\{\sum_i^{k}
M_{m_i}^{\dag}A_{m_i}M_{m_i}\right\}\ \mathbf{if}\
M[\overline{q}]=m_1\rightarrow \left\{A_{m_1}\right\}P^\ast_{m_1}\\  \ \ \ \ \ \ \ \ \ \ \ \ \ \ \ \ \ \ \ \
............ \\
\ \ \ \ \ \ \ \ \ \ \ \ \ \ \ \ \ \ \ \ \ \ \ \ \ \ \ \ \ \ \ \ \ \square\ M[\overline{q}]=m_k\rightarrow \left\{ A_{m_k}\right\}P^\ast_{m_k}\\ \ \ \ \ \ \ \ \ \ \ \ \ \ \ \ \ \ \ \ \ \ \ \ \ \ \ \ \ \ \ \ \ \ \ \ \ \ \ \ \ \ \ \ \ \ \
\mathbf{fi}\ \{B\}\ \ \ \ \ \ \ \ \ \ \ \ \ \ \ \ \ \ \ \ \ \ \ \ \ \ \ \ \ \ \ \ \ \ \ \ \ \ \ \ \ \ \ \ \ \ \ \ \ \end{array}}\\
&({\rm R.LP'})\ \ \
\frac{\{B\}P^\ast\left\{M_0^{\dag}AM_0+M_1^{\dag}BM_1\right\}}{\begin{array}{ccc}\left\{\mathbf{inv}: M_0^{\dag}AM_0+M_1^{\dag}BM_1\right\}\ \mathbf{while}\
M[\overline{q}]=1\ \mathbf{do}\ \left\{ B\right\}\ P^\ast\ \left\{M_0^{\dag}AM_0+M_1^{\dag}BM_1\right\}\ \mathbf{od}\ \{A\}\end{array}}\\
&({\rm R.Or'})\ \ \ \frac{A\sqsubseteq
A^{\prime}\ \ \ \ \{A^{\prime}\}P^\ast\{B^{\prime}\}\ \ \ \
B^{\prime}\sqsubseteq B}{\{A\}\{A^\prime\}P\{B^\prime\}\{B\}}\ \ \ \ \ \ \ \ \ \ \ \ \ \ \ \ \ ({\rm R.Del})\ \ \ \ \frac{\{A\}P^\ast\{B\}}{\{A\}P^{\ast\ast}\{B\}}
\end{split}\end{equation*}
\caption{Formation Axioms and Rules for Partial Correctness of Quantum \textbf{while}-Programs.\ \ \ \ In (R.IF'), $\{m_1,...,m_k\}$ is the set of all possible outcomes of measurement $M$. In (R.Del), $P^{\ast\ast}$ is obtained by deleting some quantum predicates from $P^\ast$, expect those labelled with \textquotedblleft\textbf{inv}\textquotedblright.}\label{fig 3.2+}
\end{figure}
\end{defn}

Obviously, (Ax.Sk'), (Ax.In'), (Ax.UT') are the same as (Ax.Sk), (Ax.In) and (Ax.UT), respectively, in Figure \ref{fig 3.2}. But (R.SC'), (R.IF'), (R.LP') and (R.Or') in Figure \ref{fig 3.2+} are obtained from their counterparts in Figure \ref{fig 3.2} by interspersing intermediate quantum predicates in appropriate places; for example, in rule (R.IF'), a predicate $A_{m_i}$ is interspersed into the branch corresponding to measurement outcome $m_i$. In particular, keyword \textquotedblleft\textbf{inv}\textquotedblright\ is introduced in rule (R.LP') to indicate loop invariants (see \cite{YYW17}, Example 4.1 for a discussion about invariants of quantum \textbf{while}-loops).Furthermore, rule (R.Del) is introduced to delete redundant intermediate predicates.    

The notion of proof outline for total correctness of quantum \textbf{while}-programs can be defined in a similar way; but we omit it here because in the rest of this section, for simplicity of presentation, we only consider partial correctness of parallel quantum programs (the proof techniques introduced in this section can be easily generalised to the case of total correctness by adding ranking functions). 

We will mainly use a special form of proof outlines defined in the following:

\begin{defn}A proof outline $\{A\}P^\ast\{B\}$ of quantum \textbf{while}-program $P$ is called standard if every subprogram $Q$ of $P$ is proceded by exactly one quantum predicate, denoted $\mathit{pre}(Q)$, in $P^\ast$.\end{defn}

The following proposition shows that the notion of standard proof outline is general enough for our purpose.

\begin{prop}\begin{enumerate} For any quantum \textbf{while}-program $P$, we have: \item If $\{A\}P^\ast\{B\}$ is a proof outline for partial correctness, then $\vdash_\mathit{qPD}\{A\}P\{B\}$.
\item If $\vdash_\mathit{qPD}\{A\}P\{B\}$, then there is a standard proof outline $\{A\}P^\ast\{B\}$ for partial correctness.
\end{enumerate}
\end{prop}

\begin{proof}This proposition can be easily proved by induction on the lengths of proof and formation; in particular, employing rule (R.Del).\end{proof}

The notion of proof outline enables us to present a soundness of quantum Hoare logic stronger than the soundness part of Theorem \ref{sound-complete}. It indicates that soundness is well maintained in each step of the proofs of quantum \textbf{while}-programs. To this end, we need an auxiliary notation defined in the following:

\begin{defn}\label{def-head}Let $P$ be a quantum \textbf{while}-program and $T$ a subprogram of $P$. Then $\mathit{at}(T,P)$ is inductively defined as follows:\begin{enumerate}\item If $T\equiv P$, then $\mathit{at}(T,P)\equiv P$;
\item If $P\equiv P_1;P_2$, then $$\mathit{at}(T,P)\equiv\begin{cases}\mathit{at}(T,P_1);P_2\ &{\rm when}\ T\ {\rm is\ a\ subprogram\ of}\ P_1,\\ \mathit{at}(T,P)\equiv\mathit{at}(T,P_2)\ &{\rm when}\ T\ {\rm is\ a\ subprogram\ of}\ P_2;\end{cases}$$ \item If $P\equiv\mathbf{if}\ (\square m\cdot M[\overline{q}]=m\rightarrow P_m)\ \mathbf{fi}$, then for each $m$, whenever $T$ is a subprogram of $P_m$, $\mathit{at}(T,P)\equiv\mathit{at}(T,P_m)$; \item If $P\equiv\mathbf{while}\ M[\overline{q}]=1\ \mathbf{do}\ P^\prime\ \mathbf{od}$ and $T$ is a subprogram of $P^\prime$, then $\mathit{at}(T,P)\equiv\mathit{at}(T,P^\prime);P$.
\end{enumerate}
\end{defn}

Intuitively, $at(T,P)$ is (a syntactic expression of) the remainder of program $P$ that is to be executed when the program control reach subprogram $T$. For a simple presentation, here we slightly abuse the notation $\mathit{at}(T,P)$ because the same subprogram $T$ can appear in different parts of $P$. So, $\mathit{at}(T,P)$ is actually defined for a fixed occurrence of $T$ within $P$. 

Now we are ready to present the strong soundness theorem for quantum \textbf{while}-programs.

\begin{thm}[Strong Soundness for Quantum \textbf{while}-Programs]\label{thm.strong-sound} Let $\{A\}P^\ast\{B\}$ be a standard proof outline for partial correctness of quantum \textbf{while}-program $P$. If $$\langle P,\rho\rangle\rightarrow^\ast\{|\langle P_i,\rho_i\rangle|\},$$ then: \begin{enumerate}
\item for each $i$, $P_i\equiv \mathit{at}(T_i,P)$ for some subprogram $T_i$ of $P$ or $P_i\equiv\ \downarrow$; and
\item it holds that $$\mathit{tr}(A\rho)\leq\sum_i\mathit{tr}\left(B_i\rho_i\right),$$ where $$B_i=\begin{cases}B\ &{\rm if}\ P_i\equiv\ \downarrow,\\
\mathit{pre}\left(T_i\right)\ &{\rm if}\ P_i\equiv\mathit{at}\left(T_i,P\right).
\end{cases}$$
\end{enumerate}
\end{thm}

\begin{proof} See Appendix \ref{proof-strong-sound}.\end{proof}

The soundness for quantum \textbf{while}-programs given in Theorem \ref{sound-complete} can be easily derived from the above theorem. Of course, the above theorem is a generalisation of the strong soundness for classical \textbf{while}-programs (see \cite{Apt09}, Theorem 3.3). But it is worthy to notice a major difference between them: due to the branching caused by quantum measurements, in the right-hand side of the inequality in clause (2) of the above theorem, we have to take a summation over a configuration ensemble $\{|\langle P_i,\rho_i\rangle|\}$ rather than considering a single configuration $\langle P_i,\rho_i\rangle$.

Proof outlines for partial correctness of component quantum programs are generated by the rules in Figure \ref{fig 3.2+} together with the rule (R.At') in Figure \ref{fig atomic+}. A proof outline of a component program $P$ is standard if every normal subprogram $Q$ is preceded by exactly one quantum predicate $\mathit{pre}(Q)$. The notation $\mathit{at}(T,P)$ is defined in the same way as in Definition \ref{def-head}, but only for normal subprograms $T$ of $P$.
The strong soundness theorem for quantum \textbf{while}-programs (Theorem \ref{thm.strong-sound}) can be easily generalised to the case of component quantum programs.
\begin{figure}[h]\centering
\begin{equation*}({\rm R.At'})\ \ \ \ \ \ \ \ \ \ \ \ \ \ \ \ \ \ \ \ \ \ \ \ \ \ \ \frac{\{A\}P^\ast\{B\}}{\{A\}\langle P\rangle\{B\}}\ \ \ \ \ \ \ \ \ \ \ \ \ \ \ \ \ \ \ \ \ \ \ \ \ \ \ \ \ \ \ \ \ \ \ \ \ \ \ \ \ \ \end{equation*}
\caption{Rule for Atomic Regions.}\label{fig atomic+}
\end{figure}

\subsection{Interference Freedom} With the preparation given in the previous subsection, we can consider how can we reason about correctness of parallel quantum programs with shared variables. Let us start from the following example showing non-compositionality in the sense that correctness of a parallel quantum program is not solely determined by correctness of its component programs.

\begin{exam}Let $q$ be a quantum variable of type $\mathbf{Bool}$ (Boolean) or $\mathbf{Int}$ (Integers). Consider the following two programs:\begin{align*}
&P_1\equiv q:=U[q],\ \ \ \ \ P_1^\prime\equiv q:=V[q]; q:=W[q]
\end{align*} where $U,V,W$ are unitary operators in $\mathcal{H}_q$ such that $U=WV$. It is obvious that $P_1$ and $P_1^\prime$ are equivalent in the following sense: for any quantum predicates $A,B$ in $\mathcal{H}_q$, $$\models_\mathit{par}\{A\}P_1\{B\}\Leftrightarrow\ \models_\mathit{par}\{A\}P_1^\prime\{B\}.$$ Now let us further consider their parallel composition with the simple initialisation program: $$P_2\equiv q:=|0\rangle.$$
We show that $P_1\|P_2$ and $P_1^\prime\|P_2$ are not equivalent; that is, $$\models_\mathit{par}\{A\}P_1\|P_2\{B\}\Leftrightarrow\ \models_\mathit{par}\{A\}P_1^\prime\|P_2\{B\}$$ is not always true. Let us define the deformation index of unitary operator $U$ as $$D(U)=\inf_\rho\frac{\langle 0|U\rho U^\dag|0\rangle}{\langle 0|\rho|0\rangle}.$$ Then we have:
 \begin{align}\label{deformation-1}&\models_\mathit{par}\{\lambda\cdot |0\rangle\langle 0|\}P_1\|P_2\{|0\rangle\langle 0|\}\ {\rm if\ and\ only\ if}\ \lambda\leq \min\left[D(U), |\langle 0|U|0\rangle|^2\right];\\
 \label{deformation-2}&\models_\mathit{par}\{\lambda\cdot |0\rangle\langle 0|\}P_1^\prime\|P_2\{|0\rangle\langle 0|\}\ {\rm if\ and\ only\ if}\ \lambda\leq \min\left[D(U), D(V)\cdot |\langle 0|W|0\rangle|^2, |\langle 0|U|0\rangle|^2\right].\end{align}  It is easy to see that the partial correctness in (\ref{deformation-1}) is true but the one in (\ref{deformation-2}) is false when $q$ is a qubit, $\lambda=1$, $U=I$ (the identity) and $V=W=H$ is the Hadamard gate.
\end{exam}

The above example clearly illustrates that as in the case of classical parallel programs, we have to take into account interference between the component programs of a parallel quantum program. Moreover, appearance of parameter $\lambda$ in Eqs. (\ref{deformation-1}) and (\ref{deformation-2}) indicates that interference between quantum programs is subtler than that between classical programs. It motivates us to introduce a parameterised notion of interference freedom for quantum programs. Let us first consider interference between a quantum predicate and a proof outline.

\begin{defn}Let $0\leq \lambda\leq 1$, and let $A$ be a quantum predicate and $\{B\}P^\ast\{C\}$ a standard proof outline for partial correctness of quantum component program $P$. We say that $A$ is $\lambda$-interference free with $\{B\}P^\ast\{C\}$ if: \begin{itemize}
\item for any atomic region, normal initialisation or unitary transformation $Q$ in $P$, it holds that \begin{equation}\label{interf-1}\models_\mathit{par}\{\lambda A+(1-\lambda)\mathit{pre}(Q)\}Q\{\lambda A+(1-\lambda)\mathit{post}(Q)\}\end{equation} where $\mathit{post}(Q)$ is the quantum predicate immediately after $Q$ in $\{B\}P^\ast\{C\}$;
\item for any normal case statement $Q\equiv\mathbf{if}\ (\square\ M[q]=m\rightarrow Q_m)\ \mathbf{fi}$ in $P$, it holds that
\begin{equation}\label{interf-2}\begin{split}\models_\mathit{par}\{\lambda A+(1-\lambda)\mathit{pre}(Q)\}\ &\mathbf{if}\ \left(\square M[q]=m \rightarrow\ \{\lambda A+(1-\lambda)\mathit{post}_m(Q)\}\ Q_m\right)\\ &\mathbf{fi}\ \{\lambda A+(1-\lambda)\mathit{post}(Q)\}\end{split}\end{equation}
where $\mathit{post}_m(Q)$ is the quantum predicate immediately after the $m$th branch of $Q$ in $\{B\}P^\ast\{C\}$.\end{itemize}
\end{defn}

\begin{rem} The reader might be wondering about why $\mathit{post}(Q)$ and $\mathit{post}_m(Q)$ appear in equations (\ref{interf-1}) and (\ref{interf-2}). This looks very different from the classical case. When defining interference freedom of $A$ with $\{B\}P^\ast\{C\}$ for a classical program $P$, we only require that \begin{equation}\label{interf-3.0}\models_\mathit{par}\{A\wedge\mathit{pre}(Q)\}Q\{A\}\end{equation} for each basic statement $Q$ in $P$ (see \cite{Apt09}, Definition 8.1). Actually, the difference between the classical and quantum cases is not as big as what we think at the first glance. In the classical case, condition (\ref{interf-3.0}) can be combined with $$\models_\mathit{par}\{\mathit{pre}(Q)\}Q\{\mathit{post}(Q)\},$$ which holds automatically, to yield: \begin{equation}\label{interf-3}\models_\mathit{par}\{A\wedge\mathit{pre}(Q)\}Q\{A\wedge\mathit{post}(Q)\}.\end{equation} If conjunctive $\wedge$ in equation (\ref{interf-3}) is replaced by a convex combination (with probabilities $\lambda$ and $1-\lambda$), then we obtain equations (\ref{interf-1}) and (\ref{interf-2}).
\end{rem}

The above definition can be straightforwardly generalised to the notion of interference freedom between a family of proof outlines, where noninterference between each quantum predicate in one proof outline and another proof outline is required. 

\begin{defn}Let $\{A_i\}P_i^\ast\{B_i\}$ be a standard proof outline for partial correctness of quantum component program $P_i$ for each $1\leq i\leq n$. \begin{enumerate}\item If $\Lambda=\{\lambda_{ij}\}_{i\neq j}$ is a family of real numbers in the unit interval, then we say that $\{A_i\}P^\ast\{B_i\}$ $(i=1,...,n)$  are $\Lambda$-interference free whenever for any $i\neq j$, each quantum predicate $C$ in $\{A_i\}P_i^\ast\{B_i\}$ is $\lambda_{ij}$-interference free with $\{A_j\}P^\ast_j\{B_j\}$.
\item In particular, $\{A_i\}P^\ast\{B_i\}$ $(i=1,...,n)$  are said to be $\lambda$-interference free if they are $\Lambda$-interference free for $\Lambda=\{\lambda_{ij}\}_{i\neq j}$ with $\lambda_{ij}\equiv\lambda$ (the same parameter) for all $i\neq j$.
\end{enumerate}\end{defn}

\subsection{A Rule for Parallel Composition of Quantum Programs with Shared Variables}\label{sec-local-H}

The notion of interference freedom introduced above provides us with a key ingredient in defining a quantum extension of inference rule (R.PC) for parallelism with shared variables. Another key ingredient would be a quantum generalisation of the logical conjuction used in combining the preconditions and postconditions. As discussed in the Introduction, tensor product is not appropriate for this purpose, but probabilistic (convex) combination can serve as a kind of approximation of conjunction. This idea leads to rule (R.PC.L) in Figure \ref{fig 3+s}.

\begin{figure}[h]\centering
\begin{align*}({\rm R.PC.L})\ \ \ \ \ \ \ \ \ \ \ \ \frac{{\rm Standard\ proof\ outlines}\ \left\{A_i\right\}P^\ast_i\left\{B_i\right\} (i=1,...,n)\ {\rm are}\ \Lambda{\rm -interference\ free}}{\left\{\sum_{i=1}^np_iA_i\}P_1\|\cdots\|P_n\{\sum_{i=1}^np_iB_i\right\}}\ \ \ \ \ \ \ \ \ \ \ \ \ \ \ \ \ \ \ \ \ \
\end{align*}
\caption{Rule for Parallel Quantum Programs with Shared Variables.\ \ \ \ \ \ $\{p_i\}_{i=1}^n$ is a probability distribution, and $\Lambda=\{\lambda_{ij}\}_{i\neq j}$ satisfies: $\sum_{i\neq j}\frac{p_i}{\lambda_{ij}}\leq 1$ for every $j$.}\label{fig 3+s}
\end{figure}

It is worth carefully comparing rule (R.PC.L) with (R.PC.P) for disjoint parallel quantum programs. First, $\Lambda$-interference freedom in (R.PC.L) is not necessary in (R.PC.P), since disjointness implies interference freedom. Second, conjunctions $\bigwedge_iA_i$ and $\bigwedge_iB_i$ of preconditions and postconditions in rule (R.PC) for classical parallel programs are replaced by tensor products $\bigotimes_iA_i$ and $\bigotimes_iB_i$ in (R.PC.P). But in (R.PC.L), programs $P_1,...,P_n$ are allowed to share variables, the tensor products of preconditions and postconditions are then not always well-defined. So, we choose to use probabilistic combinations $\sum_ip_iA_i$ and $\sum_ip_iB_i$. 
Obviously, probabilistic combination is not a perfect quantum generalisation of conjunction. 

Let us first give a simple example to illustrate how to use rule (R.PC.L) in reasoning about shared-variable parallel quantum programs. 

\begin{exam} Let $q_1,q_2,r$ be three qubit variables, and let $P_i$ be a quantum programs with variables $q_i$ and $r$:
\begin{align*}
P_i&\equiv q_i := |0\>;\ q_i := H[q_i];\ q_i,r := {\rm CNOT}[q_i,r]
\end{align*}
for $i=1,2$, where ${\rm CNOT}$ is the control-NOT gate with $q_i$ as the control qubit and $r$ as the data qubit, and $H$ is the Hadamard gate. Note that $P_1$ and $P_2$ have a shared variable $r$. We consider their parallel composition $P_1\|P_2$. 
Using rule (R.PC.L), we can derive its correctness formula:
\begin{equation}\vdash_\mathit{par}\left\{\frac{\sqrt{2}}{2}|\psi\>\<\psi|\right\}P_1\parallel P_2\left\{|\psi\>\<\psi|\right\}.\label{equ:share-form}\end{equation}
where the pure state $|\psi\rangle$ in the precondition and postcondition is given as follows:
$$|\psi\> = \frac{\sqrt{2}+1}{4}[|000\>+|001\>]+\frac{\sqrt{2}-1}{4}[|110\>+|111\>]+\frac{1}{4}[|010\>+|011\>+|100\>+|101\>].$$
with the order of register: $q_1,q_2,r$. 
First, we have the proof outlines of $P_i$:
\begin{align*}
&\left\{\frac{\sqrt{2}}{2}|\psi\>\<\psi|\right\}q_i := |0\>\{|\psi\>\<\psi|\} q_i:=H[q_i]\{|\psi\>\<\psi|\} q_i,r := {\rm CNOT}[q_i,r]\{|\psi\>\<\psi|\}
\end{align*}
for $i=1,2$, respectively. Moreover, one can verify that these two proof outlines are $0.5$-interference free because 
$$\vdash_{qPD}\left\{\frac{2+\sqrt{2}}{4}|\psi\>\<\psi|\right\}q_i := |0\>\{|\psi\>\<\psi|\}.$$
Then (\ref{equ:share-form}) is derived from (R.PC.L) with $p_0=0.5, p_1=0.5$.

One may show that with the postcondition $|\psi\>\<\psi|$, the maximal factor $c$ which guarantees validity of the correctness formula
$$\models_\mathit{par}\{c|\psi\>\<\psi|\}P_1\|P_2\{|\psi\>\<\psi|\}$$
is $c_{\max} = \frac{3+2\sqrt{2}}{8} \approx 0.728$. The the factor $\frac{\sqrt{2}}{2} \approx 0.707$ we derived in (\ref{equ:share-form}) is very close to $c_{\max}$, but a formal derivation of $c_{\max}$ is much more involved and omitted here.
\end{exam}

\begin{rem}\label{remark-loc-ham}  For some more sophisticated applications, a combination of (R.PC.P) and (P.PC.L) can achieve a better quantum approximation of the conjunctions in (R.PC). We first find maximal subfamilies, say $\mathcal{P}_j$ of $P_1,...,P_n$ of which the elements are disjoint. Then we can apply (R.PC.P) to each of these subfamily to derive:\begin{equation}\label{convex-comp}\vdash_\mathit{par}\{C_i\}\|_{P_i\in\mathcal{P}_j} P_i\{D_j\}\end{equation} where $$C_j=\bigotimes_{P_i\in\mathcal{P}_j}A_i,\qquad D_j=\bigotimes_{P_i\in\mathcal{P}_j}B_i.$$ Furthermore, a probabilistic combination of (\ref{convex-comp}) can be derived as $$\vdash_\mathit{par}\left\{\sum_jp_jC_j\right\}P_1\|\cdots\|P_n \left\{\sum_jp_jD_j\right\}.$$
We believe that this idea is strong enough to derive a large class of useful correctness properties of parallel quantum programs with shared variables. The reason is that in many-body physics, an overwhelming majority of systems of physics interest can be described by local Hamiltonian: $H=\sum_{j}H_j,$ where each $H_j$ is $k$-local, meaning that it acts over at most $k$ components of the system. It is clear that the above idea can be used to prove correctness of parallel quantum programs with their preconditions and postconditions being local Hamiltonians.
\end{rem}

Theorem \ref{thm.strong-sound} can be generalised from quantum \textbf{while}-programs to parallel quantum program, showing the strong soundness of inference rule (R.PC.L) (combined with the other rules introduced in this paper):

\begin{thm}[Strong Soundness for Parallel Quantum Programs with Convex Combination of Quantum Predicates]\label{convex-sound} Let $\{A_i\}P_i^\ast\{B_i\}$ be a standard proof outline for partial correctness of quantum component program $P_i$ $(i=1,...,n)$ and $$\langle P_1\|\cdots\|P_n,\rho\rangle\rightarrow^\ast\left\{|\langle P_{1s}\|\cdots\| P_{ns},\rho_s\rangle|\right\}.$$ Then:
\begin{enumerate}\item for each $1\leq i\leq n$ and for every $s$, $P_{is}\equiv \mathit{at}(T_{is},P_i)$ for some normal subprogram $T_{is}$ of $P_i$ or $P_{is}\equiv\ \downarrow$; and

\item for any probability distribution $\{p_i\}_{i=1}^n$, if $\{A_i\}P_i^\ast\{B_i\}$ $(i=1,...,n)$ are $\Lambda$-interference free for some $\Lambda=\{\lambda_{ij}\}_{i\neq j}$ satisfying \begin{equation}\label{lambda-con}
    \sum_{i\neq j}\frac{p_i}{\lambda_{ij}}\leq 1\ {\rm for}\ j=1,...,n;\end{equation} in particular, if they are $\lambda$-interference free for some 
    $\lambda\geq 1-\min_{i=1}^n p_i,$ then we have: $$\mathit{tr}\left[\left(\sum_{i=1}^np_iA_i\right)\rho\right]\leq\sum_s\mathit{tr}\left[\left(\sum_{i=1}^np_iB_{is}\right)\rho_s\right]$$ where $$B_{is}=\begin{cases}
B_i\ &{\rm if}\ P_{is}\equiv\ \downarrow,\\ \mathit{pre}(T_{is}) &{\rm if}\ P_{is}\equiv\mathit{at}(T_{is},P_i).
\end{cases}$$
\end{enumerate}
\end{thm}

\begin{proof} See Appendix \ref{proof-strong-sound-1}. 
\end{proof}

At this moment, we are only able to conceive rule (R.PC.L) as a quantum generalisation of the rule (R.PC) for classical parallel programs with shared variables. 
In classical computing, as proved in \cite{Owicki76-0}, rule (R.PC) together with a rule for auxiliary variables and Hoare logic for sequential programs gives rise to a (relatively) complete logical system for reasoning about parallel programs with shared variables. However, it is not the case for rule (R.PC.L) in parallel quantum programming because not every (largely entangled) precondition (resp. postcondition) of $P_1\|\cdots\|P_n$ can be written in the form of $\sum_{i=1}^np_iA_i$ (resp. $\sum_{i=1}^np_iQ_i$). As will be further discussed in the Conclusion, the problem of fining a (relatively) complete proof system for shared-variable parallel quantum programs is still widely open.  

\section{Case Study: Verification of Bravyi-Gosset-K\"{o}nig's Algorithm}\label{sec-case}

Bravyi-Gosset-K\"{o}nig's algorithm \cite{Bravyi} is a parallel quantum algorithm solving a linear algebra problem, called HLF (Hidden Linear Function). This quantum algorithm runs in a constant time, and it is proved that no classical algorithms running in a constant time can solve HLF. So,  Bravyi-Gosset-K\"{o}nig's algorithm provides for the first time an unconditional proof of quantum advantage that does not rely on any complexity-theoretic conjecture. At the same time, it is suitable for experimental realisations on near-future quantum hardwares because it only requires shallow circuits with nearest-neighbour gates.  

In this section, we present a formal verification of Bravyi-Gosset-K\"onig's parallel quantum algorithm as an application of the proof system we developed in this paper. 

\subsection{Bravyi-Gosset-K\"{o}nig's Algorithm}

For convenience of the reader, we briefly review Bravyi-Gosset-K\"onig's algorithm.

\subsubsection{HLF Problem} For any symmetric Boolean matrix $A=\left(A_{ij}\right)_{n\times n}$, where $A_{ij}=A_{ji}\in\{0,1\}$, we can define a quadratic form: 
$$q_A(x)=x^TA x=\sum_{i,j}A_{ij}x_ix_j\mod 4,$$ where (and in the sequel) superscript $^T$ stands for transpose, and $x=(x_1,...,x_n)^T$ is a column vector in $\{0,1\}^n$. The null-space of $A$ is $$\mathit{Ker}(A)=\{x\in\{0,1\}^n:Ax=0 \mod 2\}.$$ It can be shown that the restriction of $q_A$ onto $\mathit{Ker}(A)$ is linear; that is, there exists $z=(z_1,...,z_n)^T\in\{0,1\}^n$ such that \begin{equation}\label{hide}
q_A(x)=2z^Tx=2\sum_iz_ix_i\mod 4
\end{equation} for all $x\in\mathit{Ker}(A)$. Thus, linear function $$l(x)=zx^T=\sum_i z_ix_i\mod 2$$ is called an HLF (Hidden Linear Function) in $q_A$. 
The general HLF problem can be stated as follows: 

{\vskip 3pt}

\textbf{HLF Problem}: Given an $n\times n$ symmetric Boolean matrix $A$, find an HLF in $q_A$, i.e. a Boolean vector $z\in\{0,1\}^n$ satisfying equation (\ref{hide}).

{\vskip 3pt}

We first present Bravyi-Gosset-K\"{o}nig's algorithm as a sequential program. Let $q_1,...,q_n$ be $n$ qubit variables and assume that self-adjacency $$\{i:A_{ii}=1\ (1\leq i\leq n)\}=\{i_1,...,i_l\}$$ and adjacency relation $$S=\left\{(j,k):A_{jk}=1\ (1\leq j<k\leq n)\right\}=\left\{\left(j_1,k_1\right),...,\left(j_m,k_m\right)\right\}.$$ 
Recall that phase shift gate $S$ and controlled-Z gate $\mathit{CZ}$ are defined by \begin{align*}S|b\>&=\ci^b|b\>\ {\rm for}\ b\in\{0,1\};\\ \mathit{CZ}|b_1,b_2\>&=(-1)^{b_1b_2}|b_1,b_2\>\ {\rm for}\ b_1,b_2\in\{0,1\},\end{align*} respectively, where (and in the sequel) we use $\ci$ to denote the imaginary unit, i.e. the square root of $-1$ (in order to avoid confusion with index $i$, which is extensively used in this paper).
The algorithm is given program $\mathit{BGK}$ in Figure \ref{HLF-algorithm-s}. 
\begin{figure}[h]\centering
\begin{align}\label{level1}
\mathit{BGK}\equiv\ &q_1 := |0\>; \cdots; q_n:=|0\>;\\
\label{level2}&q_1 := H[q_1];\cdots;\ q_n:=H[q_n]; \\
\label{level3}&q_{i_1}:=S[q_{i_1}];\cdots;q_{i_l}:=S[q_{i_l}];\\
\label{level4}&q_{j_1},q_{k_1}:=CZ[q_{j_1},q_{k_1}];\cdots;\ q_{j_m},q_{k_m}:=CZ[q_{j_m},q_{k_m}]\\ 
\label{level5}&q_1 := H[q_1];...;q_n:=H[q_n]\end{align}
\caption{Sequential Bravyi-Gosset-K\"{o}nig algorithm.}
\label{HLF-algorithm-s}
\end{figure}
We write $P_A$ for the subprogram consisting of layers (\ref{level3}) and (\ref{level4}). It can be checked that the semantic function of subprogram $P_A$ in Figure \ref{HLF-algorithm-s} is a unitary $\llbracket P_A\rrbracket =U_A$ defined by $$U_A|x\>=\ci^{q_A(x)}|x\>\ {\rm for}\ x\in\{0,1\}^n.$$ Furthermore, if $|0\>^{\otimes n}$ is input to program $\mathit{BGK}$, then it outputs \begin{align*}
\llbracket\mathit{BGK}\rrbracket\left(|0\>^{\otimes n}\right)&=H^{\otimes n} U_A H^{\otimes n}|0\>^{\otimes n} =\frac{1}{2^n}\sum_{z\in\{0,1\}^n}\alpha_z|z\>
\end{align*} where for every $z$: $$\alpha_z=\sum_{x\in\{0,1\}^n}\ci^{q_A(x)+2z^Tx}.$$ We can show that $\alpha_z\neq0$ if and only if $z$ is a solution of the HLF problem. Thus, HLF can be finally solved by measuring the above output of $\mathit{BGK}$ in the computational basis.  
 
 \subsubsection{2D HLF Problem} It is easy to see that in general, the depth of program $\mathit{BGK}$ depends on the dimension $n$ and structure of matrix $A$ and thus is not a constant.  
We hope to parallelise $\textit{BGK}$ to a constant-depth program. Obviously, each of layers (\ref{level1})-(\ref{level3}) and (\ref{level5}) can be easily parallelised into a depth-one circuit. But only for a special class of matrices $A$, layer (\ref{level4}) can be parallelised to a constant-depth program. Let $n=N^2$ for an integer $N$. We use $i=1,...,n$ to denote the vertices of the $N\times N$ square grid. Then $A$ is called a nearest-neighbourhood matrix  of the $N\times N$ grid when: 
\begin{align*}A_{ij}=1\ {\rm only\ if}\ i=j\ {\rm or}\ i,j\ {\rm are\ nearest\text{-}neighbour\ vertices\ of\ the\ gird}.
\end{align*}
Now we consider a special case of the HLF problem:
 
 {\vskip 3pt}
 
 \textbf{2D HLF Problem}: Given a square number $n=N^2$, find an HLF of $q_A$ for a nearest-neighbourhood matrix $A$ of the $N\times N$ grid. 
 
 {\vskip 3pt}

For the 2D HLF, the adjacency relation $S$ of $A$ can be covered by the following four pairwise disjoint subsets: $S\subseteq S_1\cup S_2\cup S_3\cup S_4$, where
\begin{align*}
&S_1 = \{((i-1)N+2j-1, (i-1)N+2j): 1\le i\le N, 1\le j\le \lfloor N/2\rfloor\}, \\
&S_2 = \{((i-1)N+2j, (i-1)N+2j+1): 1\le i\le N, 1\le j\le \lfloor (N-1)/2\rfloor\},\\
&S_3 = \{(2(i-1)N+j, (2i-1)N+j): 1\le i\le \lfloor N/2\rfloor, 1\le j\le N\}, \\
&S_4 = \{((2i-1)N+j, 2iN+j): 1\le i\le \lfloor (N-1)/2\rfloor, 1\le j\le N\},
\end{align*} This division is visualised in Figure \ref{HLF-net}. 
\begin{figure}[h]\centering
\includegraphics[width=11cm]{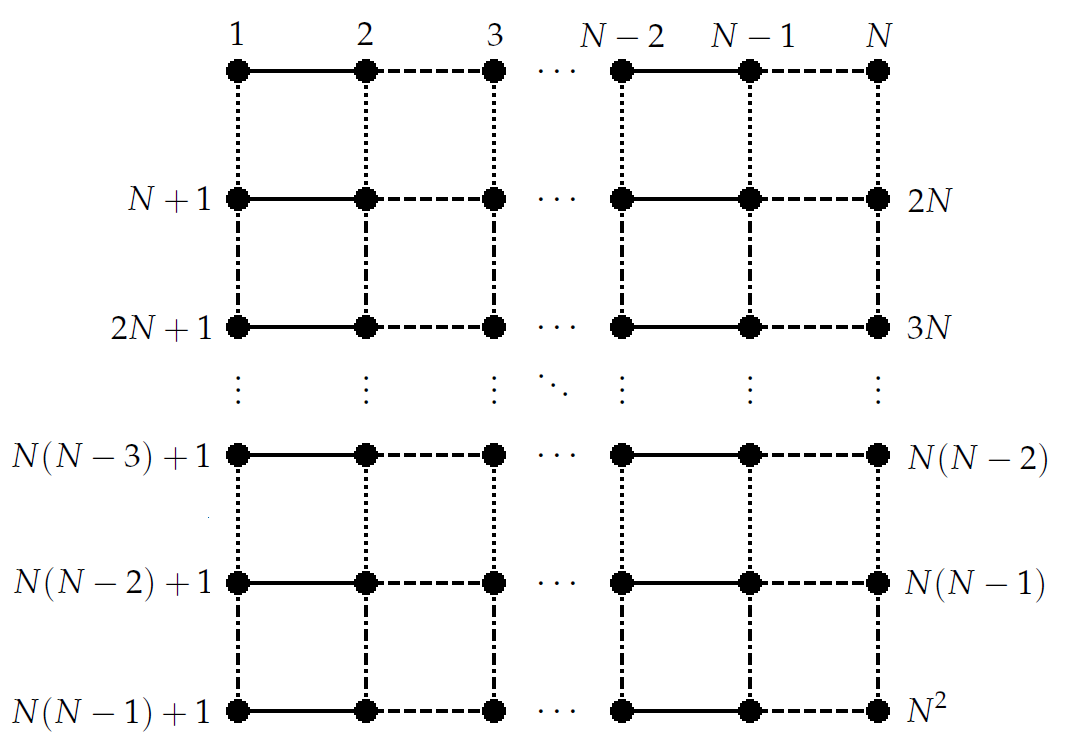}
\caption{2D HLF qubit network is a $N\times N$ grid. The vertices are numbered by 1 to $N^2$ from left to right and from top to bottom.
For even $N$, the edges in $S_1,S_2,S_3$ and $S_4$ are shown by solid line, dashed line, dotted line and dash-dotted line respectively.}
\label{HLF-net}
\end{figure}
Using the parallel quantum programming language defined in this paper, a parallelisation of $\textit{BGK}$ is presented as program $\mathit{BGK}_\|$ in Figure \ref{HLF-algorithm}. It is the sequential composition of eight subprograms with each of them being a parallel program. 
\begin{figure}[h]\centering
\begin{align*}
\mathit{BGK}_\|\equiv\ &\|_{i=1}^n\ q_i := |0\>; \\
&\|_{i=1}^n\  q_i := H[q_i]; \\
&\|_{i=1}^n\ q_i := S^{A_{i,i}}[q_i]; \\
&\|_{(i,j)\in S_1}\ q_i,q_{j} := CZ^{A_{i,j}}[q_i,q_{j}]; \\
&\|_{(i,j)\in S_2}\ q_i,q_{j} := CZ^{A_{i,j}}[q_i,q_{j}]; \\
&\|_{(i,j)\in S_3}\ q_i,q_{j} := CZ^{A_{i,j}}[q_i,q_{j}]; \\
&\|_{(i,j)\in S_4}\ q_i,q_{j} := CZ^{A_{i,j}}[q_i,q_{j}]; \\
&\|_{i=1}^n\ q_i := H[q_i]
\end{align*}
\caption{HLF algorithm. Each line indicate a layer and all gates in the same layer are separate.}
\label{HLF-algorithm}
\end{figure} Note that after such a parallelisation, $\mathit{BKG}$ is transformed to a constant-depth program because layer (\ref{level4}) is decomposed into four sublayers, each of which is a depth-one circuit.    

\subsection{Verification of $\mathit{BGK}_\|$}

Now we are going to verify $\textit{BGK}_\|$ in the proof system defined in this paper. We use $p$ to indicate the system consisting of the $n$ qubits used in $BGK_\|$. 
Then the (total) correctness of $HLF$ can be specified as the following Hoare triple:  
\begin{equation}\label{HLF-Hoare}
\vdash_\mathit{tot}\left\{I_p\right\}BGK_\|\left\{B\right\}
\end{equation}
where: $I_p=\bigotimes_{i=1}^nI_i$ is the identity operator on the state Hilbert space of $BGK_\|$ and 
$$B=\sum_{z\ {\rm is\ a\ solution\ of\ HLF}}|z\rangle\langle z|.$$ Intuitively, precondition $I_p$ is the quantum predicate representing \textquotedblleft true\textquotedblright, and postcondition $B$ is the projector onto the subspace spanned by all solutions. More precisely, as the precondition is \textquotedblleft true\textquotedblright, for any input state $\rho$ with trace one, the output $\sm{BGK_\|}(\rho)$ satisfies: $$\tr\left(B\sm{BGK_\|}(\rho)\right) = 1,$$ which implies that if we measure the output using computational basis, the outcome is just one of the solutions.

{\vskip 3pt}

\textbf{Overall Idea of the Verification}: 
Since each layer of algorithm $BGK_\|$ presented in Figure \ref{HLF-algorithm} is a disjoint parallel program, our strategy of verifying (\ref{HLF-Hoare}) is as follows: we first use parallel composition rule (R.PC.P) together with auxiliary rules (R.SO) and (R.TI) to derive a correctness formula for each layer of $BGK_\|$, and then use sequential composition rule (R.SC) to glue them in order to form a proof of a stronger correctness formula:
\begin{equation}
\label{HLF-Hoare-pure}
\vdash_\mathit{tot}\{I_p\} BGK_\| \left\{|\phi\>\<\phi|\right\} 
\end{equation}
where $|\phi\>$ is a pure state defined by
$$
|\phi\> = \frac{1}{2^n}\sum_{x,z\in\{0,1\}^n}\mathbf{i}^{x^T Ax}(-1)^{z^Tx}|z\> = \frac{1}{2^n}\sum_{\substack{\forall\ i\in[n]:\\x_i,z_i\in\{0,1\}}}\mathbf{i}^{\sum A_{i,j}x_ix_j}(-1)^{\sum z_ix_i} \bigotimes_{i\in[n]}|z_i\>_i.
$$ 
Note that $|\psi\>\<\psi|\sqsubseteq B$. Thus, (\ref{HLF-Hoare}) follows from (\ref{HLF-Hoare-pure}) and rule (R.Or).

\subsubsection{Correctness Formulas of Quantum Gates}  
Let us start from basic components. 
For each qubit $i$, we introduce an auxiliary qubit $i^\prime$. The auxiliary system consisting of qubits $i^\prime (i\in[n])$ is labeled by $p^\prime$.
First of all, using rule (Ax.UT) we obtain the following correctness formula for the quantum gates employed in $BGK_\|$:
\begin{flalign}
&&\hspace{+3cm}\label{HLF-01}\left\{I_i\right\}&& q_i &:= |0\> &&\left\{|0\>_i\<0|\right\} \hspace{+3cm}&&\\
&&\hspace{+3cm}\left\{\Psi^H_{i;i^\prime}\right\}&& q_i &:= H[q_i] &&\left\{\Phi_{i;i^\prime}\right\} \hspace{+3cm}&&\\
&&\hspace{+3cm}\left\{\Psi^S_{i;i^\prime} \right\}
&& q_i &:= S^{A_{i,i}}[q_i] &&\left\{\Phi_{i;i^\prime}\right\}\hspace{+3cm} &&\\
&&\hspace{+3cm}\label{HLF-02}\left\{\Psi_{i,j;i^\prime,j^\prime} \right\} && q_i,q_j &:= CZ^{A_{i,j}}[q_i,q_j]  &&\left\{\Phi_{i,j;i^\prime,j^\prime}\right\}\hspace{+3cm}&&
\end{flalign} where: 
\begin{align*}
\Phi_{i;i^\prime} &= \frac{1}{2}\sum_{u_i,u_i^\prime\in\{0,1\}}|u_i\>_i\<u_i^\prime|\otimes |u_i\>_{i^\prime}\<u_i^\prime|=\frac{1}{2} \left[ \sum_{u_i\in\{0,1\}}|u_i\>_i|u_i\>_{i^\prime} \right]\left[\cdot\right]^\dag\\ 
\Phi_{i,j;i^\prime,j^\prime} &= \frac{1}{4} \left[\sum_{u_i,u_j\in\{0,1\}}|u_i\>_i|u_j\>_j|u_i\>_{i^\prime}|u_j\>_{j^\prime}\right]\left[\cdot\right]^\dag\\ 
\Psi^H_{i;i^\prime} &= \sum_{v_i,v_i^\prime, u_i,u_i^\prime}\frac{1}{4}(-1)^{u_iv_i+u_i^\prime v_i^\prime}|v_i\>_i\<v_i^\prime|\otimes |u_i\>_{i^\prime}\<u_i^\prime| = \left[ \frac{1}{2}\sum_{v_i,u_i}(-1)^{u_iv_i}|v_i\>_i|u_i\>_{i^\prime}\right]\left[\cdot\right]^\dag\\ 
\Psi^S_{i;i^\prime} &= \frac{1}{2}\left[ \sum_{u_i}\mathbf{i}^{-A_{i,i}u_i}|u_i\>_i|u_i\>_{i^\prime}\right]\left[\cdot\right]^\dag\\
\Psi_{i,j;i^\prime,j^\prime} &=\frac{1}{4} \left[\sum_{u_i,u_j\in\{0,1\}}(-1)^{A_{i,j}u_iu_j}|u_i\>_i|u_j\>_j|u_i\>_{i^\prime}|u_j\>_{j^\prime}\right]\left[\cdot\right]^\dag
\end{align*} It is worth noting that $\Phi_{i,i^\prime}$ is (the unnormalized projection operator to the one-dimensional subspace spanned by) the maximal entanglement between qubits $i$ and $i^\prime$, and $\Phi_{i,j,i^\prime,j^\prime}$ is the maximal entanglement between $i,j$ and $i^\prime, j^\prime$. 

\subsubsection{Applications of Parallel Composition Rule (R.PC.P)}

Now we can apply parallel composition rule (R.PC.P) to correctness formulas (\ref{HLF-01}) - (\ref{HLF-02}), respectively, to derive a correctness formula for each layer of $BGK_\|$: 
\begin{flalign}
\label{HLF-1-1}
&&\left\{\bigotimes_{i\in[n]}I_i\right\}&& \|_{i=1}^n q_i &:= |0\> &&\left\{\bigotimes_{i\in[n]}|0\>_i\<0|\right\},\\ 
\label{HLF-2.8-1}
&&\left\{\bigotimes_{i\in[n]}\Psi^H_{i;i^\prime}\right\}&& \|_{i=1}^n\ q_i &:= H[q_i] &&\left\{\bigotimes_{i\in[n]}\Phi_{i;i^\prime}\right\},\\
\label{HLF-3-1}
&&\left\{\bigotimes_{i\in[n]}\Psi^S_{i;i^\prime}\right\} &&\|_{i=1}^n  q_i &:= S^{A_{i,i}}[q_i] &&\left\{\bigotimes_{i\in[n]}\Phi_{i;i^\prime}\right\},\\
\label{HLF-7-1}
&&\left\{\bigotimes_{(i,j)\in S_m}\Psi_{i,j;i^\prime,j^\prime}\bigotimes_{i\in T_m}\Phi_{i;i^\prime}\right\}&& \|_{(i,j)\in S_m}\ q_i,q_{j} &:= CZ^{A_{i,j}}[q_i,q_{j}] &&\left\{\bigotimes_{(i,j)\in S_m}\Phi_{i,j;i^\prime,j^\prime}\bigotimes_{i\in T_m}\Phi_{i;i^\prime}\right\},
\end{flalign}
where $m = 1,\cdots, 4$. 

\subsubsection{Applications of Auxiliary Rules (R.SO) and (R.TI)}

At this stage, we cannot directly apply rule (R.SC) to formulas (\ref{HLF-1-1}) through (\ref{HLF-7-1}) because the postcondition of each of them does not match the precondition of the next one. The auxiliary rules (R.SO) and (R.TI) can help us to resolve this issue. Let us first introduce following states:
\begin{align*}
&|\phi_m\> = \frac{1}{\sqrt{2^n}}\sum_{\forall\ i\in[n]:x_i\in\{0,1\}}\mathbf{i}^{\sum A_{i,j}x_ix_j}(-1)^{\sum_{(i,j)\in \bigcup_{l>m}S_l}A_{i,j}x_ix_j} \bigotimes_{i\in[n]}|x_i\>_i, \quad \text{for}\ m=0,1,2,3,4 \\
&|\phi_S\> = \sum_{\substack{\forall\ i\in[n]: x_i\in\{0,1\}}}\frac{1}{\sqrt{2^n}}\bigotimes_{i\in[n]}|x_i\>_i.
\end{align*}
In particular, as $\bigcup_{l>4}S_l = \emptyset$, it holds that 
$$|\phi_4\> = \frac{1}{\sqrt{2^n}}\sum_{\forall\ i\in[n]:x_i\in\{0,1\}}\mathbf{i}^{\sum A_{i,j}x_ix_j} \bigotimes_{i\in[n]}|x_i\>_i.$$
Note that $S_4\cup S_3\cup S_2 \cup S_1 = \{(i,j): \text{vertices}\ i\ \text{and}\ j\ \text{are\ adjacent}\}$. Then according to assumption that $A_{i,j} = 0$ for all $i\neq j$ and $i,j$ are not adjacent, we can simplify $|\phi_0\>$ as follows:
\begin{align*}
|\phi_0\> &= \frac{1}{\sqrt{2^n}}\sum_{\forall\ i\in[n]:x_i\in\{0,1\}}\mathbf{i}^{\sum A_{i,j}x_ix_j}(-1)^{\sum_{(i,j)\in S_4\cup S_3\cup S_2 \cup S_1}A_{i,j}x_ix_j} \bigotimes_{i\in[n]}|x_i\>_i \\
&= \frac{1}{\sqrt{2^n}}\sum_{\forall\ i\in[n]:x_i\in\{0,1\}}\mathbf{i}^{\sum A_{i,j}x_ix_j}\mathbf{i}^{\sum_{i\neq j}A_{i,j}x_ix_j} \bigotimes_{i\in[n]}|x_i\>_i \\
&= \frac{1}{\sqrt{2^n}}\sum_{\forall\ i\in[n]:x_i\in\{0,1\}}\mathbf{i}^{\sum_{i} A_{i,i}x_i} \bigotimes_{i\in[n]}|x_i\>_i
\end{align*}
because $2\sum_{i\neq j}A_{i,j}x_ix_j \mod 4 = 0$ and $x_i = x_i^2$.

Now we can construct the following quantum operations applying on system $p^\prime$ of auxiliary qubits: for any density operator $\rho$, 
\begin{align*}
\mathcal{F}(\rho) &= \sum_{\substack{\forall\ i\in[n]:\\k_i\in\{0,1\}}}|\bar{\phi}\>_{p^\prime}\left(\bigotimes_{i\in[n]}{}_{i^\prime}\<k_i|\right) \rho \left(\bigotimes_{i\in[n]}|k_i\>_{i^\prime}\right){}_{p^\prime}\<\bar{\phi}|,\\
\mathcal{F}^\prime(\rho) &= \sum_{\substack{\forall\ i\in[n]:\\k_i\in\{0,1\}}}|\bar{\phi_S}\>_{p^\prime}\left(\bigotimes_{i\in[n]}{}_{i^\prime}\<k_i|\right) \rho \left(\bigotimes_{i\in[n]}|k_i\>_{i^\prime}\right){}_{p^\prime}\<\bar{\phi_S}|,\\
\mathcal{F}_m(\rho) &= \sum_{\substack{\forall\ i\in[n]:\\k_i\in\{0,1\}}}|\bar{\phi_m}\>_{p^\prime}\left(\bigotimes_{i\in[n]}{}_{i^\prime}\<k_i|\right) \rho \left(\bigotimes_{i\in[n]}|k_i\>_{i^\prime}\right){}_{p^\prime}\<\bar{\phi_m}|,\quad \forall\ m = 1,2,3,4\\
\mathcal{F}_S(\rho) &= \sum_{\substack{\forall\ i\in[n]:\\k_i\in\{0,1\}}}|\bar{\phi_0}\>_{p^\prime}\left(\bigotimes_{i\in[n]}{}_{i^\prime}\<k_i|\right) \rho \left(\bigotimes_{i\in[n]}|k_i\>_{i^\prime}\right){}_{p^\prime}\<\bar{\phi_0}|.
\end{align*} 

Applying rule (R.SO) with the above quantum operations to correctness formulas  (\ref{HLF-2.8-1}), (\ref{HLF-3-1}), (\ref{HLF-7-1}) and (\ref{HLF-2.8-1}), respectively, we have:
\begin{flalign}&&\label{HLF-m-1}\left\{\mathcal{F}^{\prime\ast}\left(\bigotimes_{i\in[n]}\Psi^H_{i;i^\prime}\right)\right\}&& \|_{i=1}^nq_i &:= H[q_i] &&\left\{\mathcal{F}^{\prime\ast}\left(\bigotimes_{i\in[n]}\Phi_{i;i^\prime}\right)\right\},  \\
&&\left\{\mathcal{F}_S^\ast\left(\bigotimes_{i\in[n]}\Psi^S_{i;i^\prime}\right)\right\}&& \|_{i=1}^nq_i &:= S^{A_{i,i}}[q_i] &&\left\{\mathcal{F}_S^\ast\left(\bigotimes_{i\in[n]}\Phi_{i;i^\prime}\right)\right\},\\ 
&&\left\{\mathcal{F}_m^\ast\left(\bigotimes_{(i,j)\in S_m}\Psi_{i,j;i^\prime,j^\prime}\bigotimes_{i\in T_m}\Phi_{i;i^\prime}\right)\right\}&& \|_{(i,j)\in S_m}q_i,q_{j} &:= CZ^{A_{i,j}}[q_i,q_{j}] &&\left\{\mathcal{F}_m^\ast\left(\bigotimes_{(i,j)\in S_m}\Phi_{i,j;i^\prime,j^\prime}\bigotimes_{i\in T_m}\Phi_{i;i^\prime}\right)\right\},\nonumber\\
&& && & && \qquad\quad \forall\ m = 1,2,3,4 \\ 
&&\label{HLF-m-2}\left\{\mathcal{F}^\ast\left(\bigotimes_{i\in[n]}\Psi^H_{i;i^\prime}\right)\right\}&& \|_{i=1}^nq_i &:= H[q_i] &&\left\{\mathcal{F}^\ast\left(\bigotimes_{i\in[n]}\Phi_{i;i^\prime}\right)\right\},
\end{flalign}

The preconditions and postconditions of the above correctness formulas are too complicated. Their simplifications are given in the following: 

\begin{lem}\label{super-lemma-t}
\begin{align*}
&\mathcal{F}^\ast\left(\bigotimes_{i\in[n]}\Phi_{i;i^\prime}\right) = \frac{|\phi\>_p\<\phi|}{2^n}\otimes I_{p^\prime}, \qquad 
&&\hspace{-1cm}\mathcal{F}^\ast\left(\bigotimes_{i\in[n]}\Psi^H_{i;i^\prime}\right) = \frac{|\phi_4\>_p\<\phi_4|}{2^n}\otimes I_{p^\prime},\\ 
&\mathcal{F}^{\prime\ast}\left(\bigotimes_{i\in[n]}\Phi_{i;i^\prime}\right) = \frac{|\phi_S\>_p\<\phi_S|}{2^n}\otimes I_{p^\prime}, 
&&\hspace{-1cm}\mathcal{F}^{\prime\ast}\left(\bigotimes_{i\in[n]}\Psi^H_{i;i^\prime}\right) = \frac{|0\>_p\<0|}{2^n}\otimes I_{p^\prime},\\
&\mathcal{F}_S^\ast\left(\bigotimes_{i\in[n]}\Phi_{i;i^\prime}\right) = \frac{|\phi_0\>_p\<\phi_0|}{2^n}\otimes I_{p^\prime} \qquad 
&&\hspace{-1cm}\mathcal{F}_S^\ast\left(\bigotimes_{i\in[n]}\Psi^S_{i;i^\prime}\right)
= \frac{|\phi_S\>_p\<\phi_S|}{2^n}\otimes I_{p^\prime},\\
&\mathcal{F}_m^\ast\left(\bigotimes_{(i,j)\in S_m}\Phi_{i,j;i^\prime,j^\prime}\bigotimes_{i\in T_m}\Phi_{i;i^\prime}\right) = \frac{|\phi_m\>_p\<\phi_m|}{2^n}\otimes I_{p^\prime},\qquad
&&\hspace{-1cm}\mathcal{F}_m^\ast\left(\bigotimes_{(i,j)\in S_m}\Psi_{i,j;i^\prime,j^\prime}\bigotimes_{i\in T_m}\Phi_{i;i^\prime}\right) = \frac{|\phi_{m-1}\>_p\<\phi_{m-1}|}{2^n}\otimes I_{p^\prime}
\end{align*}
where $m = 1,2,3,4$.
\end{lem}

\begin{proof} See Appendix \ref{super-lem-proof}.\end{proof}

With the above lemma, correctness formulas (\ref{HLF-m-1}) - (\ref{HLF-m-2}) can be simplified as follows after applying (R.Lin): 
\begin{align}
&&\label{HLF-mm-1}\left\{|\phi_4\>_p\<\phi_4|\otimes I_{p^\prime}\right\}\hspace{-0.2cm}&& \|_{i=1}^nq_i &:= H[q_i] &&\hspace{-0.2cm} \left\{|\phi\>_p\<\phi|\otimes I_{p^\prime}\right\},\\
\forall\ m = 1,2,3,4 \hspace{-0.2cm}&&\left\{|\phi_{m-1}\>_p\<\phi_{m-1}|\otimes I_{p^\prime}\right\}\hspace{-0.2cm}&& \|_{(i,j)\in S_m}q_i,q_{j} &:= CZ^{A_{i,j}}[q_i,q_{j}] &&\hspace{-0.2cm}\left\{|\phi_m\>_p\<\phi_m|\otimes I_{p^\prime}\right\},\\
&&\left\{|\phi_S\>_p\<\phi_S|\otimes I_{p^\prime}\right\}\hspace{-0.2cm}&& \|_{i=1}^nq_i &:= S^{A_{i,i}}[q_i] &&\hspace{-0.2cm}\left\{|\phi_0\>_p\<\phi_0|\otimes I_{p^\prime}\right\},\\ 
&&\label{HLF-mm-2}\left\{|0\>_p\<0|\otimes I_{p^\prime}\right\}\hspace{-0.2cm}&& \|_{i=1}^n q_i &:= H[q_i] &&\hspace{-0.2cm}\left\{|\phi_S\>_p\<\phi_S|\otimes I_{p^\prime}\right\}.\end{align}

Now by applying rule (R.TI) to (\ref{HLF-mm-1}) - (\ref{HLF-mm-2}), we obtain:
\begin{align}
\label{HLF-8}
&&\left\{|\phi_4\>_p\<\phi_4|\right\}&& \|_{i=1}^n q_i &:= H[q_i] &&\left\{|\phi\>_p\<\phi|\right\},\\ 
\label{HLF-7}
\hspace{0.7cm}\forall\ m = 1,2,3,4 &&\left\{|\phi_{m-1}\>_p\<\phi_{m-1}|\right\}&& \|_{(i,j)\in S_m}q_i,q_{j} &:= CZ^{A_{i,j}}[q_i,q_{j}] &&\left\{|\phi_m\>_p\<\phi_m|\right\},\hspace{0.7cm}\\ 
\label{HLF-3}
&&\left\{|\phi_S\>_p\<\phi_S|\right\}&& \|_{i=1}^nq_i &:= S^{A_{i,i}}[q_i] &&\left\{|\phi_0\>_p\<\phi_0|\right\},\\ 
\label{HLF-2}
&&\left\{|0\>_p\<0|\right\}&& \|_{i=1}^nq_i &:= H[q_i] &&\left\{|\phi_S\>_p\<\phi_S|\right\}.
\end{align} 
Finally, we use rule (R.SC') to combine formulae (\ref{HLF-1-1}, \ref{HLF-2},\ref{HLF-3},\ref{HLF-7},\ref{HLF-8}) and obtain a complete proof of 
$
\left\{I_p\right\}BGK_\|\left\{|\phi\>_p\<\phi|\right\}
$
as shown in Figure \ref{HLF-proof-outline}.

\begin{figure}[h]\centering
\begin{align*}
&\left\{I_p\right\} \\
&\|_{i=1}^n\ q_i := |0\>; \\
&\left\{|0\>_p\<0|\right\} \\
&\|_{i=1}^n\ q_i := H[q_i] \\
&\left\{|\phi_S\>_p\<\phi_S|\right\} \\
&\|_{i=1}^n\ q_i := S^{A_{i,i}}[q_i] \\
&\left\{|\phi_0\>_p\<\phi_0|\right\} \\
&\|_{(i,j)\in S_1}\ q_i,q_{j} := CZ^{A_{i,j}}[q_i,q_{j}];  \\
&\left\{|\phi_1\>_p\<\phi_1|\right\} \\
&\|_{(i,j)\in S_2}\ q_i,q_{j} := CZ^{A_{i,j}}[q_i,q_{j}]; \\
&\left\{|\phi_2\>_p\<\phi_2|\right\} \\
&\|_{(i,j)\in S_3}\ q_i,q_{j} := CZ^{A_{i,j}}[q_i,q_{j}]; \\
&\left\{|\phi_3\>_p\<\phi_3|\right\} \\
&\|_{(i,j)\in S_4}\ q_i,q_{j} := CZ^{A_{i,j}}[q_i,q_{j}]; \\
&\left\{|\phi_4\>_p\<\phi_4|\right\} \\
&\|_{i=1}^n\  q_i := H[q_i] \\
&\left\{|\phi\>_p\<\phi|\right\}
\end{align*}
\caption{Proof outline for HLF algorithm.}
\label{HLF-proof-outline}
\end{figure}

\section{Conclusion}\label{sec-con}

This paper initiates the study of parallel quantum programming; more explicitly, it defines operational and denotational semantics of parallel quantum programs and presents several useful inference rules for reasoning about correctness of parallel quantum programs. In particular, it is proved that our inference rules form a (relatively) complete proof system for disjoint parallel quantum programs. However, this is certainly merely one of the first steps toward a comprehensive theory of parallel quantum programming and leaves a series of fundamental problems unsolved.

{\vskip 4pt}

\textbf{\textit{1. Completeness}}:  
Perhaps, the most important and difficult open problem at this stage is to develop a (relatively) \textit{complete} logical system for verification of parallel quantum programs with shared variables. 
\begin{itemize}\item \textit{Stronger Rule for Parallel Composition}: As pointed out in Section \ref{correct-shared}, inference rule (R.PC.L) can be used to prove some useful correctness properties of such quantum programs, but it seems far from being the rule for parallel composition needed in a (relatively) complete logical system for these quantum programs. A possible candidate for the rule that we are seeking is based on the notions of \textit{join} and \textit{margin} of operators: let $\mathcal{H}=\bigotimes_{i=1}^n\mathcal{H}_i$ and $\mathcal{J}$ be a family of subsets of $\{1,...,n\}$. For each $J\in\mathcal{J}$, given a positive operator $A_J$ in $\mathcal{H}_J=\bigotimes_{j\in J}\mathcal{H}_j$.
If positive operator $A$ in $\mathcal{H}$ satisfies: $A_J=\mathit{tr}_{J^c}A$ for every $J\in\mathcal{J}$, where $J^c=\{1,...,n\}\setminus J$, then $A$ is called a join of $\{A_J\}_{J\in\mathcal{J}}$, and each $A_J$ is called the margin of $A$ in $\mathcal{H}_J$. With the notion of join, we can conceive that the inference rule needed for parallel composition of quantum programs with shared variables should be some variant of rule (R.PC.J) given in Figure \ref{fig 3+join}.
\begin{figure}[h]\centering
\begin{align*}({\rm R.PC.J})\ \ \ \ \ \ \ \ \ \ \ \ \frac{\begin{array}{ccc}{\rm Standard\ proof\ outlines}\ \left\{A_i\right\}P^\ast_i\left\{B_i\right\} (i=1,...,n)\ {\rm are}\ \Lambda{\rm -interference\ free}\\
A\ {\rm is\ a\ join\ of}\ \{A_i\}, {\rm and}\ B\ {\rm is\ a\ join\ of}\ \{B_i\}
\end{array}}{\{A\}P_1\|\cdots\|P_n\{B\}}\ \ \ \ \ \ \ \ \ \ \ \ \ \ \ \ \ \ \ \ \ \
\end{align*}
\caption{Rule for Parallel Quantum Programs with Shared Variables.}\label{fig 3+join}
\end{figure}

{\vskip 3pt}

\item \textit{Auxiliary Variables}: As is well-known in the theory of classical parallel programming (see \cite{Apt09}, Chapters 7 and 8, and \cite{Francez}, Chapter 7), to achieve a (relatively) complete logical system for reasoning about parallel programs, except finding a strong enough  rule for parallel composition, one must introduce \textit{auxiliary variables} to record the control flow of a program, which, at the same time, should not influence the control flow inside the program.
We presented several rules in Subsection \ref{sub-auxiliary-rules} for introducing auxiliary variables, and they were employed to establish (relative) completeness of our proof system for disjoint parallel quantum programs. However, there they were used to deal with entanglement and not for recording control flows.
It seems that auxiliary variables recording control flows are also needed in parallel quantum programming. At this moment, however, we do not have a clear idea about how such auxiliary variables can be introduced in the case of parallel quantum programs with shared variables.

{\vskip 3pt}

\item \textit{Infinite-Dimension}: The issue of infinite-dimensional state Hilbert spaces naturally arises when developing a logical system for parallel quantum programs with infinite data types like integers and reals. As we can see in Subsection \ref{dis-completeness} (and Apeendices G and H), this issue was properly resolved with auxiliary rule (R.Lim) defined in terms of weak convergence of operators in the case of disjoint parallel quantum programs. But it is still unknown whether the same idea works or not for shared variables; in particular, how it can be used in combination with a parallel composition rule like (R.PC.J) considered above.   
\end{itemize}

It seems that a \textit{full} solution to the above three issues and achieving a (relatively) complete proof system for parallel quantum programs are still far beyond the current reach.

{\vskip 4pt} 

\textbf{\textit{2. Mechanisation}}: A theorem prover for quantum Hoare logic was implemented in Isabelle/HOL for verification of quantum \textbf{while}-programs \cite{Liu19}.  
We plan to further formalise the syntax, semantics and proof rules presented in this paper and to extend the theorem prover so that it can be used for verification of parallel quantum programs. 
Mechanisation of the current proof rules seems feasible. In the future, if we are able to find a stronger rule of the form (R.PC.J) discuused above, implementing an automatic tool for verification of parallel quantum programs based on it will be difficult and even rely on a breakthrough in finding an algorithmic solution to the following long-standing open problem (listed in \cite{Stil95} as one of the ten most prominent mathematical challenges in quantum chemistry; see also \cite{Kly06}) --- \textit{Quantum Marginal Problem}: given a family $\mathcal{J}$ of subsets of $\{1,...,n\}$, and for each $J\in\mathcal{J}$, given a density operator (mixed state) $\rho_J$ in $\mathcal{H}_J$. Is there a join (global state) of $\{\rho_j\}_{J\in\mathcal{J}}$ in $\mathcal{H}$?

{\vskip 4pt}

\textbf{\textit{3. Applications}}:  As pointed out in \cite{Martin18}, parallelism at various levels will be an important consideration for quantum computing; in particular, proper architectural support for parallel implementation of quantum gates may be pivotal for harnessing the power of NISQ (Noisy Intermediate Scale Quantum) devices. Our target applications of the results obtained in this paper are of course verification of parallel quantum programs and perhaps also reasoning about concurrency in operating systems of quantum computers. On the other hand, as suggested in Section 8.8 of \cite{Ying16}, some ideas in quantum programming can be applied to quantum physics. Along this line, it would be interesting to see whether our results can also be used for reasoning about many-body quantum systems (see Remark \ref{remark-loc-ham} for a brief discussion about a link between our parallel composition rule (R.PC.L) and local Hamiltonians). 

{\vskip 4pt}

\textbf{\textit{4. Extensions}}: As a first step in the studies of parallel quantum programming, this paper tried to generalise the Owicki-Gries and Lamport method. An interesting problem for future research is how to extend moderner verification techniques  beyond the Owicki-Gries and Lamport paradigm for parallel quantum programs; for example: 
\begin{itemize}
\item \textit{Compositional Techniques}: The verification technique presented in this paper is non-compositional as the Owicki-Gries and Lamport method. It is desirable to develop some compositional verification techniques that can reduce verification of a large program to independent verification of its subprograms for parallel quantum programs, e.g. quantum extension of Jone's rely-guarantee paradigm for shared variable parallelism \cite{Jones83} and 
Misra and Chandy's assumption-commitment paradigm for synchronous message passing \cite{CM88}. 

{\vskip 3pt} 

\item \textit{Separation Logic and Modular Reasoning}: Concurrent separation logic \cite{Peter19, Stephen16} is a modern logic for reasoning about parallelism and concurrency. One of its central idea is to use separating conjunctions $\star_{i=1}^nA_i, \star_{i=1}^nB_i$ of preconditions and postconditions to replace the ordinary conjunctions $\bigwedge_{i=1}^nA_i, \bigwedge_{i=1}^nB_i$ in the parallel composition rule (R.PC). In particular, the new parallel composition rule with separating conjunctions supports modular reasoning about threads and processes. It will be a great challenge to realise this idea in presence of quantum correlations that are fundamentally different from their classical counterparts. Indeed, we are even not sure this is possible or not. 

\item \textit{Message Passing}: Shared variables and message passing are two major mechanics of process interaction in parallel programming. This paper focuses on the model of parallel quantum programming with shared variables. Parallel quantum programming through message passing has been studied in \cite{JL04, GN05, Feng11} using the process algebra approach. How can we develop a proof system of the Hoare-style for parallel (or distributed) quantum programs with message passing?

\item \textit{Reasoning about Weak Memory Models}: The memory model for parallelism of quantum programs is the same as in the original Owicki-Gries and Lamport method, namely sequential consistency. Recently, the Owicki-Gries and Lamport method has been generalised to deal with various weak memory models; see for example \cite{La-Va}. How to define and reason about parallel quantum programs with weak memory models?
\end{itemize}

\newpage

\appendix

\section{Basic Properties of Operators in Hilbert Spaces}

For convenience of the reader, we review some notions of operators and super-operators and their basic properties that will be used in the remaining parts of the Appendices.

\subsection{L\"{o}wner Order between Operators}

The L\"{o}wner order is extensively used in the theory of quantum programming. Here, we list two of its properties needed in the proofs of our results.

\begin{lem}\label{lem-lowner} Let $A, B$ be observables (i.e. Hermitian operators) in Hilbert space $\mathcal{H}$. Then $A\sqsubseteq B$ if and only if for all density operators in $\mathcal{H}$: $$\mathit{tr}(A\rho)\leq\mathit{tr}(B\rho).$$
\end{lem}

\begin{lem}\label{lem-tensor} \begin{enumerate}\item If $A_1,A_2$ are positive operators in $\mathcal{H}_1$ and $\mathcal{H}_2$, respectively, then $A_1\otimes A_2$ is a positive operator in $\mathcal{H}_1\otimes\mathcal{H}_2$.
\item For any operators $A_1,B_1$ in $\mathcal{H}_1$ and $A_2, B_2$ in $\mathcal{H}_2$, $A_1\sqsubseteq B_1$ and $A_2\sqsubseteq B_2$ implies $A_1\otimes A_2\sqsubseteq B_1\otimes B_2.$
\end{enumerate}
\end{lem}

\subsection{Convergence of Operators}

We need the notions of weak and strong convergence when the state Hilbert space of a quantum program is infinite-dimensional. The following lemmas will be needed in the proofs of Theorems 4.2 and 4.3 (see Appendices G and H).     

\begin{defn} \cite{Reed} \label{wot} Let $\{A_n\}$ be a sequence of operators on a separable Hilbert space $\hs$. Then: 
\begin{enumerate}
\item $\{A_n\}$ weakly converges to an operator $A$, written:
$A_n\overset{w.o.t.}{\longrightarrow}A,$
if for all $|\psi\>,|\phi\>\in\mathcal{H}$,
$$\lim_{n\rightarrow\infty}\<\psi|A_n|\phi\>\> = \<\psi|A|\phi\>\>.$$
\item $\{A_n\}$ strongly converges to an operator $A$, written: $A_n\overset{s.o.t.}{\longrightarrow}A,$
if for all $|\psi\>\in\mathcal{H}$,
$$\lim_{n\rightarrow\infty}\|(A_n-A)|\psi\>\| = 0.$$
\end{enumerate}
\end{defn}

\begin{lem} \cite{Reed}
\label{fact-sot-wot} $A_n\overset{s.o.t.}{\longrightarrow}A$ implies $A_n\overset{w.o.t.}{\longrightarrow}A$.
\end{lem}

The next lemma shows that trace is continuous with respect to weak convergence. 

\begin{lem}
\label{fact-alter-def-wot} For quantum predicates $\{A_n\}$ and $A$, $A_n\overset{w.o.t.}{\longrightarrow}A$ if and only if for all $\rho\in\mathcal{D}(\hs)$,
\begin{equation}\label{exp-limit}\lim_{n\rightarrow \infty}\tr(A_n\rho) = \tr(A\rho).\end{equation}
\end{lem}
\begin{proof}
($\Leftarrow$) For $|\psi\>\in\hs$, we have $\rho=|\psi\>\<\psi|/\||\psi\>\|^2\in\mathcal{D}(\hs)$ and it follows from (\ref{exp-limit}) that 
$$\lim_{n\rightarrow\infty}\tr\left((A_n-A)\frac{|\psi\>\<\psi|}{\||\psi\>\|^2}\right) = 0
\quad\Rightarrow\quad\lim_{n\rightarrow\infty}\<A_n|\psi\>,|\psi\>\> = \<A|\psi\>,|\psi\>\>,
$$
which directly leads to the definition of weak operator convergence.

($\Rightarrow$) For any $\epsilon>0$ and $\rho\in\mathcal{D}(\hs)$, we assume the spectral decomposition: $$\rho=\sum_{i=1}^\infty\lambda_i|\psi_i\>\<\psi_i|.$$ Then there exists integer $N$ such that for all $k\ge N$,
$$
\tr\left(\sum_{i\ge N}\lambda_i|\psi_i\>\<\psi_i|\right)=\sum_{i\ge N}\lambda_i\le\epsilon
$$ because $\sum_{i=1}^\infty\lambda\le 1$. 
On the other hand, as $A_n\overset{w.o.t.}{\longrightarrow}A$, there exists integer $M$ such that for all $m\ge M$:
$$
\forall\ i<N,\quad |\tr(A_m|\psi_i\>\<\psi_i|) - \tr(A|\psi_i\>\<\psi_i|)|\le\epsilon.
$$ Since $A_m,A$ are quantum predicates, i.e. $0\sqsubseteq A_m,A\sqsubseteq I$, we have $\|A_m-A\|_1\le 2$. Thus by the Cauchy-Schwarz inequality, we obtain:
\begin{align*}
|\tr(A_m\rho) - \tr(A\rho)| &=
\left|\tr\left((A_m-A)\sum_{i<N}\lambda_i|\psi_i\>\<\psi_i|\right) + \tr\left((A_m-A)\sum_{i\ge N}\lambda_i|\psi_i\>\<\psi_i|\right) \right| \\
&\le \sum_{i<N}\lambda_i\max_{i<N}|\tr\left((A_m-A)|\psi_i\>\<\psi_i|\right)|
+ \|A_m-A\|_1\tr\left(\sum_{i\ge N}\lambda_i|\psi_i\>\<\psi_i|\right) \\
&\le \epsilon + 2\epsilon = 3\epsilon
\end{align*}
\end{proof}

The following lemma shows a compatibility between the L\"{o}wner order and strong convergence. 

\begin{lem}
\label{fact-wot-cpo}
For an increasing (respectively, decreasing) sequence $\{A_n\}$ of quantum predicates with respect to L\"owner order, 
$A = \bigsqcup_{n=0}^\infty A_n$ (respectively, $\bigsqcap_{n=0}^\infty A_n$) exists and $A_n$ strongly (and therefore weakly) converges to $A$: $A_n\overset{s.o.t.}{\longrightarrow}A$ 
\end{lem}
\begin{proof}
The existence of $A$ is guaranteed by the fact that the set of quantum predicates together with L\"owner order is a complete partial order (CPO). 
Furthermore, for any $|\psi\>\in\hs$, we have
$\lim_{n\rightarrow \infty} \|A_n|\psi\> - A|\psi\>\|=0$, so $A_n\overset{s.o.t.}{\longrightarrow}A$ (see \cite{Ying16}, page 100). 
\end{proof}

\subsection{Duality between Quantum Operations}

The proofs of several results in this paper require to exploit duality between quantum operations. 

\begin{defn}\label{def-dual} Let $\mathcal{E}$ be a quantum operation (i.e. super-operator) in Hilbert space $\mathcal{H}$ with the Kraus representation $\mathcal{E}=\sum_i E_i\circ E_i^\dag$. Then its (Schr\"{o}dinger-Heisenberg) dual is the super-operator $\mathcal{E}^\ast$ defined by $$\mathcal{E}^\ast(A)=\sum_i E_i^\dag AE_i$$ for any observable $A$ in $\mathcal{H}$.
\end{defn}

The following lemma the connection between a quantum operation and its dual can be given in terms of trace. 

\begin{lem}\label{lem-dual}For any quantum operation $\mathcal{E}$, observable $A$ and density operator $\rho$ in $\mathcal{H}$, we have: $$\mathit{tr}(A\mathcal{E}(\rho))=\mathit{tr}(\mathcal{E}^\ast(A)\rho).$$ In particular, it holds that $$\mathit{tr}(\mathcal{E}(\rho))=\mathit{tr}(\mathcal{E}^\ast(I)\rho),$$ where $I$ is the identity operator in $\mathcal{H}$.
\end{lem}

The next lemma shows that the dual of a quantum operation is continuous with respect to weak convergence of operators. 

\begin{lem}
\label{fact-wot-so}
For quantum predicates $\{A_n\}, A$ and quantum operation $\mathcal{E}$, if $A_n\overset{w.o.t.}{\longrightarrow}A$, then $\mathcal{E}^\ast(A_n)\overset{w.o.t.}{\longrightarrow}\mathcal{E}^\ast(A)$.
\end{lem}
\begin{proof}
For any $\rho\in\mathcal{D}(\hs)$, with Lemma \ref{fact-alter-def-wot} we observe:
\begin{align*}
\lim_{n\rightarrow\infty}\tr(\mathcal{E}^\ast(A_n)\rho) =
\lim_{n\rightarrow\infty}\tr(A_n\mathcal{E}(\rho)) 
= \tr(A\mathcal{E}(\rho)) 
= \tr(\mathcal{E}^\ast(A)\rho)
\end{align*}
which implies
$\mathcal{E}^\ast(A_n)\overset{w.o.t.}{\longrightarrow}\mathcal{E}^\ast(A).$
\end{proof}

\section{Proof of Lemmas Lemma 2.1 and 2.2}\label{proof-SCC}

Before proving Lemmas \ref{Aux-Sound} and \ref{SCC}, we present a very useful technical lemma which restates the defining inequalities for total and partial correctness (see Definition \ref{correctness-interpretation}) in a form of L\"{o}wner order.

\begin{lem}\label{lem-correctness} \begin{enumerate}\item $\models_{\mathit{tot}}\{A\}P\{B\}$ if and only if $A\sqsubseteq\llbracket P\rrbracket^\ast(B)$.
\item $\models_{\mathit{par}}\{A\}P\{B\}$ if and only if $$A\sqsubseteq\llbracket P\rrbracket^\ast(B)+(I-\llbracket P\rrbracket^\ast(I)),$$ where $I$ is the identity operator in $\mathcal{H}_P$.
\end{enumerate}
\end{lem}

\begin{proof} Immediate from Lemma 4.1.2 in \cite{Ying16} and Definition \ref{def-dual} and Lemma \ref{lem-dual}.
\end{proof}

The conditions for total and partial correctness given in the above lemma are often easier to manipulate than their defining inequalities in Definition \ref{correctness-interpretation} because the universal quantifier over density operator $\rho$ is eliminated.

\begin{proof}[Proof of Lemma \ref{Aux-Sound}] This lemma was proved in \cite{Ying18} except that rule (R.Lim) is strengthened with weak operator convergence. So, we are going to prove soundness of (R.Lim). 
We only consider the case of partial correctness (the case of total correctness is similar).  
For any $\rho\in\mathcal{D}(\hs)$ and for each $n$, according to the assumption $\models_{par}\{A_n\}P\{B_n\}$, we have: 
$$\tr(A_n\rho)\le \tr(B_n)\llbracket P\rrbracket(\rho)) + [\tr(\rho) - \tr(\llbracket P\rrbracket(\rho))],$$
and we can take $n\rightarrow\infty$ to obtain:
$$\lim_{n\rightarrow\infty}\tr(A_n\rho)\le \lim_{n\rightarrow\infty}\tr(B_n\llbracket P\rrbracket(\rho)) + [\tr(\rho) - \tr(\llbracket P\rrbracket(\rho))].$$
On the other hand, it is assumed that 
$A_n\overset{w.o.t.}{\longrightarrow}A$ and $B_n\overset{w.o.t.}{\longrightarrow}B.$ So, 
it follows from Lemma \ref{fact-alter-def-wot} that
$$\tr(A\rho)\le \tr(B\llbracket P\rrbracket(\rho)) + \tr(\rho) - \tr(\llbracket P\rrbracket (\rho)).$$
Thus, $\models_{par}\{A\}P\{B\}$.
\end{proof}

\begin{proof}[Proof of Lemma \ref{SCC}]
$\bullet\ $(R.CC1) For each $i$, it follows from Lemma \ref{lem-correctness} and the assumptions $\models_{par}\{A_i\}P\{B_i\}$ and $P:{\rm Abort}(A)$; i.e. $\models_\mathit{par}\{A\}P\{0\}$ that
$$A_i\sqsubseteq\llbracket P\rrbracket^\ast(B_i)+(I-\llbracket P\rrbracket^\ast(I)),\qquad
A\sqsubseteq I-\llbracket P\rrbracket^\ast(I).$$
As $p_i\ge0$ and $\sum_ip_i\le1$, it immediately follows that 
\begin{align*}
\sum_ip_iA_i + \left(1-\sum_ip_i\right)B&\sqsubseteq
\sum_ip_i\llbracket P\rrbracket^\ast(B_i)+\sum_ip_i(I-\llbracket P\rrbracket^\ast(I))+\left(1-\sum_ip_i\right)(I-\llbracket P\rrbracket^\ast(I))\\
&=\llbracket P\rrbracket^\ast\left(\sum_ip_iB_i\right) + \left[I-\llbracket P\rrbracket^\ast(I)\right]
\end{align*}
which, together with Lemma \ref{lem-correctness}, implies
$$\models_\mathit{par}\left\{\sum_ip_iA_i+\left(1-\sum_ip_i\right)A\right\}P\left\{\sum_ip_iB_i\right\}.$$

$\bullet\ $(R.CC2) For each $i$, it follows from the assumption $\models_{par}\{A_i\}P\{B_i\}$ that
\begin{equation}\label{pcc-1}A_i\sqsubseteq^\ast(B_i)+(I-\llbracket P\rrbracket^\ast(I)).\end{equation}
Another assumption $\models P:{\rm NTerm}(A)$ ensures that:
\begin{equation}\label{pcc-2}I-\llbracket P\rrbracket^\ast(I)\sqsubseteq A.\end{equation}
Since $\lambda_i\ge0$ and $\sum_i\lambda_i\ge1$, we can combine (\ref{pcc-1}) and (\ref{pcc-2}) to obtain: 
\begin{align*}
\sum_i\lambda_iA_i - \left(\sum_i\lambda_i-1\right)A&\sqsubseteq
\sum_i\lambda_i\llbracket P\rrbracket^\ast(B_i)+\sum_i\lambda_i(I-\llbracket P\rrbracket^\ast(I))-\left(\sum_i\lambda_i-1\right)(I-\llbracket P\rrbracket^\ast(I))\\
&= \llbracket P\rrbracket^\ast\left(\sum_i\lambda_iB_i\right)+\left[I-\llbracket P\rrbracket^\ast(I)\right]
\end{align*}
which implies
$$\models_\mathit{par}\left\{\sum_{i=1}^m\lambda_iA_i - \left(\sum_{i=1}^m\lambda_i-1\right)A\right\}P\left\{\sum_{i=1}^m\lambda_iB_i\right\}.$$
\end{proof} 

\section{Proof of Lemma 3.2}\label{proof-lem-det}

We first prove a diamond property for disjoint parallel quantum programs.

\begin{lem}\label{diamond} The operational semantics of disjoint parallel quantum programs enjoys the diamond property: \begin{itemize}\item if $\mathcal{A}\rightarrow\mathcal{A}_1$, $\mathcal{A}\rightarrow\mathcal{A}_2$ and $\mathcal{A}_1\neq\mathcal{A}_2$, then there exists $\mathcal{B}$ such that $\mathcal{A}_1\rightarrow\mathcal{B}$ and $\mathcal{A}_2\rightarrow\mathcal{B}$.\end{itemize}
\end{lem}
\begin{proof} Assume that $\mathcal{A}$ comes from certain transitions of a parallel composition of $n$ programs. Then all programs in $\mathcal{A}$ are parallel compositions of $n$ programs (exept the terminating ones $\downarrow$).
Suppose that $\mathcal{A}\rightarrow\mathcal{A}_1$ results from a transition of the $i$th component of some $$P\equiv\langle P_1\|\cdots\|P_i\|\cdots\|P_n,\rho\rangle\in\mathcal{A};$$ that is, $\langle P_i,\rho\rangle\rightarrow\{|\langle Q_{ik},\sigma_k\rangle|\}$ and $\mathcal{A}_1=(\mathcal{A}\setminus\{P\})\cup\mathcal{B}_1,$ where: $$\mathcal{B}_1=\{|\langle P_1\|\cdots\|P_{i-1}\|Q_{ik}\|P_{i+1}\|\cdots\|P_n,\sigma_k\rangle|\}.$$ Note that here, only rules (PC) and (MS1) are used to derive the transition. The conclusion can be easily generalised to the case where rule (MS2) is employed. Also, suppose that $\mathcal{A}\rightarrow\mathcal{A}_2$ results from a transition of the $j$th component of some $$R\equiv\langle R_1\|\cdots\|R_i\|\cdots\|R_n,\delta\rangle\in\mathcal{A};$$ that is, $\langle R_j,\delta\rangle\rightarrow\{|\langle S_{jl},\theta_l\rangle|\}$ and $\mathcal{A}_2=(\mathcal{A}\setminus\{R\})\cup\mathcal{B}_2,$ where: $$\mathcal{B}_2=\{|\langle R_1\|\cdots\|R_{j-1}\|S_{jl}\|R_{j+1}\|\cdots\|R_n,\theta_l\rangle|\}.$$

\textbf{Case 1}. $P$ and $R$ are two different elements of multi-set $\mathcal{A}$. Put $$\mathcal{B}=(\mathcal{A}\setminus\{P,R\})\cup\mathcal{B}_1\cup\mathcal{B}_2.$$ Then it is easy to derive $\mathcal{A}_1\rightarrow\mathcal{B}$ and $\mathcal{A}_2\rightarrow\mathcal{B}$ by transitional rules (PC) and (MS).

{\vskip 4pt}

\textbf{Case 2}. $P$ and $R$ are the same element of multi-set $\mathcal{A}$. Then it must be that $i\neq j$ because $\mathcal{A}_1\neq\mathcal{A}_2$.
Note that each element of $\mathcal{B}_1$ is of the form $$\langle P_1\|\cdots\|P_{i-1}\|Q_{ik}\|P_{i+1}\|\cdots\|P_n,\sigma_k\rangle,$$ and each element of $\mathcal{B}_2$ is of the form $$\langle P_1\|\cdots\|P_{j-1}\|R_{jl}\|P_{j+1}\|\cdots\|P_n,\delta_l\rangle.$$
Here, $\sigma_k$ is obtained from applying certain operators in $P_i$ to $\rho$, and $\delta_l$ is obtained from applying some operators in $P_j$ to $\rho$. Now we can make the transition of the $j$th component in $\mathcal{B}_1$ and the transition of the $i$th component in $\mathcal{B}_2$. After that, an element in $\mathcal{B}_1$ becomes some configuration(s) of the form \begin{equation}\label{ele-1}\langle P_1\|\cdots\|P_{i-1}\|Q_{ik}\|P_{i+1}\|\cdots\|P_{j-1}\|R_{jl}\|P_{j+1}\|\cdots\|P_n,\theta_{kl}\rangle\end{equation} where $\theta_{kl}$ is obtained from applying the operators that generated $\delta_l$ to $\sigma_k$, and an element in $\mathcal{A}_2$ becomes some configuration(s) of the form
\begin{equation}\label{ele-2}\langle P_1\|\cdots\|P_{i-1}\|Q_{ik}\|P_{i+1}\|\cdots\|P_{j-1}\|R_{jl}\|P_{j+1}\|\cdots\|P_n,\eta_{kl}\rangle\end{equation} where $\eta_{kl}$ is obtained from applying the operators that generated $\sigma_k$ to $\delta_l$. We use $\mathcal{B}^\prime_1, \mathcal{B}^\prime_2$ to denote the sets of elements given in equations (\ref{ele-1}), (\ref{ele-2}), respectively.
Note that $\mathit{var}\left(P_i\right)\cap\mathit{var}\left(P_j\right)=\emptyset$. So, an operator in $\mathcal{H}_{P_i}$ and an operator in $\mathcal{H}_{P_j}$ always commute. Therefore, $\mathcal{B}^\prime_1=\mathcal{B}^\prime_2$, and we complete the proof by setting $\mathcal{B}=\mathcal{B}^\prime_1(=\mathcal{B}^\prime_2)$.
\end{proof}

A confluence property follows from the above diamond property.

\begin{lem}\label{conf-ext} If $\langle P,\rho\rangle\rightarrow\mathcal{A}_1\rightarrow\cdots\rightarrow\mathcal{A}_k\rightarrow\cdots$ and $\langle P,\rho\rangle\rightarrow\mathcal{B}_1\rightarrow\cdots\rightarrow\mathcal{B}_k\rightarrow\cdots$, then there are $\mathcal{C}_1,\cdots,\mathcal{C}_k\cdots$ such that $$\langle P,\rho\rangle\rightarrow^\ast\mathcal{C}_1\rightarrow^\ast\cdots\rightarrow^\ast\mathcal{C}_k\rightarrow^\ast\cdots$$ and $\mathcal{A}_k\rightarrow^\ast\mathcal{C}_k$ and $\mathcal{B}_k\rightarrow^\ast\mathcal{C}_k$ for every $k$.
\end{lem}

\begin{proof} We proceed by induction on $k$ to find $\mathcal{C}_k$. For the case of $k=1$, it follows immediately from the diamond property (Lemma \ref{diamond}) that $\mathcal{C}_1$ exists. Now assume that we have
$$\langle P,\rho\rangle\rightarrow^\ast\mathcal{C}_1\rightarrow^\ast\cdots\rightarrow^\ast\mathcal{C}_k,$$ $\mathcal{A}_k\rightarrow^m\mathcal{C}_k$ and $\mathcal{B}_k\rightarrow^r\mathcal{C}_k$. Then repeatedly using the diamond property we can find $A_{k+1}^\prime$ such that $A_{k+1}\rightarrow^m\mathcal{A}_{k+1}^\prime$ and $C_k\rightarrow\mathcal{A}_{k+1}^\prime$. Similarly, we have $\mathcal{B}_{k+1}^\prime$ such that $C_k\rightarrow\mathcal{B}_{k+1}^\prime$ and $\mathcal{B}_{k+1}^\prime\rightarrow^r\mathcal{B}_{k+1}^\prime$. Consequently, we can use the diamond property once again and find $\mathcal{C}_{k+1}$ such that $\mathcal{A}_{k+1}^\prime\rightarrow\mathcal{C}_{k+1}$ and $\mathcal{B}_{k+1}^\prime\rightarrow\mathcal{C}_{k+1}$. Obviously, it holds that $\mathcal{C}_k\rightarrow^\ast\mathcal{C}_{k+1}$, $\mathcal{A}_{k+1}\rightarrow^\ast\mathcal{C}_{k+1}$ and $\mathcal{B}_{k+1}\rightarrow^\ast\mathcal{C}_{k+1}$.
\end{proof}

Now we are ready to prove Lemma \ref{lem-det} by refutation. 

\begin{proof}[Proof of Lemma \ref{lem-det}]

Assume that $\llbracket P\rrbracket(\rho)$ has two different elements $\mathit{val}(\pi_1)\neq \mathit{val}(\pi_2)$. We consider the following two cases:

{\vskip 4pt}

\textbf{Case} 1. One of $\pi_1$ and $\pi_2$ is finite. Suppose that, say, $\pi_1$ is longer than $\pi_2$. Then $\pi_1$ is finite, and using Lemma \ref{conf-ext} we can show that $\pi_1$ can be extended. Thus, $\pi_1$ is not a computation, a contradiction.

{\vskip 4pt}

\textbf{Case} 2. Both $\pi_1=\langle P,\rho\rangle\rightarrow\mathcal{A}_1\rightarrow\cdots\rightarrow\mathcal{A}_k\rightarrow\cdots$ and $\pi_2=\langle P,\rho\rangle\rightarrow\mathcal{B}_1\rightarrow\cdots\rightarrow\mathcal{B}_k\rightarrow\cdots$ are infinite. Then by Lemma \ref{conf-ext} we have a computation: $$\pi=\langle P,\rho\rangle\rightarrow\mathcal{C}_1\rightarrow^\ast\cdots\rightarrow^\ast\mathcal{C}_k\rightarrow^\ast\cdots$$ such that $\mathcal{A}_k\rightarrow^\ast\mathcal{C}_k$ and $\mathcal{B}_k\rightarrow^\ast\mathcal{C}_k$ for every $k$. It follows that $\mathit{val}(\mathcal{A}_k)\leq \mathit{val}(\mathcal{C}_k)$ and $\mathit{val}(\mathcal{B}_k)\leq \mathit{val}(\mathcal{C}_k)$ for all $k$. Furthermore, we have: $$\mathit{val}(\pi_1)=\lim_{k\rightarrow\infty}\mathit{val}(\mathcal{A}_k)\leq\lim_{k\rightarrow\infty}\mathit{val}(\mathcal{C}_k)=\mathit{val}(\pi)$$ and $\mathit{val}(\pi_2)\leq\mathit{val}(\pi)$. Since $\mathit{val}(\pi_1)\neq \mathit{val}(\pi_2)$, we have either $\mathit{val}(\pi_1)<\mathit{val}(\pi)$ or $\mathit{val}(\pi_2)< \mathit{val}(\pi)$. This contradicts to the assumption that both $\mathit{val}(\pi_1)$ and $\mathit{val}(\pi_2)$ are maximal elements of $\mathcal{V}(P,\rho)$ in Definition \ref{dp-den-sem}.
\end{proof}

\section{Proof of Lemma 3.3}\label{proof-lem-seq}
\begin{proof} The proof of Lemma \ref{lem-seq} is carried out in the following two steps:  
\begin{enumerate}\item For any input $\rho$, by Definition \ref{disjoint-op-semantics} we see that $\pi$ is a computation of $P_1\|\cdots\|P_n$ starting in $\rho$ if and only if it is a computation of $P_{i_1}\|\cdots\|P_{i_n}$ starting in $\rho$. Then by Definition \ref{dp-den-sem} it follows that $$\llbracket P_1\|\cdots\|P_n\rrbracket(\rho)=\llbracket P_{i_1}\|\cdots\|P_{i_n}\rrbracket(\rho).$$

{\vskip 4pt}

\item For each computation $$\pi=\langle P_1;\cdots;P_n,\rho\rangle\rightarrow\mathcal{A}_1\rightarrow\cdots\rightarrow\mathcal{A}_k\rightarrow\cdots$$ of $P_1;\cdots;P_n$ starting in $\rho$, we note that each configuration in $\mathcal{A}_k$ must be of the form $\langle Q;P_l;\cdots;P_n,\sigma\rangle$ for some $1\leq l\leq n$. Then we can replace it by the configuration $$\langle\downarrow\|\cdots\|\downarrow\|Q\|P_l\|\cdots\|P_n,\sigma\rangle$$ and thus obtain configuration ensemble $\mathcal{A}_k^\prime$. It is easy to see that $$\pi^\prime=\langle P_1\|\cdots\|P_n,\rho\rangle\rightarrow\cdots\mathcal{A}_1^\prime\rightarrow\cdots\rightarrow\mathcal{A}_k^\prime\rightarrow\cdots$$ is a computation of $P_1\|\cdots\|P_n$ and $\mathit{val}(\pi^\prime)=\mathit{val}(\pi).$ Therefore, we have: $$\mathcal{V}(P_1;\cdots;P_n,\rho)\subseteq\mathcal{V}(P_1\|\cdots\|P_n,\rho).$$
By Lemma \ref{lem-det} (Determinism) we know that $\llbracket P_1\|\cdots\|P_n\rrbracket(\rho)$ is the greatest element of $(\mathcal{V}(P_1\|\cdots\|P_n,\rho),\sqsubseteq)$.
Then it follows from Definition \ref{dp-den-sem} that $$\llbracket P_1;\cdots;P_n\rrbracket(\rho)\sqsubseteq\llbracket P_1\|\cdots\|P_n\rrbracket(\rho).$$

To prove the reverse of the above inequality, we note that for any quantum \textbf{while}-program $P$ and any computation $\langle P,\rho\rangle\rightarrow\mathcal{B}_1\rightarrow\cdots\rightarrow\mathcal{B}_k\rightarrow\cdots$, sequence $\{\mathit{val}(\mathcal{B}_k\}$ is increasing with respect to $\sqsubseteq$. Then for any other quantum \textbf{while}-program $P^\prime$, we have: \begin{equation}\label{seq-com-inequality}\llbracket P\rrbracket(\rho)\sqsubseteq \llbracket P;P^\prime\rrbracket(\rho).\end{equation} Furthermore, we can prove: $$\llbracket P_1;\cdots;P_n\rrbracket(\rho)\sqsubseteq \llbracket (P_1;P_1^\prime);\cdots;(P_n;P_n^\prime)\rrbracket(\rho)$$ by induction on $n$. Indeed, the induction hypothesis for $n-1$ implies: \begin{align*}
\llbracket P_1;\cdots;P_n\rrbracket(\rho)&=\llbracket P_n\rrbracket(\llbracket P_1;\cdots;P_{n-1}\rrbracket(\rho))\\ &\sqsubseteq \llbracket P_n\rrbracket(\llbracket (P_1;P_1^\prime);\cdots;(P_{n-1};P_{n-1}^\prime)\rrbracket(\rho))\\
&\sqsubseteq \llbracket P_n;P_n^\prime\rrbracket(\llbracket (P_1;P_1^\prime);\cdots;(P_{n-1};P_{n-1}^\prime)\rrbracket(\rho))\\ &=\llbracket (P_1;P_1^\prime);\cdots;(P_{n};P_{n}^\prime)\rrbracket(\rho).
\end{align*} Now, for any computation $$\pi=\langle P_1\|\cdots\|P_n,\rho\rangle\rightarrow\mathcal{A}_1\rightarrow\cdots\rightarrow\mathcal{A}_k\rightarrow\cdots$$ of $P_1\|\cdots\|P_n$, and for any $k\geq 1$, we observe that $\mathcal{A}_k$ is obtained  from $\langle P_1\|\cdots\|P_n,\rho\rangle$ by a finite number of transitions, each of which is performed by one of $P_1,\cdots,P_n$. Since $P_1,\cdots,P_n$ are disjoint, an operator in $\mathcal{H}_{P_i}$ always commutes with any operator in $\mathcal{H}_{P_j}$ provided $i\neq j$. Thus, these transitions can be re-ordered in order to satisfy the following requirements:

{\vskip 4pt}

\begin{enumerate}\item $P_i$ is (semantically) equivalent to $Q_i;P_i^\prime$ for every $i=1,...,n$;
\item the first group of transitions are done by $Q_1$, the second by $Q_2$, and so on.
\end{enumerate}

{\vskip 4pt}

Using inequality (\ref{seq-com-inequality}), we obtain: \begin{align*}
\mathit{val}(\mathcal{A}_k)&\sqsubseteq\llbracket Q_1;\cdots;Q_n\rrbracket(\rho)\\ &\sqsubseteq\left\llbracket \left(Q_1;P_1^\prime\right);\cdots;\left(Q_n;P_n^\prime\right)\right\rrbracket(\rho)= \llbracket P_1;\cdots;P_n\rrbracket(\rho)
\end{align*} and it follows that $$\mathit{val}(\pi)=\lim_{k\rightarrow\infty}\mathit{val}(\mathcal{A}_k)\sqsubseteq \llbracket P_1;\cdots;P_n\rrbracket(\rho).$$
Therefore, $\llbracket P_1;\cdots;P_n\rrbracket(\rho)$ is an upper bound of $\mathcal{V}(P_1\|\cdots\|P_n,\rho)$, and $$\llbracket P_1\|\cdots\|P_n\rrbracket(\rho)\sqsubseteq\llbracket P_1;\cdots;P_n\rrbracket(\rho).$$
\end{enumerate}
\end{proof}

\section{Proof of Lemma 4.2}\label{proof-TP-sound}

\begin{proof} For each $i$, since $$\mathit{var}(P_i)\cap\left(\bigcup_{j\neq i}\mathit{var}(P_i)\right)=\emptyset,$$ super-operator $\llbracket P_i\rrbracket$ can be written in the following Kraus form: $$\llbracket P_i\rrbracket(\rho)=\sum_{k}\left(E_{ik}\otimes I_{\overline{i}}\right)\rho \left(E_{ik}^\dag\otimes I_{\overline{i}}\right)$$ for any $\rho\in\bigotimes_{i=1}^n\mathcal{H}_{P_i}$, where $I_{\overline{i}}$ is the identity operator in $\bigotimes_{j\neq i}\mathcal{H}_j$. Then by Proposition 3.3.1(IV) in \cite{Ying16} and the Sequentialisation Lemma we obtain:\begin{align*}\llbracket P_1\|\cdots\|P_n\rrbracket(\rho)&=\llbracket P_1;\cdots;P_n\rrbracket(\rho)\\
&=\llbracket P_n\rrbracket\left(\cdots\llbracket P_2\rrbracket\left(\llbracket P_1\rrbracket\left(\rho\right)\right)\cdots\right)\\
&=\sum_{k_1,...,k_n}\left(\bigotimes_{i=1}^nE_{ik_i}\right)\rho\left(\bigotimes_{i=1}^nE_{ik_i}^\dag\right).
\end{align*} Consequently, it holds that for any observable $B_i$ in $\mathcal{H}_{P_i}$ $(i=1,...,n)$,
\begin{equation}\label{comp-dual}\begin{split}\llbracket P_1 \|\cdots\|P_n\rrbracket^\ast\left(\bigotimes_{i=1}^nB_i\right) &=\sum_{k_1,...,k_n}\left(\bigotimes_{i=1}^nE^\dag_{ik_i}\right)\left(\bigotimes_{i=1}^nB_{i}\right)\left(\bigotimes_{i=1}^nE_{ik_i}\right)\\ &=\sum_{k_1,...,k_n}\bigotimes_{i=1}^n\left(E^\dag_{ik_i}B_{i}E_{ik_i}\right)\\
&=\bigotimes_{i=1}^n\left(\sum_{k_i}E^\dag_{ik_i}B_{i}E_{ik_i}\right)\\
&=\bigotimes_{i=1}^n\llbracket P_i\rrbracket^\ast(B_{i}).
\end{split}\end{equation}

Now assume that $\models_\mathit{par}\left\{A_i\right\}P_i\left\{B_i\right\}$ $(i=1,...,n)$. Then with Lemma \ref{lem-correctness} we have: $$A_i\sqsubseteq \llbracket P_i\rrbracket^\ast\left(B_i\right)+\left(I_i- \llbracket P_i\rrbracket^\ast \left(I_i\right)\right)\ (i=1,...,n)$$ where $I_i$ is the identity operator in $\mathcal{H}_{P_i}$. To simplify the presentation, we write $F_i$ for $\llbracket P_i\rrbracket^\ast\left(I_i\right)$. Note that $B_i\sqsubseteq I_i$ and $$\llbracket P_i\rrbracket^\ast\left(B_i\right)\sqsubseteq\llbracket P_i\rrbracket^\ast\left(I_i\right)=F_i.$$ Then using Lemma \ref{lem-tensor}, we obtain:\begin{align*}
\bigotimes_{i=1}^nA_i &\sqsubseteq\bigotimes_{i=1}^n\left[\llbracket P_i\rrbracket^\ast\left(B_i\right)+\left(I_i- \llbracket P_i\rrbracket^\ast \left(I_i\right)\right)\right]
\\ &=\bigotimes_{i=1}^n\llbracket P_i\rrbracket^\ast\left(B_i\right) +\sum_{1\leq i_1<\cdots <i_k\leq n\ (k\geq 1)} \left[\left(\bigotimes_{i\notin\{i_1,...,i_k\}}\llbracket P_i\rrbracket^\ast\left(B_i\right)\right)\otimes\left(\bigotimes_{i\in\{i_1,...,i_k\}}\left(I_i-F_i\right)\right)\right]\\
&\sqsubseteq\bigotimes_{i=1}^n\llbracket P_i\rrbracket^\ast\left(B_i\right) +\sum_{1\leq i_1<\cdots <i_k\leq n\ (k\geq 1)}\ \left[\left(\bigotimes_{i\notin\{i_1,...,i_k\}}F_i\right)\otimes\left(\bigotimes_{i\in\{i_1,...,i_k\}}\left(I_i-F_i\right)\right)\right]\\
&=\bigotimes_{i=1}^n\llbracket P_i\rrbracket^\ast\left(B_i\right) +\left(\bigotimes_{i=1}^nI_i-\bigotimes_{i=1}^nF_i\right).
\end{align*} It follows from equation (\ref{comp-dual}) that $$\bigotimes_{i=1}^n F_i=\llbracket P_1\|\cdots\| P_n\rrbracket^\ast\left(\bigotimes_{i=1}^nI_i\right)$$ and thus $$\bigotimes_{i=1}^nA_i\sqsubseteq\llbracket P_1\|\cdots\|P_n\rrbracket^\ast\left(\bigotimes_{i=1}^nB_i\right)+\left[I-\llbracket P_1\|\cdots\|P_n\rrbracket^\ast(I)\right].$$ Therefore, with Lemma \ref{lem-correctness} we assert that $$\models_\mathit{par}\left\{\bigotimes_{i=1}^nA_i\right\}P_1\|\cdots\|P_n\left\{\bigotimes_{i=1}^n\right\}.$$
\end{proof}
 
\section{Proof of Lemma 4.3}\label{proof-S2E-sound}
\begin{proof} The proof of Lemma \ref{S2E-sound} is carried out in the following two steps: 
\begin{enumerate}\item We first consider the total correctness. Assume that $$\models_\mathit{tot}\left\{(1-\epsilon)I+\epsilon A\right\}P\left\{(1-\epsilon)I+\epsilon B\right\}.$$ Then by Lemma \ref{lem-correctness} we obtain:
\begin{equation}\label{S2E-proof-1}\begin{split}(1-\epsilon)I+\epsilon A&\sqsubseteq\llbracket P\rrbracket^\ast((1-\epsilon)I+\epsilon P)\\ &= (1-\epsilon)\llbracket P\rrbracket^\ast(I)+\epsilon \llbracket P\rrbracket^\ast(B)\end{split}\end{equation} because $\llbracket P\rrbracket$ is linear.
Note that for any super-operator $\mathcal{E}=\sum_i E_i\circ E_i^\dag$, it holds that $$\mathcal{E}^\ast(I)=\sum_i E_i^\dag E_i\sqsubseteq I.$$
Therefore, we have: $$(1-\epsilon)\llbracket P\rrbracket^\ast(I)\sqsubseteq (1-\epsilon)I$$ since $\epsilon\leq 1$. Consequently,
it follows from equation (\ref{S2E-proof-1}) that $\epsilon A\sqsubseteq\epsilon \llbracket P\rrbracket^\ast(B)$ and $A\sqsubseteq\llbracket P\rrbracket^\ast(B)$ because $\epsilon>0$. So, we otain $\models_\mathit{tot}\{A\}P\{B\}$.

{\vskip 4pt}

\item Now we consider the partial correctness. Let $$\models_\mathit{par}\left\{(1-\epsilon)I+\epsilon A\right\}P\left\{(1-\epsilon)I+\epsilon B\right\}.$$ Then it follows from Lemma \ref{lem-correctness} that
\begin{equation*}\begin{split}(1-\epsilon)I+\epsilon A &\sqsubseteq\llbracket P\rrbracket^\ast((1-\epsilon)I+\epsilon P)+\left[I-\llbracket P\rrbracket^\ast(I)\right]\\ &= (1-\epsilon)\llbracket P\rrbracket^\ast(I)+\epsilon \llbracket P\rrbracket^\ast(B)+\left[I-\llbracket P\rrbracket^\ast(I)\right].\end{split}\end{equation*} Since $\epsilon>0$, a routine calculation yields: $$A\sqsubseteq\llbracket P\rrbracket^\ast(B)+\left[I-\llbracket P\rrbracket^\ast(I)\right],$$ and thus $\models_\mathit{par}\{A\}P\{B\}$.
\end{enumerate}
\end{proof}

\section{Proof of Theorem 4.2}\label{proof-d-complete}

Assume that $$\models_\mathit{tot}\{A\}P_1\|\cdots\|P_n\{B\}.$$ We write $\mathcal{E}=\llbracket P_1\|\cdots\|P_n\rrbracket$ for the semantic function of parallel program and for each $i=1,...,n$, let $\mathcal{E}_i$ be the semantic function of $P_i$. Then by Lemma \ref{lem-correctness} we have $A\sqsubseteq\mathcal{E}^\ast(B)$, and by rule (R.Or) it suffices to show that
\begin{equation}\label{complete-dis}
\vdash_{qTP}\{\mathcal{E}^\ast(B)\}P_1\|\cdots\| P_n\{B\}.
\end{equation}

In what follows, we prove the theorem in three steps, gradually from a special form of $B$ to a general $B$. For each $i=1,...,n$, we use $\mathcal{H}_i=\mathcal{H}_{P_i}$ to denote the state Hilbert space of program $P_i$. We use $p$ to indicate the system of the parallel program $P_1\|\cdots\|P_n$, called the principal system. Thus, it has the state space $\mathcal{H}_p=\bigotimes_{i=1}^n\mathcal{H}_{i}$.

Let start from the very special case of $B=|\beta\rangle\langle\beta|$ with some constraints on $|\beta\rangle$. 

\begin{clm}
\label{tot-pure-finite}
For any vector $|\beta\>$ in $\mathcal{H}_p$, if its norm is less than or equal to $1$ and its reduced density operator to each $\mathcal{H}_i$ is of finite rank, then we have: 
$$\vdash_{qTP}\{\mathcal{E}^\ast(|\beta\>\<\beta|)\}P_1\|\cdots\| P_n\{|\beta\>\<\beta|\}.$$
\end{clm}

\begin{proof}[Proof of Claim \ref{tot-pure-finite}]
For each $i$, as the reduced density operator of $|\beta\>$ to $\mathcal{H}_i$ is of finite rank, we can use $\mathcal{K}_i$ to denote the support of the reduced density operator and assume its rank is $d_i$: 
$$\mathcal{K}_i = \supp(\tr_{1\cdots (i-1)(i+1)\cdots n}|\beta\>\<\beta|).$$
Obviously, $\mathcal{K}_i\subseteq\mathcal{H}_i$. We assume that $\Phi_{\mathcal{K}_i} = \{|1\>_i,\cdots,|d_i\>_i\}$ is an orthonormal basis of $\mathcal{K}_i$ and its expansion $\Phi_i = \{|1\>_i,\cdots,|d_i\>_i,|d_i+1\>_i,\cdots\}$ is an orthonormal basis of $\mathcal{H}_i$. Then $|\beta\>$ can be written as follows:
$$|\beta\rangle=\sum_{\forall i\in[n]: j_i\in[d_i]}\alpha_{j_1...j_n}\left(\bigotimes_{i=1}^n|j_i\rangle_i\right).$$ We define the conjugate vector of $|\beta\rangle$ as: $$|\overline{\beta}\rangle=\sum_{\forall i\in[n]: j_i\in[d_i]}\alpha^\ast_{j_1...j_n}\left(\bigotimes_{i=1}^n|j_i\rangle_i\right).$$
For each $i$, we further introduce an auxiliary system with the state Hilbert space $\mathcal{H}_{i^\prime}$ isomorphic to $\mathcal{H}_i$. Let $\{|j_i\>_{i^\prime}\}_{j\in[d_i]}$ and $\{|j_i\>_{i^\prime}\}_{j\in \Phi_i}$ be the orthonormal basis of $\mathcal{K}_{i^\prime}$ and $\mathcal{H}_{i^\prime}$ corresponding to $\Phi_{\mathcal{K}_i}$ and $\Phi_i$, respectively. Then
\begin{equation}\label{nota-ent-i}|\Psi_i\>=\sum_{j_i\in[d_i]}\frac{1}{\sqrt{d_i}}|j_i\>_i|j_i\>_{i^\prime}\end{equation}
is the maximally entangled state in $\mathcal{K}_i\otimes\mathcal{K}_{i^\prime}$.
We use $p^\prime$ to indicate the composed auxiliary system with state space $\mathcal{H}_{p^\prime}=\bigotimes_{i=1}^n\mathcal{H}_{i^\prime}$.
Then putting all of the entangled states together yields:
\begin{equation}\label{nota-E}E = \bigotimes_{i=1}^n (|\Psi_i\>\<\Psi_i|),\end{equation} which is a density operator in $\mathcal{H}_p\otimes\mathcal{H}_{p^\prime}$.

Now for each program $P_i$, completeness of qTD ensures that:
$$
\vdash_{qTD}\{(\mathcal{E}_i^\ast\otimes\mathcal{I}_{i^\prime})(|\Psi_i\>\<\Psi_i|)\}P_i\{|\Psi_i\>\<\Psi_i|\}.
$$
where $\mathcal{I}_{i^\prime}$ is the identity super-operator on $\mathcal{H}_{i^\prime}$.
Applying rule (R.PC.P), we obtain:
$$
\vdash_{qTD}\left\{\bigotimes_{i=1}^n(\mathcal{E}_i^\ast\otimes\mathcal{I}_{i^\prime})(|\Psi_i\>\<\Psi_i|)\right\}P_i\left\{\bigotimes_{i=1}^n|\Psi_i\>\<\Psi_i|\right\},
$$
or simply, 
\begin{equation}\label{rpcs1}\vdash_\mathit{qTP} \{D\}P_1\|\cdots\|P_n\{E\}\end{equation}
where:
\begin{equation}\label{nota-D}
D = \bigotimes_{i=1}^n \big[\left(\mathcal{E}_i^\ast\otimes \mathcal{I}_{i^\prime}\right)\left(|\Psi_i\>\<\Psi_i|\right)\big].
\end{equation}
 
We further define a super-operator $\mathcal{F}_\beta$ on $\mathcal{H}_{p^\prime}$ as follows:
\begin{equation}\label{nota-beta}
\mathcal{F}_\beta(\rho) = \sum_{\forall i\in[n]: k_i\in\mathcal{J}_i}|\bar{\beta}\>_{p^\prime}\left(\bigotimes_{i=1}^n{}_{i^\prime}\<k_i|\right) \rho \left(\bigotimes_{i=1}^n|k_i\>_{i^\prime}\right){}_{p^\prime}\<\bar{\beta}| \\
\end{equation} for every density operator $\rho$ in $\mathcal{H}_{p^\prime}$. It is easy to see that $\mathcal{F}_\beta$ is well-defined and is completely positive and trace non-increasing. Moreover, we observe:

\begin{fact}\label{big-entangle-1}
\begin{align*}
&(\mathcal{I}_p\otimes\mathcal{F}_\beta^\ast) (E)= \frac{1}{\prod_{i}d_i}|\beta\>_p\<\beta|\otimes I_{p^\prime} \\
&\left(\mathcal{I}_p\otimes \mathcal{F}_\beta^\ast \right)(D)=\frac{1}{\prod_{i}d_i}\left(\bigotimes_{i=1}^n\mathcal{E}_i^\ast\right)(|\beta\>_p\<\beta|)\otimes I_{p^\prime}.
\end{align*} where $\mathcal{I}_p$ is the identity super-operator on $\mathcal{H}_{p}$ and $I_{p^\prime}$ the identity operator on $\mathcal{H}_{p^\prime}$.
\end{fact}

\begin{proof}[Proof of Fact \ref{big-entangle-1}]
We directly compute: \begin{align*}
(\mathcal{I}_p\otimes\mathcal{F}_\beta^\ast) (E) &= \sum_{\forall i\in[n]: k_i\in\mathcal{J}_i}
\left(\bigotimes_{i=1}^n|k_i\>_{i^\prime}\right){}_{p^\prime}\<\bar{\beta}| \left[\bigotimes_{i=1}^n (|\Psi_i\>\<\Psi_i|)\right]
|\bar{\beta}\>_{p^\prime}\left(\bigotimes_{i=1}^n{}_{i^\prime}\<k_i|\right) \\
&= \sum_{\substack{\forall\ i\in[n]:\\k_i\in\mathcal{J}_i,\ l_i, l_i^\prime, j_i, j_i^\prime\in[d_i]}} \alpha_{l_1,\cdots,l_n} \alpha^\ast_{l_1^\prime,\cdots,l_n^\prime}
\left(\bigotimes_{i=1}^n|k_i\>_{i^\prime}\right)  \left(\bigotimes_{i=1}^n{}_{i^\prime}\<l_i|\right)\times\\
&\qquad \qquad\qquad \qquad\qquad \frac{1}{\prod_{i}d_i}\left[\bigotimes_{i=1}^n (|j_i\>_i\<j_i^\prime|\otimes|j_i\>_{i^\prime}\<j_i^\prime|)\right]
\left(\bigotimes_{i=1}^n|l_i^\prime\>_{i^\prime}\right)\left(\bigotimes_{i=1}^n{}_{i^\prime}\<k_i|\right) \\
&= \sum_{\substack{\forall\ i\in[n]:\\k_i\in\mathcal{J}_i,\ l_i, l_i^\prime\in[d_i]}} \alpha_{l_1,\cdots,l_n} \alpha^\ast_{l_1^\prime,\cdots,l_n^\prime}\frac{1}{\prod_{i}d_i}
\left(\bigotimes_{i=1}^n|k_i\>_{i^\prime}\right)
\left[\bigotimes_{i=1}^n (|l_i\>_i\<l_i^\prime|)\right]
\left(\bigotimes_{i=1}^n{}_{i^\prime}\<k_i|\right) \\
&= \left\{\frac{1}{\prod_{i}d_i}\sum_{\substack{\forall\ i\in[n]:\\l_i, l_i^\prime\in[d_i]}}\alpha_{l_1,\cdots,l_n} \alpha^\ast_{l_1^\prime,\cdots,l_n^\prime}\left[\bigotimes_{i=1}^n (|l_i\>_i\<l_i^\prime|)\right]\right\}\\ &\qquad\qquad \qquad \qquad\otimes \left\{\sum_{\substack{\forall\ i\in[n]:\\k_i\in\mathcal{J}_i}} \left(\bigotimes_{i=1}^n|k_i\>_{i^\prime}\right)\left(\bigotimes_{i=1}^n{}_{i^\prime}\<k_i|\right)   \right\} \\
&= \frac{1}{\prod_{i}d_i}|\beta\>_p\<\beta|\otimes I_{p^\prime}.
\end{align*}

\begin{align*}
(\mathcal{I}_p\otimes\mathcal{F}_\beta^\ast) (D) &= \sum_{\forall i\in[n]: k_i\in\mathcal{J}_i}
\left(\bigotimes_{i=1}^n|k_i\>_{i^\prime}\right){}_{p^\prime}\<\bar{\beta}| \left[\bigotimes_{i=1}^n (\mathcal{E}_i^\ast\otimes \mathcal{I}_{i^\prime})(|\Psi_i\>\<\Psi_i|)\right]
|\bar{\beta}\>_{p^\prime}\left(\bigotimes_{i=1}^n{}_{i^\prime}\<k_i|\right) \\
&= \sum_{\substack{\forall\ i\in[n]:\\k_i\in\mathcal{J}_i,\ l_i, l_i^\prime, j_i, j_i^\prime\in[d_i]}} \alpha_{l_1,\cdots,l_n} \alpha^\ast_{l_1^\prime,\cdots,l_n^\prime}
\left(\bigotimes_{i=1}^n|k_i\>_{i^\prime}\right)  \left(\bigotimes_{i=1}^n{}_{i^\prime}\<l_i|\right)\times\\ &\qquad\qquad\qquad
\left[\frac{1}{\prod_{i}d_i}\bigotimes_{i=1}^n (\mathcal{E}_i^\ast(|j_i\>_i\<j_i^\prime|)\otimes|j_i\>_{i^\prime}\<j_i^\prime|)\right]
\left(\bigotimes_{i=1}^n|l_i^\prime\>_{i^\prime}\right)\left(\bigotimes_{i=1}^n{}_{i^\prime}\<k_i|\right) \\
&= \sum_{\substack{\forall\ i\in[n]:\\k_i\in\mathcal{J}_i,\ l_i, l_i^\prime\in[d_i]}} \alpha_{l_1,\cdots,l_n} \alpha^\ast_{l_1^\prime,\cdots,l_n^\prime}
\left(\bigotimes_{i=1}^n|k_i\>_{i^\prime}\right)
\left[\frac{1}{\prod_{i}d_i}\bigotimes_{i=1}^n \mathcal{E}_i^\ast(|l_i\>_i\<l_i^\prime|)\right]
\left(\bigotimes_{i=1}^n{}_{i^\prime}\<k_i|\right) \\
&= \left\{\frac{1}{\prod_{i}d_i}\sum_{\substack{\forall\ i\in[n]:\\l_i, l_i^\prime\in[d_i]}}\alpha_{l_1,\cdots,l_n} \alpha^\ast_{l_1^\prime,\cdots,l_n^\prime}\left[\left(\bigotimes_{i=1}^n \mathcal{E}_i^\ast\right)\left(\bigotimes_{i=1}^n |l_i\>_i\<l_i^\prime|\right)\right]\right\}
\otimes \left\{\sum_{\substack{\forall\ i\in[n]:\\k_i\in\mathcal{J}_i}} \left(\bigotimes_{i=1}^n(|k_i\>_{i^\prime}\<k_i|)\right)\right\} \\
&= \frac{1}{\prod_{i}d_i}\left(\bigotimes_{i=1}^n \mathcal{E}_i^\ast\right)\left\{\sum_{\substack{\forall\ i\in[n]:\\l_i, l_i^\prime\in[d_i]}}\alpha_{l_1,\cdots,l_n} \alpha^\ast_{l_1^\prime,\cdots,l_n^\prime}\left[\bigotimes_{i=1}^n |l_i\>_i\<l_i^\prime|\right]\right\}
\otimes \left\{\sum_{\substack{\forall\ i\in[n]:\\k_i\in\mathcal{J}_i}} \left(\bigotimes_{i=1}^n(|k_i\>_{i^\prime}\<k_i|)\right)\right\} \\
&= \frac{1}{\prod_{i}d_i}\left(\bigotimes_{i=1}^n \mathcal{E}_i^\ast\right)(|\beta\>_p\<\beta|)\otimes I_{p^\prime}.
\end{align*}
\end{proof}

Now we can apply (R.SO) with completely positive and trace non-increasing super-operator $\mathcal{F}_\beta(\rho)$ on $p^\prime$ to (\ref{rpcs1}) and obtain: 
\begin{equation*}
\vdash_\mathit{qTP} \left\{\left(\mathcal{I}_p\otimes\mathcal{F}_\beta^\ast\right)(D)\right\} P_1\|\cdots\|P_n \left\{\left(\mathcal{I}_p\otimes\mathcal{F}_\beta^\ast\right) (E)\right\}.
\end{equation*}
or by Fact \ref{big-entangle-1} equivalently, 
\begin{equation}\label{q3}
\left\{\frac{1}{\prod_{i}d_i}\left(\bigotimes_{i=1}^n\mathcal{E}_i^\ast\right)(|\beta\>_p\<\beta|)\otimes I_{p^\prime}\right\}P_1\|\cdots\|P_n\left\{\frac{1}{\prod_{i}d_i}|\beta\>_p\<\beta|\otimes I_{p^\prime}\right\}.
\end{equation}
Therefore, applying rules (R.TI) and (R.Lin) to (\ref{q3}) yields:
$$
\left\{\left(\bigotimes_{i=1}^n\mathcal{E}_i^\ast\right)(|\beta\>_p\<\beta|)\right\}P_1\|\cdots\|P_n\left\{|\beta\>_p\<\beta|\right\}
$$
as we desired.
\end{proof}

Our next step is to generalise Claim \ref{tot-pure-finite} to the case of $B=|\beta\rangle\langle\beta$ with a general $|\beta\rangle$. 

\begin{clm}
\label{tot-pure}
For any pure state $|\beta\>$, we have: 
$$\vdash_{qTP}\{\mathcal{E}^\ast(|\beta\>\<\beta|)\}P_1\|\cdots\| P_n\{|\beta\>\<\beta|\}.$$
\end{clm}

\begin{proof}[Proof of Claim \ref{tot-pure}]
For each $i\in[n]$, assume $\Phi_i = \{|1\>_i,|2\>_i,\cdots\}$ is an orthonormal basis of $\mathcal{H}_i$.
We first define a sequence $\{P_k\}_{k\ge1}$ of projectors in $\hs_p$, the state Hilbert space of the whole program, as follows:
$$
P_k = P_{1k}\otimes P_{2k}\otimes\cdots\otimes P_{nk},
$$
where:
$$
\forall\ i\in[n]:\quad
\left\{
\begin{array}{ll}
 P_{ik} = I_i, & {\rm if\ } k > \dim \hs_i, \\
 P_{ik} = \sum_{j \le k} |j\>_i\<j|, &{\rm otherwise}.
\end{array}
\right.
$$
We further define: $$|\beta_k\> = P_k|\beta\>.$$ It is obvious that for all $i,k$, $P_{ik}$ has a finite rank, and therefore, the reduced density operator of $|\beta_k\>$ on each $\hs_i$ also has a finite rank and the norm of $|\beta_k\>$ is less than or equal to $1$.
Thus we can use Claim \ref{tot-pure-finite} to derive that 
\begin{equation}
\label{pure-fin2infin}
\vdash_{qTP} \{\mathcal{E}^\ast(|\beta_k\>\<\beta_k|)\}P_1\|\cdots\|P_n\{|\beta_k\>\<\beta_k|\}.
\end{equation}
Moreover, we observe:
\begin{fact}
\label{fact-pure}
\begin{equation}
|\beta_k\>\<\beta_k|\overset{w.o.t.}{\longrightarrow}|\beta\>\<\beta|,\qquad \mathcal{E}^\ast(|\beta_k\>\<\beta_k|)\overset{w.o.t.}{\longrightarrow}\mathcal{E}^\ast(|\beta\>\<\beta|).
\end{equation}
\end{fact}

\begin{proof}[Proof of Fact \ref{fact-pure}]
We first note that for each $i$, $P_{ik}\overset{s.o.t.}{\longrightarrow} I_i$ because for any $|\psi\>\in\hs_i$, $\lim_{n\rightarrow\infty}\|P_{ik}|\psi\>-I_i|\psi\>\|=0.$
Then by Theorem 1 in \cite{KV08} we have: $$P_{1k}\otimes P_{2k}\otimes\cdots\otimes P_{nk}\overset{s.o.t.}{\longrightarrow} I_1\otimes I_2\cdots\otimes I_n,$$ or equivalently,
$$P_k\overset{s.o.t.}{\longrightarrow} I.$$ According to the definition of $|\beta_k\>$, we see that 
$$
|\beta_k\>\overset{\|\cdot\|}{\longrightarrow} |\beta\>,
$$
and therefore,
$$
|\beta_k\>\<\beta_k|\overset{s.o.t.}{\longrightarrow}|\beta\>\<\beta|
\quad\Rightarrow\quad|\beta_k\>\<\beta_k|\overset{w.o.t.}{\longrightarrow}|\beta\>\<\beta|.
$$
To see this, we notice that for any $|\psi\>\in\hs$, it follows from the Cauchy-Schwarz inequality and Lemma \ref{fact-sot-wot} that  
\begin{align*}
\lim_{k\rightarrow\infty}\||\beta_k\>\<\beta_k|\psi\> - |\beta\>\<\beta|\psi\>\| &=
\lim_{k\rightarrow\infty}\||\beta_k\>(\<\beta_k|-\<\beta|)|\psi\> + (|\beta_k\>-|\beta\>)\<\beta|\psi\>\| \\
&\le \lim_{k\rightarrow\infty} \||\beta_k\>\|\||\psi\>\|\||\beta_k\>-|\beta\>\| +
|\<\beta|\psi\>|\||\beta_k\>-|\beta\>\| \\
&= 0
\end{align*}
Furthermore, using Lemma \ref{fact-wot-so} we obtain: 
$$\mathcal{E}^\ast(|\beta_k\>\<\beta_k|)\overset{w.o.t.}{\longrightarrow}\mathcal{E}^\ast(|\beta\>\<\beta|).$$
\end{proof}

Now we can apply rule (R.Lim) to equation (\ref{pure-fin2infin}) and then use Fact \ref{fact-pure} to derive: 
$$
\vdash_{qTP} \{\cE^\ast(|\beta\>\<\beta|)\}P_1\|\cdots\|P_n\{|\beta\>\<\beta|\}
$$
as we desired.
\end{proof}

Finally, we are able to deal with a general quantum predicate $B$. 

\begin{clm}
\label{tot-gen}
For any quantum predicate $B$, we have:
$$\vdash_{qTP}\{\mathcal{E}^\ast(B)\}P_1\|\cdots\| P_n\{B\}.$$
\end{clm}

\begin{proof}[Proof of Claim \ref{tot-gen}] 
For any quantum predicate $B$, we can always diagonalise it as follows: $$B = \sum_i\lambda_i|\beta_i\>\<\beta_i|$$ with $|\beta_i\rangle$ being a pure state and $0\leq\lambda_i\leq 1$ for all $i$ (spectral decomposition). Let us set: $$B_k = \sum_{i\le k}\lambda_i|\beta_i\>\<\beta_i|$$ for every $k\ge0$. Then with Lemma \ref{tot-pure}, we see that for each $i$,
\begin{equation}\label{before-cc}\vdash_{qTP}\{\mathcal{E}^\ast(|\beta_i\>\<\beta_i|)\}P_1\|\cdots\| P_n\{|\beta_i\>\<\beta_i|\}.\end{equation}
Applying rule (R.CC) to (\ref{before-cc}) yields:
$$\vdash_{qTP}\left\{\sum_{i\leq k}\lambda_i\mathcal{E}^\ast(|\beta_i\>\<\beta_i|)\right\}P_1\|\cdots\| P_n\left\{\sum_{i\leq k}\lambda_i|\beta_i\>\<\beta_i|\right\},$$
or simply, 
\begin{equation}
\label{tot-gen-pre}
\vdash_{qTP}\{A_k\}P_1\|\cdots\| P_n\{B_k\},
\end{equation}
where: $$A_k = \sum_{i\leq k}\lambda_i\mathcal{E}^\ast(|\beta_i\>\<\beta_i|) = \mathcal{E}^\ast(B_k).$$
Note that $\{B_k\}$ is an increasing sequence with respect to L\"owner order. So, by Fact \ref{fact-wot-cpo}, we have:
$$
B_k\overset{w.o.t.}{\longrightarrow}B.
$$ Furthermore, by 
Lemma \ref{fact-wot-so}, we obtain:
$$
A_k\overset{w.o.t.}{\longrightarrow}\mathcal{E}^\ast(B).
$$
Therefore, applying rule (R.Lim) to (\ref{tot-gen-pre}) yields
$$
\vdash_{qTP}\{\cE^\ast(B)\}P_1\|\cdots\| P_n\{B\}.
$$
\end{proof}

\section{Proof of Theorem 4.3}\label{proof-dp-complete}

We first prove soundness of qPP. It suffices to show that rule (R.PC.SP) is sound for partial correctness because soundness of the other rules in qPP have been proved before.
\begin{proof} For each $i$, it follows from assumptions $\models P_i:{\rm Abort}(C_i)$, $\models P_i:{\rm Term}(D_i)$ and $\models_\mathit{par}\left\{D_i+A_i\right\}P_i\left\{B_i\right\}$ that
$$C_i\sqsubseteq I_i-\llbracket P_i\rrbracket^\ast(I_i),\quad I_i-\llbracket P_i\rrbracket^\ast(I_i)\sqsubseteq D_i,\quad
D_i + A_i\sqsubseteq I_i-\llbracket P_i\rrbracket ^\ast(I_i)+\llbracket P_i\rrbracket^\ast(B_i).
$$
Consequently, we obtain: 
$$\llbracket P_i\rrbracket^\ast(I_i)\sqsubseteq I_i-C_i, \quad A_i\sqsubseteq\llbracket P_i\rrbracket^\ast(B_i).$$
Note that the semantic function for disjoint parallel program $P \equiv P_1\|\cdots\|P_n$ is $\llbracket P\rrbracket = \bigotimes_{i=1}^n\llbracket P_i\rrbracket$. Then it is straightforward to see that 
\begin{align*}
I - \bigotimes_{i=1}^n (I_i-C_i) + \bigotimes_{i=1}^nA_i&\sqsubseteq
I - \bigotimes_{i=1}^n\llbracket P_i\rrbracket^\ast(I_i) + \bigotimes_{i=1}^n\llbracket P_i\rrbracket^\ast(B_i)\\
&= I - \llbracket P\rrbracket^\ast(I) + \llbracket P\rrbracket^\ast\left(\bigotimes_{i=1}^nB_i\right)
\end{align*}
which actually means:
$$
\models_\mathit{par}\left\{I - \bigotimes_{i=1}^n (I_i-C_i) + \bigotimes_{i=1}^nA_i\right\} P_1\|\cdots\|P_n \left\{\bigotimes_{i=1}^nB_i\right\}.
$$\end{proof}

\begin{rem} In Remark \ref{aux-dis-p}, we pointed out that if nested parallelism is allowed, then rule (R.T.P) must be added in order to preserve the completeness of proof system qPP. 
The soundness of (R.T.P) is proved as follows. For each $i$, it follows from the assumption $\models P_i:{\rm Term}(A_i)$ that
$$I_i-\llbracket P_i\rrbracket^\ast(I_i)\sqsubseteq A_i, i.e.\ I_i-A_i\sqsubseteq\llbracket P_i\rrbracket^\ast(I_i).$$
Since $P \equiv P_1\|\cdots\|P_n$ is a disjoint parallel program, we have $\llbracket P\rrbracket = \bigotimes_{i=1}^n\llbracket P_i\rrbracket$ and 
\begin{align*}
I-\llbracket P\rrbracket^\ast(I) = I - \bigotimes_{i=1}^m\llbracket P_i\rrbracket^\ast(I_i) \sqsubseteq I - \bigotimes_{i=1}^m (I_i-A_i), 
\end{align*}
which means: 
$$\models P_1\|\cdots\|P_n: {\rm Term}\left(I - \bigotimes_{i=1}^m (I_i-A_i)\right).$$
\end{rem}

Now we turn to prove completeness of qPP. The idea is similar to the proof of Theorem \ref{thm-complete-dis}. So, we use the notations defined there. Assume that $$\models_\mathit{par}\{A\}P_1\|\cdots\|P_n\{B\}.$$ We write $\mathcal{E}=\llbracket P_1\|\cdots\|P_n\rrbracket$ for the semantic function of parallel program, and for each $i=1,...,n$, let $\mathcal{E}_i$ be the semantic function of program $P_i$. Then by Lemma \ref{lem-correctness} we have: $$A\sqsubseteq I-\mathcal{E}^\ast(I)+\mathcal{E}^\ast(B)$$ where $I$ is the identity operator on $\mathcal{H}_p$, the state Hilbert space of whole parallel program, and by rule (R.Or) it suffices to show that
\begin{equation}\label{complete-dis-par}
\vdash_{qPP}\{I-\mathcal{E}^\ast(I)+\mathcal{E}^\ast(B)\}P_1\|\cdots\| P_n\{B\}.
\end{equation}

We will complete the proof of (\ref{complete-dis-par}) in four steps. Our first step is to consider a special form of $B=|\beta\rangle\langle\beta|$ with certain constraint on $|\beta\rangle$. 

\begin{clm}
\label{par-pure-finite}
For any vector $|\beta\>$ such that its norm less than or equal to 1 and for each $i$, its reduced density operator on $\mathcal{H}_i$ is of finite rank, we have: 
$$\vdash_{qPP}\{I-\mathcal{E}^\ast(I)+\mathcal{E}^\ast(|\beta\>\<\beta|)\}P_1\|\cdots\| P_n\{|\beta\>\<\beta|\}.$$
\end{clm}

\begin{proof}[Proof of Claim \ref{par-pure-finite}]
We use the notations defined in the proof of Claim \ref{tot-pure-finite}.
For each program $P_i$, completeness of qPD ensures that
$$
\vdash_{qPD}\{I_i\otimes I_{i^\prime} - (\mathcal{E}_i^\ast\otimes\mathcal{I}_{i^\prime})(I_i\otimes I_{i^\prime}) + (\cE_i^\ast\otimes\cI_{i^\prime})(|\Psi_i\>\<\Psi_i|)\}P_i\{|\Psi_i\>\<\Psi_i|\}.
$$
where $\mathcal{I}_{i^\prime}$ is the identity super-operator on $\mathcal{H}_{i^\prime}$.
On the other hand, it is easy to check that
\begin{align*}
\label{par-abort-term}
&\models P_i:{\rm Abort}\left(I_i\otimes I_{i^\prime} - (\mathcal{E}_i^\ast\otimes\mathcal{I}_{i^\prime})(I_i\otimes I_{i^\prime})\right),\\
&\models P_i: {\rm Term} \left(I_i\otimes I_{i^\prime} - (\mathcal{E}_i^\ast\otimes\mathcal{I}_{i^\prime})(I_i\otimes I_{i^\prime})\right).
\end{align*}
Then we can use rule (R.PC.SP) to derive:
\begin{align*}
\vdash_{qPP}\left\{
I_p\otimes I_{p^\prime} - \bigotimes_{i=1}^n\left[I_i\otimes I_{i^\prime} - \left(I_i\otimes I_{i^\prime} - (\mathcal{E}_i^\ast\otimes\mathcal{I}_{i^\prime})(I_i\otimes I_{i^\prime})\right) \right]
+ \bigotimes_{i=1}^n(\mathcal{E}_i^\ast\otimes\mathcal{I}_{i^\prime})(|\Psi_i\>\<\Psi_i|)
\right\}& \\
 P_1\|\cdots\|P_n
\left\{
\bigotimes_{i=1}^n(|\Psi_i\>\<\Psi_i|)
\right\},&
\end{align*}
or simply, $$\vdash_{qPP}\{D^\prime\} P_1\|\cdots\|P_n \{E\}$$ with the notation: 
\begin{align*}
D^\prime &= I_{pp^\prime} - \bigotimes_{i=1}^n\left[ (\mathcal{E}_i^\ast\otimes\mathcal{I}_{i^\prime})(I_i\otimes I_{i^\prime}) \right] + D \\
&= I_p\otimes I_{p^\prime} - \mathcal{E}^\ast(I_p)\otimes I_{p^\prime} + D.
\end{align*}

\begin{fact}
\label{big-entangle-4}
$$
(\mathcal{I}_p\otimes\mathcal{F}_\beta^\ast)(D^\prime) = \left(I_p -  \mathcal{E}^\ast(I_p) + \frac{1}{\prod_{i=1}^nd_i}\mathcal{E}^\ast(|\beta\>_p\<\beta|)\right)\otimes I_{p^\prime}.
$$
\end{fact}
Notice that $\mathcal{F}_\beta^\ast(I_{p^\prime}) = I_{p^\prime}$. Then the above fact immediately follows from Fact \ref{big-entangle-1} and linearity.

Furthermore, applying rule (R.SO) with quantum operation $\mathcal{F}_{\beta}$, we obtain:
\begin{align*}
\vdash_{qPP}\left\{\left(\mathcal{I}_p\otimes\mathcal{F}_{\beta}^\ast\right)(D^\prime)\right\} P_1\|\cdots\|P_n \left\{\left(\mathcal{I}_p\otimes\mathcal{F}_{\beta}^\ast\right)(E^\prime)\right\},
\end{align*}
which, together with Facts \ref{big-entangle-1} and \ref{big-entangle-4} and rule (R.TI), implies: 
\begin{align*}
\vdash_{qPP}\left\{I_p - \mathcal{E}^\ast(I_p) + \frac{1}{\prod_{i=1}^nd_i}\mathcal{E}^\ast(|\beta\>_p\<\beta|)\right\} P_1\|\cdots\|P_n \left\{\frac{1}{\prod_{i=1}^nd_i}|\beta\>_p\<\beta|\right\}.
\end{align*}

For each $i$, it is obvious that
\begin{equation}\label{middle-term}
\models P_i: {\rm Term} \left(I_i - \mathcal{E}_i^\ast(I_i)\right).
\end{equation}
Thus, using rule (R.T.P) we obtain:
\begin{equation}
\label{par-term}
\vdash_{qPP} P_1\|\cdots\|P_n: {\rm Term} \left(I - \mathcal{E}^\ast(I)\right).
\end{equation}
Therefore, applying rule (R.CC2) with $p = \prod_{i=1}^nd_i$ we have:
\begin{align*}
\vdash_{qPP}\left\{I - \mathcal{E}^\ast(I) + \mathcal{E}^\ast(|\beta\>\<\beta|)\right\} P_1\|\cdots\|P_n \left\{|\beta\>\<\beta|\right\}
\end{align*}
as we desired.
\end{proof}

Our second step is to generalise the conclusion of Claim \ref{par-pure-finite} to the case of $B=|\beta\rangle\langle\beta|$ with a general $|\beta\rangle$.  

\begin{clm}
\label{par-pure}
For any pure state $|\beta\>$, we have: 
$$\vdash_{qPP}\{I-\mathcal{E}^\ast(I)+\mathcal{E}^\ast(|\beta\>\<\beta|)\}P_1\|\cdots\| P_n\{|\beta\>\<\beta|\}.$$
\end{clm}

\begin{proof}[Proof of Claim \ref{par-pure}]
We use the notations introduced in the proof of Claim \ref{tot-pure-finite}. First, we can use Claim \ref{par-pure-finite} to assert that 
\begin{equation}
\label{pure-fin2infin}
\vdash_{qPP} \{I - \mathcal{E}^\ast(I) + \cE^\ast(|\beta_k\>\<\beta_k|)\}P_1\|\cdots\|P_n\{|\beta_k\>\<\beta_k|\}.
\end{equation}
Moreover, as shown in the proof of Lemma \ref{tot-pure-finite}, it holds that $$|\beta_k\>\<\beta_k|\overset{w.o.t.}{\longrightarrow}|\beta\>\<\beta|.$$ Thus, with Lemma \ref{fact-wot-so}, it follows immediately that
\begin{align*}
I - \mathcal{E}^\ast(I) + \cE^\ast(|\beta_k\>\<\beta_k|)\overset{w.o.t.}{\longrightarrow}I - \mathcal{E}^\ast(I) + \cE^\ast(|\beta\>\<\beta|)
\end{align*}
Now, applying rule (R.Lim), we obtain:
$$
\vdash_{qPP} \{I - \mathcal{E}^\ast(I) + \cE^\ast(|\beta\>\<\beta|)\}P_1\|\cdots\|P_n\{|\beta\>\<\beta|\}
$$
as we desired.
\end{proof}

Next we further generalise the conclusion of Claim \ref{par-pure} to a more general $B$, using the spectral decomposition of $B$.

\begin{clm}
\label{par-gen-fin}
For any quantum predicate $B$ with finite rank, we have: 
$$\vdash_{qPP}\{I-\mathcal{E}^\ast(I)+\mathcal{E}^\ast(B)\}P_1\|\cdots\| P_n\{B\}.$$
\end{clm}

\begin{proof}[Proof of Claim \ref{par-gen-fin}]
We assume that postcondition $B$ has a finite rank $d_B<\infty$. Then $B$ can be diagonalised as follows:
$$B = \sum_{i=1}^{d_B}p_i|\beta_i\>\<\beta_i|$$
with $p_i\ge 0$ and $|\beta_i\>$ being pure states  
for all $i\in[d_B]$. Now for each $i$, according to Claim \ref{par-pure}, we have:
\begin{equation}
\label{par-gen-fin-middle}
\vdash_{qPP} \{I - \mathcal{E}^\ast(I) + \cE^\ast(|\beta_i\>\<\beta_i|)\}P_1\|\cdots\|P_n\{|\beta_i\>\<\beta_i|\}.
\end{equation}

We consider the following two cases:

$\bullet\ $Case 1: $\sum_{i=1}^{d_B}p_i\le1$. The completeness of qPD implies:
$$
\vdash_{qPD}\left\{I_i - \mathcal{E}_i^\ast(I_i)\right\} P_i \left\{ 0 \right\}, {\rm i.e.} \vdash_{qPD}P_i:{\rm Abort}(I_i - \mathcal{E}_i^\ast(I_i)).
$$
Then using rule (R.A.P), we obtain:
$$
\vdash_{qPP}P_1\|\cdots\|P_n:{\rm Abort}(I - \mathcal{E}^\ast(I)).
$$
Therefore, combining the above equation with equation (\ref{par-gen-fin-middle}) for all $i\in[d_B]$ and rule (R.CC1) yields:
\begin{align*}
\vdash_{qPP}\left\{
\sum_{i\in[d_B]}p_i\left(I-\mathcal{E}^\ast(I)+\mathcal{E}^\ast(|\beta_i\>\<\beta_i|)\right)
+ \left(1-\sum_{i\in[d_B]}p_i\right)(I - \mathcal{E}^\ast(I))
\right\}  P_1\|\cdots\|P_n \left\{
\sum_{i\in[d_B]}p_i|\beta_i\>\<\beta_i|
\right\},
\end{align*}
or equivalently,
\begin{align*}
\vdash_{qPP}\left\{I-\mathcal{E}^\ast(I)+\mathcal{E}^\ast(B)
\right\}  P_1\|\cdots\|P_n \left\{B\right\}.
\end{align*}

$\bullet\ $Case 2: $\sum_{i=1}^{d_B}p_i>1$. 
Combining equation (\ref{par-term}) with equation (\ref{par-gen-fin-middle}) for all $i\in[d_B]$ and rule (R.CC2), we have:
\begin{align*}
\vdash_{qPP}\left\{
\sum_{i\in[d_B]}p_i\left(I-\mathcal{E}^\ast(I)+\mathcal{E}^\ast(|\beta_i\>\<\beta_i|)\right)
- \left(\sum_{i\in[d_B]}p_i-1\right)(I - \mathcal{E}^\ast(I))
\right\}  P_1\|\cdots\|P_n & \\ \left\{
\sum_{i\in[d_B]}p_i|\beta_i\>\<\beta_i|
\right\}, 
\end{align*}
or equivalently,
\begin{align*}
\vdash_{qPP}\left\{I-\mathcal{E}^\ast(I)+\mathcal{E}^\ast(B)
\right\}  P_1\|\cdots\|P_n \left\{B\right\}.
\end{align*}
\end{proof}

Finally, we can complete the proof by showing the following:

\begin{clm}
\label{par-gen}
For any quantum predicate $B$, we have: 
$$\vdash_{qPP}\{I-\mathcal{E}^\ast(I)+\mathcal{E}^\ast(B)\}P_1\|\cdots\| P_n\{B\}.$$
\end{clm}

\begin{proof}[Proof of Claim \ref{par-gen}]
For any quantum predicate $B$, we can always diagonalise it as follows: $$B = \sum_i\lambda_i|\beta_i\>\<\beta_i|$$ with $0\leq\lambda_i\leq 1$ for all $i$ (spectral decomposition). Let us set $$B_k = \sum_{i\le k}\lambda_i|\beta_i\>\<\beta_i|$$ for each $k\ge0$. According to Claim \ref{par-gen-fin}, we know that for all $k$,
\begin{equation}
\label{par-gen-pre}
\vdash_{qPP}\{I-\mathcal{E}^\ast(I) + \mathcal{E}^\ast(B_k)\}P_1\|\cdots\| P_n\{B_k\}.
\end{equation}
On the other hand, $\{B_k\}$ is an increasing sequence with respect to L\"owner order. So, by Lemma \ref{fact-wot-cpo} we have:
$$
B_k\overset{w.o.t.}{\longrightarrow}B.
$$
Furthermore, we can use Lemma \ref{fact-wot-so} to deduce that 
$$
I-\mathcal{E}^\ast(I) + \mathcal{E}^\ast(B_k)\overset{w.o.t.}{\longrightarrow}I-\mathcal{E}^\ast(I) + \mathcal{E}^\ast(B).
$$
Therefore, applying rule (R.Lim) to equation (\ref{par-gen-pre}), we obtain:
$$
\vdash_{qTP}\{I-\mathcal{E}^\ast(I)+\cE^\ast(B)\}P_1\|\cdots\| P_n\{B\}.
$$
\end{proof}
\section{Proof of Theorem 6.1}\label{proof-strong-sound} 

\begin{proof} Suppose that $$\langle P,\rho\rangle\rightarrow^n\{|\langle P_i,\rho_i\rangle|\}.$$ We proceed by induction on the length $n$ of computation.

{\vskip 4pt}

$\blacktriangleright$  Induction basis: For $n=0$, $\{|\langle P_i,\rho_i\rangle|\}$ is a singleton $\{|\langle P_1,\rho_1\rangle|\}$ with $P_1\equiv P$ and $\rho_1\equiv \rho$. Then we can choose $T_1\equiv P$ and it holds that $P_1\equiv\mathit{at}(T_1,P)$. Note that in the proof outline $\{A\}P^\ast\{B\}$, we have $A\sqsubseteq\mathit{pre}(P)=B_1$. Thus, $$\mathit{tr}(A\rho)\leq\mathit{tr}(B_1\rho_1)=\sum_i\mathit{tr}(B_i\rho_i).$$

$\blacktriangleright$ Induction step: Now we assume that $$\langle P,\rho\rangle\rightarrow^{n-1}\mathcal{A}\rightarrow\mathcal{A}^\prime$$ and the conclusion is true for length $n-1$. Here, we only consider the simple case where the last step is derived by rule (MS1) with $\mathcal{A}= \{|\langle P_i,\rho_i\rangle|\}$ and $\mathcal{A}_\downarrow=\{\langle P,\rho\rangle\in\mathcal{A}: P\not\equiv\downarrow\}$ being a singleton $\left\{|\langle P_{i_0},\rho_{i_0}\rangle|\right\}$. (A general case with $\mathcal{A}_\downarrow$ having more than one element follows from the fact that rule (MS2) preserves the inequality in clause (2) of Theorem 6.1.) Then we can assume that $$\mathcal{A}^\prime =  \{|\langle P_i,\rho_i\rangle|i\neq i_0|\}\cup\{|\langle Q_j,\sigma_j\rangle|\}$$ where $\langle P_{i_0},\rho_{i_0}\rangle\rightarrow\{|\langle Q_j,\sigma_j\rangle|\}$ is derived by one of the rules used in Definition \ref{def-tran-ensemble} except (MS1) and (MS2). Thus, we need to consider the following cases:

{\vskip 4pt}

\textbf{Case} 1. The last step uses rule (IF$^\prime$). Then $P_{i_0}$ can be written in the following form: $$P_{i_0}\equiv\ \mathbf{if}\ \left(\square m\cdot M[\overline{q}]=m\rightarrow R_m\right)\ \mathbf{fi},$$ and for each $j$, $Q_j\equiv R_m\equiv \mathit{at}(R_m,P)$ and $\sigma_j=M_m\rho_{i_0}M_m^\dag$ for some $m$. On the other hand, a segment of the proof outline $\{A\}P^\ast\{B\}$ must be derived by the following inference:
$$\frac{\{A_m\}R_m^\ast\{C\}\ {\rm for\ every}\ m}{\left\{\sum_mM_m^\dag A_m M_m\right\}\ \mathbf{if}\ \left(\square m\cdot M[\overline{q}]=m\rightarrow\left\{A_m\right\} R_m^\ast\right)\ \mathbf{fi}\{C\}}$$ and $B_{i_0}=\mathit{pre}\left(P_{i_0}\right)\sqsubseteq\sum_m M_m^\dag A_mM_m$, $A_m=\mathit{pre}\left(R_m\right)$. Therefore, \begin{align*}\mathit{tr}\left(B_{i_0}\rho_{i_0}\right)&\leq\mathit{tr}\left(\sum_mM_m^\dag A_mM_m\rho_{i_0}\right)\\
&=\sum_m\mathit{tr}\left(M_m^\dag A_mM_m\rho_{i_0}\right)\\ &=\sum_m\mathit{tr}\left(A_mM_m\rho_{i_0}M_m^\dag\right)\\ &=\sum_j\mathit{tr}\left(\mathit{pre}\left(Q_j\right)\sigma_j\right).
\end{align*} By the induction hypothesis, we obtain: \begin{align*}\mathit{tr}(A\rho)&\leq\sum_{i\neq i_0}\mathit{tr}\left(B_i\rho_i\right)+\mathit{tr}\left(B_{i_0}\rho_{i_0}\right)\\
&\leq \sum_{i\neq i_0}\mathit{tr}\left(B_i\rho_i\right)+\sum_j\mathit{tr}\left(\mathit{pre}\left(Q_{j}\right)\sigma_{j}\right).
\end{align*} So, the conclusion is true in this case.

{\vskip 4pt}

\textbf{Case} 2. The last step uses rule (L$^\prime$). Then $P_{i_0}$ must be in the following form: $$P_{i_0}\equiv\ \mathbf{while}\ M[\overline{q}]=1\ \mathbf{do}\ R\ \mathbf{od}$$ and $\{|\langle Q_j,\sigma_j\rangle|\}=\{|\langle Q_0,\sigma_0\rangle,\langle Q_1,\sigma_1\rangle|\}$ with $Q_0\equiv\mathbf{skip}, \sigma_0=M_0\rho_{i_0}M_0^\dag, Q_1\equiv\ R;P_{i_0}$ and $\sigma_1=M_1\rho_{i_0}M_1^\dag$. A segment of $\{A\}P^\ast\{B\}$ must be derived by the following inference:$$\frac{\{D\}R^\ast\{M_0CM_0^\dag+M_1DM_1^\dag\}}{\begin{array}{ccc}\{M_0CM_0^\dag+M_1DM_1^\dag\}\ \mathbf{while}\ M[\overline{q}]=0\ \mathbf{do}\ \{C\}\ \mathbf{skip}\ \{C\}\\
\ \ \ \ \ \ \ \ \ \ \ \ \ \ \ \ \ \ \ \ \ \ \ \ \ \ \ \ \ \ \ \ \ \ \ \ \ \ \ \ \ \ \ \ \ \ \ \ \ \ \ \ \ \ \ \ \ \ \ \ \ \ \ \ \ \ \ \ \ \ =1\ \mathbf{do}\ R^\ast \{M_0CM_0^\dag+M_1DM_1^\dag\}\\ \mathbf{od}\ \{C\}
\end{array}}$$ and $B_{i_0}\sqsubseteq M_0CM_0^\dag +M_1DM_1^\dag.$ Then $Q_0\equiv\mathit{at}(\mathbf{skip},P), \mathit{pre}(Q_0)=C, Q_1\equiv\mathit{at}(R,P)$ and $\mathit{pre}(Q_1)=D$. It follows that \begin{align*}
\mathit{tr}\left(B_{i_0}\rho_{i_0}\right)&\leq\mathit{tr}\left[\left(M_0CM_0^\dag+M_1DM_1^\dag\right)\rho_{i_0}\right]\\ &=\mathit{tr}\left(M_0CM_0^\dag\rho_{i_0}\right) +\mathit{tr}\left(M_1DM_1^\dag\rho_{i_0}\right)\\
&=\mathit{tr}\left(CM_0^\dag\rho_{i_0}M_0\right)+\mathit{tr}\left(DM_1^\dag\rho_{i_0}M_1\right)\\
&= \mathit{tr}\left(\mathit{pre}(Q_0)\sigma_0\right)+\mathit{tr}\left(\mathit{pre}(Q_1)\sigma_1\right).
\end{align*} Furthermore, by the induction hypothesis, we have: \begin{align*}\mathit{tr}(A\rho)&\leq\sum_{i\neq i_0}\mathit{tr}\left(B_i\rho_i\right)+\mathit{tr}\left(B_{i_0}\rho_{i_0}\right)\\
&\leq \sum_{i\neq i_0}\mathit{tr}\left(B_i\rho_i\right)+\sum_j\mathit{tr}\left(\mathit{pre}\left(Q_{j}\right)\sigma_{j}\right).
\end{align*} Thus, the conclusion is true in this case.

{\vskip 4pt}

\textbf{Case} 3. The last step uses rule (Sk), (In) or (UT). Similar but easier.
\end{proof}

\section{Proof of Theorem 6.2}\label{proof-strong-sound-1} 

\begin{proof} We prove the conclusion by induction on the length $l$ of transition sequence: $$\langle P_1\|\cdots\| P_n,\rho\rangle\rightarrow^l \left\{|\langle P_{1s}\|\cdots\| P_{ns},\rho_s\rangle |\right\}.$$ The conclusion is obviously true in the induction basis case of $l=0$. Now we assume that $$\langle P_1\|\cdots\| P_n,\rho\rangle\rightarrow^l \left\{|\langle P_{1k}^\prime\|\cdots\| P_{nk}^\prime,\rho_k^\prime\rangle |\right\}\rightarrow\left\{|\langle P_{1s}\|\cdots\| P_{ns},\rho_s\rangle |\right\}$$ and the last step is a transition performed by the $r$th component $(1\leq r\leq n)$: \begin{equation}\label{r-comp}\langle P_{rk}^\prime,\rho_k^\prime\rangle\rightarrow\left\{|\langle Q_{rk}^{(h)},\rho_k^{(h)}\rangle|\right\}\end{equation}
for each $k$. Then \begin{equation}\label{set-equ}\left\{|\langle P_{1s}\|\cdots\|P_{ns},\rho_s\rangle|\right\}=\bigcup_k\left\{|\langle P_{1k}^\prime\|\cdots\|P_{(r-1)k}^\prime\|Q_{rk}^{(h)}\|P_{(r+1)k}^\prime\|\cdots\|P_{nk}^\prime,\rho_k^{(h)}\rangle|\right\}.\end{equation}
By the induction hypothesis for the first $l$ steps, we obtain: \begin{equation}\label{Ind-Hy}\mathit{tr}\left[\left(\sum_{i=1}^np_iA_i\right)\rho\right]\leq\sum_k\mathit{tr}\left[\left(\sum_{i=1}^np_iB_{ik}^\prime\right)\rho_k^\prime\right]\end{equation} where $$B_{ik}^\prime=\begin{cases}B_i\ &{\rm if}\ P_{ik}^\prime\equiv\ \downarrow,\\ \mathit{pre}(T_{ik}^\prime) &{\rm if}\ P_{ik}^\prime\equiv\mathit{at}(T_{ik}^\prime,P_i).\end{cases}$$
We set $$B_{rk}^{(h)}=\begin{cases}B_r\ &{\rm if}\ Q_{rk}^{(h)}\equiv\ \downarrow,\\ \mathit{pre}(S_{rk}^{(h)}) &{\rm if}\ Q_{rk}^{(h)}\equiv\mathit{at}(S_{rk}^{(h)},P_r).\end{cases}$$ Then for each $k$, by an argument similar to the case of Theorem \ref{thm.strong-sound} on transition (\ref{r-comp}), we can prove that \begin{equation}\label{r-comp-1}\mathit{tr}\left(B_{rk}^\prime\rho_k^\prime\right)\leq\sum_h\mathit{tr}\left(B_{rk}^{(h)}\rho_k^{(h)}\right).\end{equation}
On the other hand, $\{A_i\}P_i^\ast\{B_i\}$ $(i=1,...,n)$ are $\Lambda$-interference free. Then for every $i\neq r$, it follows from transition (\ref{r-comp}) that \begin{equation}\label{ind-inter}
\mathit{tr}\left[\left(\lambda_{ir}B_{ik}^\prime+\left(1-\lambda_{ik}\right)B_{rk}^\prime\right)\rho_k^\prime\right]\leq\sum_h\mathit{tr}\left[\left(\lambda_{ir}B_{ik}^\prime+\left(1-\lambda_{ir}\right)B_{rk}^{(h)}\right)\rho_k^{(h)}\right].
\end{equation} Note that $\sum_{i=1}^n p_i=1$, thus condition (\ref{lambda-con}) implies: $$p_r-\sum_{i\neq r}\frac{p_i\left(1-\lambda_{ir}\right)}{\lambda_{ir}}\geq 0.$$
Therefore,  we have:
\begin{align}\label{total-1}
&\mathit{tr}\left[\left(\sum_{i=1}^np_iA_i\right)\rho\right]\leq\sum_k\mathit{tr}\left[\left(\sum_{i=1}^np_iB_{ik}^\prime\right)\rho_k^\prime\right]\\
&=\sum_k\mathit{tr}\left[\left(\sum_{i\neq r}^np_iB_{ik}^\prime+p_rB_{rk}^\prime\right)\rho_k^\prime\right]\\
&=\sum_k\mathit{tr}\left\{\left[\sum_{i\neq r}^n\frac{p_i}{\lambda_{ir}}\left(\lambda_{ir}B_{ik}^\prime+\left(1-\lambda_{ir}\right)B_{rk}^\prime\right)+\left(p_r-\sum_{i\neq r}\frac{p_i\left(1-\lambda_{ir}\right)}{\lambda_{ir}}\right)B_{rk}^\prime\right]\rho_k^\prime\right\}\\
\label{total-2}&\leq \sum_k\left\{\sum_{i\neq r}^n\frac{p_i}{\lambda_{ir}}\sum_h\mathit{tr}\left[\left(\lambda_{ir}B_{ik}^\prime+\left(1-\lambda_{ir}\right)B_{rk}^{(h)}\right)\rho_k^{(h)}\right]+\left(p_r-\sum_{i\neq r}\frac{p_i\left(1-\lambda_{ir}\right)}{\lambda_{ir}}\right)\sum_h B_{rk}^{(h)}\rho_k^{(h)}\right\}\\
&=\sum_{k,h}\mathit{tr}\left[\left(\sum_{i\neq r}p_iB_{ik}^\prime+p_rB_{rk}^{(h)}\right)\rho_k^{(h)}\right]\\ \label{total-3}&=\sum_s\mathit{tr}\left[\left(\sum_{i=1}^np_iB_{is}\right)\rho_s\right].
\end{align}
Here, (\ref{total-1}) comes from equation (\ref{Ind-Hy}), the first and second part of (\ref{total-2}) from (\ref{ind-inter}), (\ref{r-comp-1}), respectively, and (\ref{total-3}) from (\ref{set-equ}).
\end{proof}

\section{Proof of Lemma 7.1}\label{super-lem-proof}

\begin{proof} 
The left four equations is obvious due to Fact \ref{big-entangle-1}. The rest proof is completed by straightforward calculations:

\begin{align*}
&\mathcal{F}^\ast\left(\bigotimes_{i\in[n]}\Psi^H_{i;i^\prime}\right) =\frac{1}{2^n}\sum_{\substack{\forall\ i\in[n]:\\k_i\in\{0,1\}}}
\left[\sum_{\substack{\forall\ i\in[n]:\\v_i,u_i\in\{0,1\}}}
\left(\bigotimes_{i\in[n]}|k_i\>_{i^\prime}\right){}_{p^\prime}\<\bar{\phi}| \bigotimes_{i\in[n]}\left(\frac{1}{\sqrt{2}}(-1)^{u_iv_i}|v_i\>_i|u_i\>_{i^\prime} \right)\right]\left[\cdot\right]^\dag \\
&=\frac{1}{2^n}\sum_{\substack{\forall\ i\in[n]:\\k_i\in\{0,1\}}}\left[\sum_{\substack{\forall\ i\in[n]:\\v_i,u_i,x_i,z_i\in\{0,1\}}}
\left(\bigotimes_{i\in[n]}|k_i\>_{i^\prime}\right)
\mathbf{i}^{\sum A_{i,j}x_ix_j}(-1)^{\sum z_ix_i} \bigotimes_{i\in[n]}\left(\frac{1}{2\sqrt{2}}{}_{i^\prime}\<z_i| (-1)^{u_iv_i}|v_i\>_i|u_i\>_{i^\prime} \right)
\right]\left[\cdot\right]^\dag \\
&=\frac{1}{2^n}\sum_{\substack{\forall\ i\in[n]:\\k_i\in\{0,1\}}}\left[\left(\bigotimes_{i\in[n]}|k_i\>_{i^\prime}\right) \otimes
\sum_{\substack{\forall\ i\in[n]:\\x_i\in\{0,1\}}}
\mathbf{i}^{\sum A_{i,j}x_ix_j} \bigotimes_{i\in[n]}\left( \frac{1}{2\sqrt{2}} \sum_{v_i,z_i\in\{0,1\}}(-1)^{z_i(x_i+v_i)}|v_i\>_i \right)
\right]\left[\cdot\right]^\dag \\
&=\frac{1}{2^n}\left[\frac{1}{\sqrt{2^n}}\sum_{\forall\ i\in[n]:x_i\in\{0,1\}}\mathbf{i}^{\sum A_{i,j}x_ix_j} \bigotimes_{i\in[n]}|x_i\>_i
\right]\left[\cdot\right]^\dag\otimes I_{p^\prime} \\
&=\frac{1}{2^n}|\phi_4\>_p\<\phi_4|\otimes I_{p^\prime},
\end{align*}

\begin{align*}
\mathcal{F}_S^\ast\left(\bigotimes_{i\in[n]}\Psi^S_{i;i^\prime}\right) &=\frac{1}{2^n}\sum_{\substack{\forall\ i\in[n]:\\k_i\in\{0,1\}}}
\left[\sum_{\substack{\forall\ i\in[n]:\\u_i\in\{0,1\}}}
\left(\bigotimes_{i\in[n]}|k_i\>_{i^\prime}\right){}_{p^\prime}\<\bar{\phi_0}| \bigotimes_{i\in[n]}\left(\sum_{u_i}\mathbf{i}^{-A_{i,i}u_i}|u_i\>_i|u_i\>_{i^\prime} \right)\right]\left[\cdot\right]^\dag \\
&=\frac{1}{2^n}\sum_{\substack{\forall\ i\in[n]:\\k_i\in\{0,1\}}}\left[\sum_{\substack{\forall\ i\in[n]:\\u_i,x_i\in\{0,1\}}}
\left(\bigotimes_{i\in[n]}|k_i\>_{i^\prime}\right)\frac{1}{\sqrt{2^n}}
 \bigotimes_{i\in[n]}\left(\mathbf{i}^{A_{i,i}x_i}{}_{i^\prime}\<x_i| \mathbf{i}^{-A_{i,i}u_i}|u_i\>_i|u_i\>_{i^\prime} \right)
\right]\left[\cdot\right]^\dag \\
&=\frac{1}{2^n}\sum_{\substack{\forall\ i\in[n]:\\k_i\in\{0,1\}}}\left[\left(\bigotimes_{i\in[n]}|k_i\>_{i^\prime}\right)\otimes\sum_{\substack{\forall\ i\in[n]:\\x_i\in\{0,1\}}}
\frac{1}{\sqrt{2^n}}\bigotimes_{i\in[n]}|x_i\>_i
\right]\left[\cdot\right]^\dag \\
&=\frac{1}{2^n}\left[\sum_{\substack{\forall\ i\in[n]: x_i\in\{0,1\}}}\frac{1}{\sqrt{2^n}}\bigotimes_{i\in[n]}|x_i\>_i
\right]\left[\cdot\right]^\dag\otimes I_{p^\prime} \\
&=\frac{1}{2^n}|\phi_S\>_p\<\phi_S|\otimes I_{p^\prime},
\end{align*}

\begin{align*}
\mathcal{F}^{\prime\ast}\left(\bigotimes_{i\in[n]}\Psi^H_{i;i^\prime}\right) &=\frac{1}{2^n}\sum_{\substack{\forall\ i\in[n]:\\k_i\in\{0,1\}}}
\left[\sum_{\substack{\forall\ i\in[n]:\\x_i,v_i,u_i\in\{0,1\}}}
\left(\bigotimes_{i\in[n]}|k_i\>_{i^\prime}\right)\frac{1}{\sqrt{2^n}} \bigotimes_{i\in[n]}\left(\frac{1}{\sqrt{2}}{}_{i^\prime}\<x_i|(-1)^{u_iv_i}|v_i\>_i|u_i\>_{i^\prime} \right)\right]\left[\cdot\right]^\dag \\
&=\frac{1}{2^n}\sum_{\substack{\forall\ i\in[n]:\\k_i\in\{0,1\}}}
\left[\left(\bigotimes_{i\in[n]}|k_i\>_{i^\prime}\right) \bigotimes_{i\in[n]}\left(\frac{1}{2}\sum_{x_i,v_i\in\{0,1\}}(-1)^{x_iv_i}|v_i\>_i \right)\right]\left[\cdot\right]^\dag \\
&=\frac{1}{2^n}
\left[\bigotimes_{i\in[n]}|0\>_i\right]\left[\cdot\right]^\dag\otimes I_{p^\prime} = \frac{1}{2^n}|0\>_p\<0|\otimes I_{p^\prime},
\end{align*}

 \begin{align*}
&\mathcal{F}_m^\ast\left(\bigotimes_{(i,j)\in S_m}\Psi_{i,j;i^\prime,j^\prime}\bigotimes_{i\in T_m}\Phi_{i;i^\prime}\right) \\
&=\frac{1}{2^n}\sum_{\substack{\forall\ i\in[n]:\\k_i\in\{0,1\}}}
\left[
\left(\bigotimes_{i\in[n]}|k_i\>_{i^\prime}\right){}_{p^\prime}\<\bar{\phi_m}| \bigotimes_{(i,j)\in S_m} \left(\sum_{\substack{u_i,u_j\\\in\{0,1\}}}(-1)^{A_{i,j}u_iu_j}|u_i\>_i|u_j\>_j|u_i\>_{i^\prime}|u_j\>_{j^\prime} \right)
\bigotimes_{i\in T_m}\left(\sum_{u_i\in\{0,1\}}|u_i\>_i|u_i\>_{i^\prime} \right)
\right]\left[\cdot\right]^\dag \\
&=\frac{1}{2^n}\sum_{\substack{\forall\ i\in[n]:\\k_i\in\{0,1\}}}
\left[\rule{0cm}{1cm}
\left(\bigotimes_{i\in[n]}|k_i\>_{i^\prime}\right)  \frac{1}{\sqrt{2^n}} \sum_{\substack{\forall\ i\in[n]:\\x_i\in\{0,1\}}} \ci^{-\sum A_{i,j}x_ix_j}(-1)^{\sum_{(i,j)\in\bigcup_{l>m}S_l}A_{i,j}x_ix_j}\bigotimes_{i\in[n]}{}_{i^\prime}\<x_i|\right.\\
&\left.\rule{0cm}{1cm}\qquad\qquad\qquad\bigotimes_{(i,j)\in S_m} \left(\sum_{u_i,u_j\in\{0,1\}}(-1)^{A_{i,j}u_iu_j}|u_i\>_i|u_j\>_j|u_i\>_{i^\prime}|u_j\>_{j^\prime} \right)
\bigotimes_{i\in T_m}\left(\sum_{u_i\in\{0,1\}}|u_i\>_i|u_i\>_{i^\prime} \right)
\right]\left[\cdot\right]^\dag \\
&= \frac{1}{2^n}\sum_{\substack{\forall\ i\in[n]:\\k_i\in\{0,1\}}}\left[
\left(\bigotimes_{i\in[n]}|k_i\>_{i^\prime}\right) \frac{1}{\sqrt{2^n}}  \sum_{\substack{\forall\ i\in[n]:\\x_i\in\{0,1\}}}
\ci^{-\sum A_{i,j}x_ix_j}(-1)^{\sum_{(i,j)\in\bigcup_{l>m}S_l}A_{i,j}x_ix_j}(-1)^{\sum_{(i,j)\in S_m}A_{i,j}x_ix_j}
\bigotimes_{i\in [n]}|x_i\>_i\right]\left[\cdot\right]^\dag\\
&=\frac{1}{2^n}
\left[\frac{1}{\sqrt{2^n}}\sum_{\forall\ i\in[n]:x_i\in\{0,1\}}\ci^{\sum A_{i,j}x_ix_j}(-1)^{\sum_{(i,j)\in \bigcup_{l\ge m}S_l}A_{i,j}x_ix_j} \bigotimes_{i\in[n]}|x_i\>_i
\right]\left[\cdot\right]^\dag\otimes I_{p^\prime} \\
&= \frac{1}{2^n}|\phi_{m-1}\>_p\<\phi_{m-1}|\otimes I_{p^\prime}. 
\end{align*}
\end{proof}

\end{document}